\documentclass[a4paper,11pt]{article} 
\usepackage[a4paper,hmargin=2.8cm,vmargin=3cm]{geometry}
\usepackage[utf8x]{inputenc} 
\usepackage[english]{babel}
\usepackage{amsmath,amsthm,amsfonts,amssymb,mathrsfs,bbm}
\usepackage{comment}
\usepackage{authblk}
\usepackage{lineno}
\usepackage[dvipsnames]{xcolor}

\definecolor{BeauBlue}{rgb}{0, 0.2, .9}
\definecolor{BeauOrange}{rgb}{.8, .1, 0}
\usepackage{hyperref}
\hypersetup{
	colorlinks = true,
	linkcolor = BeauBlue,
	urlcolor = cyan,
	citecolor = BeauOrange,
}

\numberwithin{equation}{section}

\usepackage{tikz}
\usetikzlibrary{calc,shapes,patterns,decorations.pathreplacing}

\makeatletter
\DeclareRobustCommand\widecheck[1]{{\mathpalette\@widecheck{#1}}}
\def\@widecheck#1#2{%
    \setbox\z@\hbox{\m@th$#1#2$}%
    \setbox\tw@\hbox{\m@th$#1%
       \widehat{%
          \vrule\@width\z@\@height\ht\z@
          \vrule\@height\z@\@width\wd\z@}$}%
    \dp\tw@-\ht\z@
    \@tempdima\ht\z@ \advance\@tempdima2\ht\tw@ \divide\@tempdima\thr@@
    \setbox\tw@\hbox{%
       \raise\@tempdima\hbox{\scalebox{1}[-1]{\lower\@tempdima\box
\tw@}}}%
    {\ooalign{\box\tw@ \cr \box\z@}}}
\makeatother

\makeatletter
\newcommand{\leqnos}{\tagsleft@true\let\veqno\@@leqno}
\newcommand{\reqnos}{\tagsleft@false\let\veqno\@@eqno}
\reqnos
\makeatother

\newcommand{\triple}[1]{{\left\vert\kern-0.25ex\left\vert\kern-0.25ex\left\vert #1 
    \right\vert\kern-0.25ex\right\vert\kern-0.25ex\right\vert}}

\newcommand{\tr}{\mathrm{tr}}

\newcommand{\veps}{\varepsilon}

\newcommand{\F}{\mathcal{F}}
\newcommand{\cN}{\mathcal{N}}
\def\tr{\mathrm{tr}}

\def\cH{\mathcal{H}}
\def\cN{\mathcal{N}}

\newtheorem{theorem}{Theorem}[section] % reset theorem numbering for each chapter
\newtheorem{proposition}[theorem]{Proposition} 

\newtheorem{lemma}[theorem]{Lemma}

\newtheorem{remark}[theorem]{Remark}
\newtheorem{assumption}[theorem]{Assumption}

\title{Dynamics of mean-field Fermi systems with nonzero pairing}

\author[1]{Stefano Marcantoni} 
\author[2]{Marcello Porta}
\author[3]{Julien Sabin}
\affil[1]{Universit\'e C\^ote d'Azur, Parc Valrose, 06108 Nice, France}
\affil[2]{SISSA, Via Bonomea 265, 34136 Trieste, Italy}
\affil[3]{Universit\'e de Rennes, Avenue du G\'en\'eral Leclerc 263, 35042 Rennes, France}

\begin{document}

\maketitle

\begin{abstract}
We study the dynamics of many-body Fermi systems, for a class of initial data which are close to quasi-free states exhibiting a nonvanishing pairing matrix. We focus on the mean-field scaling, which for fermionic systems is naturally coupled with a semiclassical scaling. Under the assumption that the initial datum enjoys a suitable semiclassical structure, we give a rigorous derivation of the time-dependent Hartree-Fock-Bogoliubov equation, a nonlinear effective evolution equation for the generalized one-particle density matrix of the system, as the number of particles goes to infinity. Our result holds for all macroscopic times, and provides bounds for the rate of convergence.
\end{abstract}

\tableofcontents     

\section{Introduction}

Let us consider a system of $N \gg 1$ quantum particles in $\mathbb{R}^{3}$, with $M$ internal degrees of freedom, or spins. At zero temperature, the state of the system is denoted by the wave function $\psi_{N} \in L^{2}((\mathbb{R}^{3} \times \{1, \ldots, M\})^{N})$, under the normalization condition
\begin{equation}
\| \psi_{N} \|_{2}^{2} := \sum_{\sigma_{1}, \ldots, \sigma_{N}} \int dx_{1} \ldots dx_{N}\, | \psi_{N}({\bf x}_{1}, \ldots, {\bf x}_{N}) |^{2} = 1\;,
\end{equation}
where ${\bf x} = (x, \sigma)$ with $x\in \mathbb{R}^{3}$ the position label and $\sigma \in \{ 1, \ldots, M \}$ the spin label.  We shall restrict the attention to fermionic particles, which correspond to antisymmetric wave functions:
\begin{equation}
\psi_{N}({\bf x}_{1}, \ldots {\bf x}_{N}) = \text{sgn}(\pi) \psi_{N}({\bf x}_{\pi(1)}, \ldots {\bf x}_{\pi(N)})\;,
\end{equation}
where $\pi$ is a permutation of the particle labels $\{1, \ldots, N\}$. We shall denote by $L_{\text{a}}^{2}((\mathbb{R}^{3} \times \{1, \ldots, M\})^{N})$ the set of all antisymmetric wave functions in $L^{2}((\mathbb{R}^{3} \times \{1, \ldots, M\})^{N})$. A simple example of fermionic wave functions is provided by Slater determinants, defined as antisymmetric tensor products of orthonormal functions in $L^{2}(\mathbb{R}^{3}\times \{1,\ldots, M\})$:
\begin{equation}\label{eq:Slaterdef}
\psi_{\text{Slater}}({\bf x}_{1}, \ldots {\bf x}_{N})  = \frac{1}{\sqrt{N!}} \sum_{\pi \in S_{N}} \text{sgn}(\pi) f_{1}({\bf x}_{\pi(1)}) \cdots f_{N}({\bf x}_{\pi(N)})\;,
\end{equation}
with $\langle f_{i}, f_{j} \rangle = \delta_{ij}$, and where $S_{N}$ is the set of permutations of $\{1, 2, \ldots, N\}$. These wave functions encode the least amount of correlations between the particles, besides the requirement of antisymmetry.

Let us now introduce the many-body Hamiltonian, which generates the dynamics of the system. A typical class of non-relativistic Hamiltonians is given by:
\begin{equation}\label{eq:HNtrap}
H^{\text{trap}}_{N} = \sum_{j=1}^{N} ( -\Delta_{j} + V_{\text{trap}}(x_{j}) ) + \lambda \sum_{i<j}^{N} V(x_{i} - x_{j})\;,
\end{equation}
where $V_{\text{trap}}$ is a confining potential, and $V$ is the two-body interaction, with coupling constant $\lambda$. We choose units such that the Planck constant is $\hbar=1$ and the mass of the particles is $m=1/2$. Here $\Delta_{j}$ is the Laplacian with respect to the coordinate $x_j$ (it does not involve the spin). The ground state energy, which physically corresponds to the energy of the system at zero temperature, is given by:
\begin{equation}
E_{N} = \inf_{\psi_{N} \in L_{\text{a}}^{2}((\mathbb{R}^{3} \times \{1, \ldots, M\})^{N})} \frac{\langle \psi_{N}, H_{N} \psi_{N} \rangle}{\|\psi_{N}\|_{2}^{2}}\;.
\end{equation}
Let $\psi_{N}$ be a minimizer of $H^{\text{trap}}_{N}$. Clearly, $\psi_{N}$ is a stationary state for the quantum dynamics generated by $H^{\text{trap}}_{N}$. In order to observe a non-trivial evolution, we turn off the external potential, and evolve the state according to the time-dependent Schr\"odinger equation:
\begin{equation}\label{eq:S}
i\partial_{\tau} \psi_{N,\tau} = \Big(- \sum_{j=1}^{N} \Delta_{j} + \lambda \sum_{i<j}^{N} V(x_{i} - x_{j})\Big) \psi_{N,\tau}\;,\qquad \psi_{N,0} = \psi_{N}\;,
\end{equation}
where $\tau \in \mathbb{R}$ is the time variable.
In physical applications, it is extremely hard to obtain quantitative understanding on the ground state of the system or on the solution of the Schr\"odinger equation, for realistic values of the number of particles $N$. One is therefore interested in approximate descriptions, which depend on much fewer degrees of freedom, and which allow to capture the collective behavior of the system on a macroscopic scale. The price to pay for the reduction in dimensionality is that the laws of motion for the effective theories are typically nonlinear. The paradigmatic scaling regime in which such effective nonlinear description is expected to arise is the mean-field scaling. In this setting, the potential $V_{\text{trap}}$ in (\ref{eq:HNtrap}) confines the system in a region whose volume is of order $1$ in $N$, and the coupling constant $\lambda$ is chosen to guarantee the balance between the kinetic and the interaction energy. Due to the growth as $N^{5/3}$ of the kinetic energy in the ground state of the Fermi gas (a general lower bound compatible with this growth can be proved using the Lieb-Thirring inequality, see {\it e.g.} \cite{LiSei}), the natural choice for the coupling constant is $\lambda = N^{-1/3}$.

In this scaling regime, the equilibrium and the dynamical properties of the system are expected to be well captured by  Hartree-Fock theory, defined as the restriction of the space of fermionic wave functions to the subsets of Slater determinants (\ref{eq:Slaterdef}). The advantage in considering this class of states is that they behave as Gaussian states, in the sense that their higher order correlation functions can be algorithmically retrieved from single particle information, in analogy with Gaussian distributions that are completely characterized by their covariance. More precisely, given $\psi_{N} \in L_{\text{a}}^{2}((\mathbb{R}^{3} \times \{1, \ldots, M\})^{N})$ we define its $k$-particle density matrix as the linear operator on $L^{2}((\mathbb{R}^{3} \times \{1, \ldots, M\})^{k})$ with integral kernel:
\begin{equation}
\begin{split}
&\gamma^{(k)}_{N} ({\bf x}_{1}, \dots, {\bf x}_{k}; {\bf y}_{1}, \ldots, {\bf y}_{k}) \\&\qquad = {N \choose k} \int d{\bf z}_{k+1}\cdots d{\bf z}_{N}\, \psi_{N}({\bf x}_{1}, \ldots, {\bf x}_{k}, {\bf z}_{k+1}, \ldots, {\bf z}_{N}) \overline{\psi_{N}({\bf y}_{1}, \ldots, {\bf y}_{k}, {\bf z}_{k+1}, \ldots, {\bf z}_{N}) }\;,
\end{split}
\end{equation}
with the notation:
\begin{equation}
\int d{\bf z}\, f({\bf z}) = \sum_{\sigma = 1}^{M} \int_{\mathbb{R}^{3}} dz \, f(z,\sigma)\;.
\end{equation}
The $k$-particle density matrix is important because it allows to compute expectation values of $k$-particle observables; it plays the same role as the $k$-particle marginal in classical statistical mechanics. For Slater determinants, it turns out that all $k$-particle density matrices are determined by the one-particle density matrix. Let $\omega_{N}$ be the one-particle density matrix of a Slater determinant (\ref{eq:Slaterdef}), which is equal to:
\begin{equation}\label{eq:omegaNSlater}
\omega_{N} = \sum_{i=1}^{N} |f_{i} \rangle \langle f_{i}|\;;
\end{equation}
in particular, observe that $\omega_{N}$ is a rank-$N$ orthogonal projector. The $k$-particle density matrix of a Slater determinant can be computed starting from the one-particle density matrix (\ref{eq:omegaNSlater}) via the application of the fermionic Wick's rule:
\begin{equation}\label{eq:wick0}
\gamma^{(k)}_{N} ({\bf x}_{1}, \dots, {\bf x}_{k}; {\bf y}_{1}, \ldots, {\bf y}_{k}) = \frac{1}{k!} \sum_{\pi\in S_{k}} \text{sgn}(\pi)\prod_{i=1}^{k} \omega_{N}({\bf x}_{i}; {\bf y}_{\pi(k)})\;;
\end{equation}
thus, all the information about the state described by a Slater determinant is encoded in the one-particle density matrix $\omega_{N}$. For instance, the energy of a Slater determinant is a functional of the one-particle density matrix only:
\begin{equation}\label{eq:HFf}
\begin{split}
\langle \psi_{\text{Slater}}, H^{\text{trap}}_{N} \psi_{\text{Slater}} \rangle &= \tr ( -\Delta + V_{\text{ext}} ) \omega_{N} + \frac{1}{2 N^{1/3}} \int d{\bf x} d{\bf y}\, V({\bf x} - {\bf y}) \big( \rho({\bf x}) \rho({\bf y}) - | \omega_{N}({\bf x};{\bf y}) |^{2}\big) \\
&=: \mathcal{E}^{\text{HF}}_{N}(\omega_{N})\;,
\end{split}
\end{equation}
with $\rho({\bf x}) = \omega_{N}({\bf x}; {\bf x})$ the density of particles at ${\bf x}$; $\mathcal{E}^{\text{HF}}_{N}(\omega_{N})$ is called the Hartree-Fock energy functional. 

Hartree-Fock theory can also be used to investigate the dynamical properties of the system. Observe preliminarily that, in this confined setting, the relevant time scale of the system is $\tau \propto N^{-1/3}$. This is due to the fact that the typical kinetic energy per particle grows as $N^{2/3}$, and hence the classical velocity of the particles is of order $N^{1/3}$. Thus, for times $\tau$ of order $N^{-1/3}$, particles are expected to cover a distance of order $1$, which has to be understood as the diameter of the region in space occupied by the initial datum. For this reason, $\tau \propto N^{-1/3}$ is the time-scale for a macroscopic change of the system. From now on, we will write $\tau = N^{-1/3} t$, and we will be interested in studying the dynamics of the system for times $t$ of order $1$ in $N$.  Recalling that at time $t=0$ the external potential is turned off,  the Schr\"odinger equation in the rescaled time variable reads:
\begin{equation}\label{eq:S0}
i\varepsilon \partial_{t} \psi_{N,t} = \Big( -\sum_{j=1}^{N} \varepsilon^{2}\Delta_{j} + \frac{1}{N} \sum_{i<j}^{N} V(x_{i} - x_{j})\Big)\psi_{N,t}\;,\qquad \psi_{N,0} = \psi_{N}\;,
\end{equation}
where $\varepsilon = N^{-1/3}$ has the role of an effective Planck constant. Since $\varepsilon$ becomes vanishingly small in the limit of large $N$, Eq. (\ref{eq:S0}) shows that the fermionic mean-field can be understood as the coupling of the usual mean-field with a semiclassical scaling. If one assumes that the state of the system is a Slater determinant for all times, then the one-particle density matrix of the system evolves according to the time-dependent Hartree-Fock equation:
\begin{equation}\label{eq:tHF}
i \varepsilon \partial_{t} \omega_{N,t} = [ h_{\text{HF}}(t), \omega_{N,t} ]\;,\qquad \omega_{N,0} = \omega_{N}\;,
\end{equation}
where the Hartree-Fock Hamiltonian reads $h_{\text{HF}}(t) = -\varepsilon^{2}\Delta + N^{-1} V*\rho_{t} - X_{t}$. Explicitly, $\rho_{t}({\bf x}) = \omega_{N,t}({\bf x}; {\bf x})$, $N^{-1} V*\rho_{t}$ is the spatial convolution between the density and the potential, and $X_{t}$ is the linear operator with integral kernel given by $X_{t}({\bf x}; {\bf y}) = N^{-1} V(x-y) \omega_{N,t}({\bf x}; {\bf y})$.

Of course, the restriction to the set of Slater determinants is a nontrivial assumption, and it is a natural mathematical question to understand when such approximation is justified. Concerning the equilibrium problem, the validity of the Hartree-Fock approximation for the ground state of many-body fermionic systems has been proved in \cite{Bach} for large atoms and molecules, and it has been extended to high density Coulomb systems (jellium) in \cite{GS}. More recently, there has been substantial progress in the development of nonperturbative methods to go beyond the Hartree-Fock approximation, and in particular to determine the large $N$ asymptotics of the correlation energy, defined as the difference between the many-body ground state energy minus the Hartree-Fock ground state energy; see \cite{HPR,BNPSSa,BNPSSb,CHN,BPSScorr,CHN2}.

Concerning quantum dynamics, the first derivation of the time-dependent Hartree-Fock equation (\ref{eq:tHF}) in the mean-field/semiclassical scaling has been obtained in \cite{EESY}, for analytic potentials and for short times. More recently, in \cite{BPS} (and later in \cite{PP} with a different method), this result has been generalized to a much larger class of potentials and to all times. The result of \cite{BPS} holds in the sense of convergence of density matrices. Convergence in the $L^{2}$-sense for a special class of initial data has been obtained in \cite{BNPSS}, using the rigorous bosonization methods of \cite{BNPSSa,BNPSSb,BPSScorr}. The result of \cite{BPS} has been extended to pseudo-relativistic fermions in \cite{BPS2} and to mixed states in \cite{BJPSS}. In the mean-field regime without semiclassical scaling, the derivation of the time-dependent Hartree-Fock equation has been obtained in \cite{BGGM}, for bounded potentials, and in \cite{FK}, for Coulomb potentials (see also \cite{BBPPT}). See \cite{BPSbook} for a review.  Concerning unbounded potentials in the mean-field/semiclassical scaling, the time-dependent Hartree-Fock equation for particles interacting through a Coulomb potential has been derived in \cite{PRSS}, under the assumption that a suitable semiclassical structure of the initial datum propagates along the flow of the Hartree-Fock equation.  In \cite{CLS}, a stronger version of this semiclassical structure,  suitable for mixed states, has been proved to propagate along the Hartree-Fock flow for a class of singular potentials. This allows to extend the derivation of the Hartree-Fock equation to singular potentials, also including the Coulomb potential for short times $t$ (vanishing as $N\to \infty$). Finally, we mention the recent work \cite{FPS}, where the derivation of \cite{BPS} has been generalized to the case of extended Fermi gases at high density, for bounded potentials.

Slater determiants are a special example of quasi-free states. In general, a quasi-free state is a state for which all correlation functions can be computed according to the Wick rule. If one works in the grand-canonical ensemble and relaxes the constraint that the number of particles is fixed, there exist other quantum states at zero temperature besides Slater determinants which are quasi-free. In order to introduce them, we will use the notations of second quantization and of the Fock space, introduced in detail in Section \ref{sec:fock}. 

A pure state is described by a vector in the Fock space, $\psi \in \mathcal{F} = \mathbb{C} \oplus \bigoplus_{n\geq 1} L^{2}_{\text{a}}((\mathbb{R}^{3} \times \{1, \ldots, M\})^{n})$, with $\|\psi\|_{\mathcal{F}} =1$. In general, the vector $\psi$ corresponds to a superposition of states with different particle numbers. The one-particle density matrix of $\psi$ can be expressed in Fock space as:
\begin{equation}
\gamma^{(1)}_{\psi}({\bf x}; {\bf y}) = \langle \psi, a^{*}_{{\bf y}} a_{{\bf x}} \psi \rangle\;,
\end{equation}
with $a_{{\bf x}}$ and $a_{{\bf y}}^{*}$ the operator-valued distributions associated with the usual fermionic creation and annihilation operators, whose definition will be recalled in Section \ref{sec:fock}. Since the state $\psi$ does not have in general a definite particle number (it can be a superposition of states with different number of particles), we can also introduce a nontrivial pairing matrix:
\begin{equation}
\pi^{(1)}_{\psi}({\bf x}; {\bf y}) = \langle \psi, a_{{\bf y}} a_{{\bf x}} \psi \rangle\;.
\end{equation}
We shall denote by $N$ the average number of particles of the state, defined as $N= \tr\, \gamma^{(1)}_{\psi}$. In this more general setting, quasi-free states are completely determined by $\gamma^{(1)}_{\psi} \equiv \omega_{N}$ and $\pi^{(1)}_{\psi} \equiv \alpha_{N}$, via a generalization of the Wick's rule (\ref{eq:wick0}); see Section \ref{sec:qfstates}. Furthermore, as discussed in Section \ref{sec:qfstates}, it turns out that the one-particle density matrix and the pairing matrix of a pure quasi-free state satisfy the constraints:
\begin{equation}
0 \leq \omega_{N} \leq 1\;, \quad \alpha_{N}^{T} = -\alpha_{N}\;,\quad \omega_{N} (1 - \omega_{N}) = \alpha_{N} \alpha_{N}^{*}\;,\quad \omega_{N} \alpha_{N} - \alpha_{N} \overline{\omega_{N}} = 0\;,
\end{equation}
where $O^{T}$ and $\overline{O}$ denote the transposition and the complex conjugation of the operator $O$, namely $\overline{O} = (O^{T})^{*}$. Next, let us introduce the second quantization of the Hamiltonian $H_{N}^{\text{trap}}$,
\begin{equation}
\mathcal{H}_{N}^{\text{trap}} = 0 \oplus \bigoplus_{n\geq 1} \Big( \sum_{j=1}^{n} ( -\Delta_{j} + V_{\text{ext}}(x_{j}) ) + \frac{1}{N^{1/3}} \sum_{i<j}^{n} V(x_{i} - x_{j}) \Big)\;.
\end{equation}
The energy of a quasi-free state $\psi$ is completely specified by $\omega_{N}$ and $\alpha_{N}$. We have:
\begin{equation}
\begin{split}
\langle \psi, \mathcal{H}^{\text{trap}}_{N} \psi \rangle &= \tr\,( -\Delta + V_{\text{trap}} ) \omega_{N} \\ &\quad + \frac{1}{2 N^{1/3}} \int d{\bf x} d{\bf y}\, V(x-y)\big( \rho({\bf x}) \rho({\bf y}) - |\omega_{N}({\bf x}; {\bf y})|^{2} + | \alpha_{N}({\bf x}; {\bf y}) |^{2} \big) \\
&=: \mathcal{E}_{N}^{\text{HFB}}(\omega_{N}, \alpha_{N})\;,
\end{split}
\end{equation}
where $\mathcal{E}_{N}^{\text{HFB}}(\omega_{N}, \alpha_{N})$ is called the Hartree-Fock-Bogoliubov (HFB) functional. This functional plays an important role in physics, for instance in the Bardeen-Cooper-Schrieffer (BCS) theory of superconductivity and in superfluidity. It differs from the Hartree-Fock functional (\ref{eq:HFf}) because of the presence of the last term, taking into account the energy of Cooper pairs. In general, it is a difficult problem to determine whether the ground state of the HFB functional is attained in correspondence with a nonzero pairing matrix; a necessary condition for this to happen is that the potential is non-positive. See \cite{BLS, HS} for mathematical reviews of BCS theory and of the Hartree-Fock-Bogoliubov approximation, as well as \cite{HHSS, BacFroJon-09,LenLew-10,LewPau-14,BraHaiSei-16,Lau,HL} for recent results about existence and properties of HFB/BCS minimizers and occurrence of pairing in the ground state. 

Let us now consider the many-body evolution of an initial quasi-free state, after removing the external trap. The quantum dynamics of the system in Fock space, parametrized by the macroscopic time variable $t$, is defined as
\begin{equation}\label{eq:S2}
i\varepsilon \partial_{t} \psi_{t} = \mathcal{H}_{N} \psi_{t}\;,\qquad \psi_{0} = \psi\;,
\end{equation}
with $\mathcal{H}_{N}$ given by:
\begin{equation}
\mathcal{H}_{N} = 0 \oplus \bigoplus_{n\geq 1} \Big( -\sum_{j=1}^{n} \varepsilon^{2}\Delta_{j} + \frac{1}{2N} \sum_{i<j}^{n} V(x_{i} - x_{j}) \Big)\;.
\end{equation}
Similarly to the case of Hartree-Fock theory, if one assumes that the state remains quasi-free for all times, one obtains the following coupled nonlinear evolution equations for the one-particle density matrix and for the pairing matrix, which generalize (\ref{eq:tHF}) to the case in which the number of particles is not fixed:
\begin{equation}\label{eq:HFB0}
\begin{split}
i \varepsilon \partial_t \omega_{N,t} &= \big[ h_{\text{HF}}(t),  \omega_{N,t} \big] + \Pi_{t} \alpha^*_{N,t} - \alpha_{N,t} \Pi_{t} ^* \\
i \varepsilon \partial_t \alpha_{N,t} &= h_{\text{HF}}(t)\alpha_{N,t} + \alpha_{N,t} \overline{h_{\text{HF}}(t)} + \Pi_{t} \big( 1- \overline{\omega}_{N,t} \big) - \omega_{N,t} \Pi_{t}\;,
\end{split}
\end{equation}
where $\Pi_{t}$ a linear operator with integral kernel given by $\Pi_{t}({\bf x},{\bf y}):= \frac{1}{N} V(x-y)\alpha_{N,t}({\bf x},{\bf y})$. The equation (\ref{eq:HFB0}) for the pair $(\omega_{N,t}, \alpha_{N,t})$ is called the Hartree-Fock-Bogoliubov (HFB) equation. 

HFB theory is expected to be a good approximation of many-body quantum mechanics in the mean-field regime, for many-body systems with attractive interactions. However, in contrast to the Hartree-Fock case, until now there has been no rigorous proof of the validity of the HFB approximation from the many-body Schr\"odinger equation, neither at equilibrium nor for quantum dynamics. See \cite{BSS} for the proof of well-posedness of the fermionic HFB equation, and for the derivation of the HFB equation (\ref{eq:HFB0}) via the Dirac-Frenkel principle. See also \cite{HLLS} for the proof of blow-up at finite time of the relativistic fermionic HFB equation. The rigorous connection between the HFB dynamics and other nonlinear effective models has also been discussed in the literature. For instance, in \cite{HaSc} it has been shown that, in the low density regime, the dynamics of the macroscopic variations of the pair density in BCS theory can be approximated by the time-dependent Gross-Pitaevskii equation, as the evolution of a condensate of pairs. Furthermore, the work \cite{FHSS} discussed the incompatibility between the HFB dynamics and the Ginzburg-Landau dynamics, as a purely nonlinear phenomenon. Finally, a very recent work \cite{CLS2} obtained a classical equation describing solutions to the HFB equation \eqref{eq:HFB0} as $\varepsilon\to0$.

In this paper we prove the validity of the Hartree-Fock-Bogoliubov equation (\ref{eq:HFB0}) for bounded and regular enough interaction potentials, in the mean-field regime. We shall consider the many-body dynamics of initial data $\psi \in \mathcal{F}$ that are approximated by a suitable class of quasi-free states with one-particle density matrix $\omega_{N}$ and with pairing matrix $\alpha_{N}$, with $N = \tr\, \omega_{N}$. We shall suppose that:
\begin{equation}\label{eq:sc0}
\begin{split}
& \sup_{p\in \mathbb{R}^{3}} \frac{1}{1+|p|} \| [ \omega_{N}, e^{ip\cdot \hat x} ] \|_{\text{HS}} \leq C  N^{1/3}\;,\qquad \| [ \omega_{N}, \varepsilon \nabla ] \|_{\text{HS}} \leq C N^{1/3}\;, \\ & \| [ \alpha_{N}, \varepsilon \nabla ] \|_{\text{HS}} \leq C N\;,\qquad \|\alpha_{N}\|_{\text{HS}} \leq CN^{1/3}\;,
\end{split}
\end{equation}
where $\| \cdot \|_{\text{HS}}$ indicates the Hilbert-Schmidt norm, $\|O\|_{\text{HS}}^{2} = \tr\, |O|^{2}$. We shall suppose that the initial datum $\psi$ is close, in a suitable sense to be specified below, to a quasi-free state parametrized by $\omega_{N}, \alpha_{N}$, satisfying the bounds (\ref{eq:sc0}). In particular:
\begin{equation}
\| \gamma^{(1)}_{\psi} - \omega_{N} \|_{\text{HS}} \ll N^{1/2}\;,\qquad \| \pi^{(1)}_{\psi} - \alpha_{N} \|_{\text{HS}} \ll N^{1/3}\;,
\end{equation}
that is, $\| \gamma^{(1)}_{\psi} - \omega_{N} \|_{\text{HS}} / N^{1/2}$ and $\| \pi^{(1)}_{\psi} - \alpha_{N} \|_{\text{HS}} / N^{1/3}$ converge to zero as $N\to \infty$.

The estimates (\ref{eq:sc0}) encode the semiclassical structure of the initial datum, and are expected to hold for quasi-free states at low energy. For $\alpha_{N} = 0$, similar estimates play a key role in the derivation of the time-dependent Hartree-Fock equation in the mean-field/semiclassical scaling \cite{BPS}. It is not difficult to check them for the coherent-states approximation of confined equilibrium states at zero temperature, see {\it e.g.} Appendix A.2 in \cite{FPS}.  These estimates can be proved to hold for the ground state of confined, non-interacting systems \cite{FM}\footnote{The bounds of \cite{FM} are in trace norm, but the adaptation to the Hilbert-Schmidt case is straightforward. We thank the Referee for pointing this out.}.

The novelty with respect to \cite{BPS} are the last two estimates in (\ref{eq:sc0}). The third estimate can be shown to be related to the kinetic energy of $\alpha_{N}$. Instead, the last estimate captures the fact that we are looking at initial data that are close to Slater determinants, since $\| \alpha_{N} \|_{\text{HS}}^{2} = \tr\, \omega_{N}(1-\omega_{N})$, and $\| \alpha_{N} \|_{\text{HS}}^{2} \leq CN^{2/3}$ while $\tr\,\omega_{N} = N$. Informally, this class of initial data can be viewed as obtained from Slater determinants after introducing a local smoothing of the distribution of the occupation number around the Fermi energy. The smallness of $\|\alpha_{N}\|_{\text{HS}}$ plays an important role in our proof; it is an interesting question to understand whether the derivation of the HFB equation extends to initial data with larger $\|\alpha_{N}\|_{\text{HS}}$. In Appendix \ref{app:BCS} we prove the validity of bounds analogous to (\ref{eq:sc0}) for all translation-invariant quasi-free states on the torus of side $2\pi$ at low enough energy (actually, for this class of states the third estimate in (\ref{eq:sc0}) holds with $N$ replaced by $N^{1/3}$). Furthermore, we discuss a class of natural states for which the bounds hold, obtained from a smearing of the characteristic function of the Fermi ball, in an order $1$ neighbourhood of the Fermi surface (while the Fermi momentum is of order $N^{1/3}$). If translation-invariance is assumed, the lowest energy state of the HFB functional on the torus turns out to be the free Fermi gas, which in particular has $\alpha_{N} = 0$. This is ultimately a consequence of the existence of a spectral gap for the relative kinetic energy at the Fermi level, due to the finite volume of the torus. Thus, for confined translation-invariant states, the condition $\alpha_{N} \neq 0$ implies that the state is an excited state. In general, we do not know whether the ground state of the HFB functional is attained in correspondence with a nonzero pairing matrix. Numerical investigations, see {\it e.g.} \cite{LewPau-14}, suggest that the minimizer of the HFB functional is indeed attained in correspondence with a small nonzero pairing matrix. Finally, let us mention that the persistence of translation-invariance for the ground state of the BCS functional has been discussed in \cite{DGHL}, in the thermodynamic limit. There it is shown that translation-invariance persists in two-dimensions, and also in three-dimensions in the case of vanishing angular momentum.

Let us denote by $\gamma^{(1)}_{N,t}$ and $\pi^{(1)}_{N,t}$ the many-body evolution of the one-particle density matrix and of the pairing matrix associated with the solution of (\ref{eq:S2}), and by $\omega_{N,t}, \alpha_{N,t}$ the solution of Eqs. (\ref{eq:HFB0}). In this work we prove that the many-body dynamics of $\gamma^{(1)}_{N,t}$, $\pi^{(1)}_{N,t}$ stay close to the solution of the HFB equation, in the sense that:
\begin{equation}\label{eq:1pdm}
\| \gamma^{(1)}_{N,t} - \omega_{N,t} \|_{\text{HS}} \ll C(t) N^{1/2}\;,\qquad \| \pi^{(1)}_{N,t} - \alpha_{N,t} \|_{\text{HS}} \ll C(t) N^{1/3}\qquad \forall t\in \mathbb{R},
\end{equation}
where the function $C(t)$ grows as a double exponential in time, and it is independent of $N$. For initial data such that $\|\omega_{N}\|_{\text{HS}} \simeq N^{1/2}$ and $\|\alpha_{N}\|_{\text{HS}} \simeq N^{1/3}$, which are allowed by our assumptions, the bounds (\ref{eq:1pdm}) imply
that the difference between the many-body dynamics and the HFB dynamics at the level of the one-particle density matrix is much smaller than the individual contributions. With a slight abuse of terminology, we say that our result is a derivation of HFB in the sense of convergence of density matrices, even if we do not prove any existence of a limiting object as $N\to+\infty$.  More generally, our result also extends to higher order density matrices. For instance, we are able to prove that the two-particle density matrix $\gamma^{(2)}_{N,t}$, an operator on $L^{2}((\mathbb{R}^{3} \times \{ 1,\ldots, M \})^{2})$, can be approximated as:
\begin{equation}
\begin{split}
\gamma^{(2)}_{N,t}({\bf x}_{1}, {\bf x}_{2}; {\bf y}_{1}, {\bf y}_{2}) &\simeq \omega_{N,t}({\bf x}_{2}; {\bf y}_{1}) \omega_{N,t}({\bf x}_{1}; {\bf y}_{2}) - \omega_{N,t}({\bf x}_{1}; {\bf y}_{1}) \omega_{N,t}({\bf x}_{2} ; {\bf y}_{2}) \\ 
&\quad + \overline{\alpha_{N,t}({\bf y}_{2}; {\bf y}_{1})} \alpha_{N,t}({\bf x}_{1}; {\bf x}_{2})\;.
\end{split}
\end{equation}
in the HS topology.

The method of proof is based on an extension of the technique of \cite{BPS}. As in \cite{BPS}, a key role is played by the fluctuation dynamics, a unitary evolution in Fock space that does not preserve the number of particles, and that allows to understand the distance between the many-body state and the reference quasi-free states in term of particle excitations over a suitable time-dependent Fock space vacuum. With respect to the Hartree-Fock case, however, in the present setting the control of the fluctuation dynamics is considerably more involved, due to the more complex structure of the effective evolution equation (\ref{eq:HFB0}), and to the fact that $\omega_{N}$ is not an orthogonal projection. Also, the propagation of the semiclassical structure (\ref{eq:sc0}) along the HFB flow is more involved than in the HF case. Let us point out that the proof could be simplified, at the cost of assuming a stronger notion of semiclassical structure (\ref{eq:sc0}), replacing in (\ref{eq:sc0}) the Hilbert-Schmidt norm by the square root of the trace norm, similarly to \cite{BPS}. Furthermore, with this stronger notion of semiclassical structure we could also prove that the Hartree-Fock-Bogoliubov dynamics cannot be distinguished from the Hartree-Bogoliubov (HB) dynamics, given by Eqs. (\ref{eq:HFB0}) after neglecting all terms involving the operators $X_{t}$ and $\Pi_{t}$: the distance between the solutions of the HFB and HB evolutions can be proved to be smaller than the estimate quantifying the distance between many-body and HFB dynamics (\ref{eq:1pdm}). Since we have not been able to verify the trace-norm semiclassical structure for interesting quasi-free states with non-zero pairing, we prefer to work under the weaker semiclassical structure given by (\ref{eq:sc0}). 

In order to deal with the weaker form of semiclassical structure (\ref{eq:sc0}), we introduce a suitable auxiliary fluctuation dynamics, as already done in the case of fermionic mixed states (without pairing) \cite{BJPSS}. The growth of fluctuations over this new dynamics can be controlled relying only on (\ref{eq:sc0}); furthermore, the auxiliary dynamics and the fluctuation dynamics can be proved to be close in norm. Then, a bootstrap argument allows to control the growth of fluctuations on the original fluctuation dynamics, and ultimately to prove convergence of the many-body dynamics to the HFB evolution.

The paper is organized as follows. In Section~\ref{sec:fock} we review the formalism of second quantization, we recall useful estimates in Fock space, we introduce the Bogoliubov transformations and the quasi-free states. In Section~\ref{sec:main} we state our main result, Theorem~\ref{thm:main}, about the derivation of the time-dependent HFB equation from the many-body dynamics in the semiclassical/mean-field regime. The rest of the paper is devoted to the proof of Theorem~\ref{thm:main}. In Section~\ref{sec:fluct} we introduce the fluctuation dynamics, and prove a key identity for the growth of fluctuations, Proposition~\ref{prp:growth}; in Proposition~\ref{prp:sc} we prove the propagation of the semiclassical structure (\ref{eq:sc0}), and we use it to control the growth of fluctuations on the auxiliary fluctuation dynamics, Proposition~\ref{prp:Nk}. Then, in Proposition~\ref{prp:normapprox} we prove norm-closeness of the auxiliary fluctuation dynamics and of the full fluctuation dynamics, and we use this result to prove a bound on the growth of fluctuations on the full fluctuation dynamics, Proposition~\ref{prp:Nktrue}. In Section~\ref{sec:proofmain} we put everything together and we prove Theorem~\ref{thm:main}. Finally, in Appendix~\ref{app:BCS} we prove the validity of the estimates (\ref{eq:sc0}) for translation-invariant quasi-free states with low energy.

\paragraph{Acknowledgements.} We thank Benjamin Schlein for useful discussions.  S.M.  and M.P. acknowledge financial support by the European Research Council (ERC) under the European Union’s Horizon 2020 research and innovation program ERC StG MaMBoQ,  n.\ 802901.  S.M.  acknowledges financial support from the MSCA project ConNEqtions,  n.\ 101056638.  The work of S.M.  and M.P.  has been carried out under the auspices of the GNFM of INdAM. We thank the anonymous Referees for their comments on a previous version of the manuscript.

\paragraph{Declarations.} The authors have no competing interests to declare that are relevant to the content of this article.

\section{Fock space, Bogoliubov transformations and quasi-free states}\label{sec:fock}

\subsection{Fock space}
Let $\frak{h} = L^{2}(\mathbb{R}^{3} \times \{1, \ldots, M\})$, with $M\in \mathbb{N}$, be the single-particle Hilbert space. Let $\mathcal{F}$ be the fermionic Fock space over $\frak{h}$,
\begin{align*}
\notag
\F = \mathbb{C} \oplus \bigoplus_{n\geq 1} \frak{h}^{\wedge n}\;.
\end{align*}
Vectors in the Fock space correspond to infinite sequences of functions $(\psi^{(n)})_n$ with $\psi^{(n)} \in \frak{h}^{\wedge n}$. We shall use the notation
\begin{equation}
\psi^{(n)} \equiv \psi^{(n)}({\bf x}_{1}, \ldots, {\bf x}_{n})\;,
\end{equation}
where ${\bf x} = (x, \sigma)$ with $x\in \mathbb{R}^{3}$ and $\sigma \in \{1, \ldots, M\}$. A simple example of vector in $\mathcal{F}$ is the vacuum state, $\Omega = (1, 0, \ldots, 0, \ldots)$. Given $\psi_1,\psi_2 \in \mathcal{F}$, we introduce the Fock space scalar product as:
\begin{equation*}
\begin{split}
\langle \psi_{1}, \psi_{2} \rangle &= \sum_{ n \geq 0 } \langle \psi_{1}^{(n)},\psi_{2}^{(n)}  \rangle_{\frak{h}^{\otimes n}} \\
&= \sum_{n\geq 0} \int d{\bf x}_{1} \ldots d{\bf x}_{n}\, \overline{\psi^{(n)}_{1}({\bf x}_{1}, \ldots, {\bf x}_{n})} \psi^{(n)}_{2}({\bf x}_{1}, \ldots, {\bf x}_{n})\;,
\end{split}
\end{equation*}
where we recall the notation $\int d{\bf x}\, f({\bf x}) := \sum_{\sigma=1}^{M} \int_{\mathbb{R}^{3}} d x\, f(x, \sigma)$. Equipped with this natural inner product, the Fock space $\mathcal{F}$ is a Hilbert space. We denote by $\| \cdot \|$ the norm induced by this inner product, and with a slight abuse of notation we shall use the same notation to denote the operator norm of linear operators acting on $\mathcal{F}$.

It is convenient to introduce creation and annihilation operators on $\mathcal{F}$. Let $f\in \frak{h}$. We define the creation operator $a^{*}(f)$ and the annihilation operator $a(f)$ as
\begin{equation*}
\begin{split} 
\left(a^*(f) \psi\right)^{(n)}({\bf x}_1,\ldots, {\bf x}_n) & := \frac{1}{\sqrt{n}} \sum_{j=1}^n (-1)^{j-1} f({\bf x}_j)\psi^{(n-1)}({\bf x}_1, \ldots, {\bf x}_{j-1}, {\bf x}_{j+1}, \ldots, {\bf x}_n) \\
\left(a(f)\psi\right)^{(n)}({\bf x}_1,\ldots, {\bf x}_n)&:= \sqrt{n+1} \int d {\bf x} \; \overline{f({\bf x})} \, \psi^{(n+1)} ({\bf x}, {\bf x}_1,\ldots, {\bf x}_n)\;,
\end{split}
\end{equation*}
for any $\psi \in \mathcal{F}$. These definitions imply that $a(f) \Omega = 0$.

Also, it turns out that $a^{*}(f) = a(f)^{*}$ and one can check that the canonical anticommutation relations hold true:
\begin{equation*}
\{ a(f), a(g) \} = \{ a^{*}(f), a^{*}(g) \} = 0\;,\qquad \{ a(f), a^{*}(g) \} = \langle f, g \rangle_{\frak{h}}\;.
\end{equation*}
Due to these relations, one has $\|a(f)\| \leq \|f\|_{2}$, $\|a^{*}(f)\|\leq \|f\|_{2}$. We will also make use of operator-valued distributions $a_{\bf x}^*$ and $ a_{\bf x}$, in terms of which:
\begin{equation*}
a^* (f) = \int d {\bf x} \, f({\bf x})\, a_{\bf x}^* , \quad a(f) = \int d {\bf x}  \, \overline{f ({\bf x})}\, a_{\bf x}\,. 
\end{equation*}
Note that the integrals denote the action of the distribution rather than integration with respect to the Lebesgue measure. The creation and annihilation operators can be used to define the second quantization of observables. To begin, consider the number operator $\mathcal{N}$, acting on a given Fock space vector as $(\mathcal{N} \psi)^{(n)} = n \psi^{(n)}$. In terms of the operator-valued distributions, it can be written as:
\begin{equation*}
\cN = \int d {\bf x} \, a_{\bf x}^* a_{\bf x} \;.
\end{equation*}
The average number of particles over a state $\psi$ is then:
\begin{equation}\label{eq:Nav}
N = \langle \psi, \mathcal{N} \psi \rangle\;.
\end{equation}
More generally, for a given one-particle operator $O$ on $\frak{h}$ we define its second quantization $d\Gamma(O)$ as the operator on the Fock space acting as follows:
\begin{equation*}
 d\Gamma (O) \upharpoonright_{\F^{(n)}} = \sum_{j=1}^n O^{(j)}
\end{equation*}
where $O^{(j)} = \mathbf{1}^{\otimes (n-j)} \otimes O \otimes \mathbf{1}^{\otimes (j-1)}$, and the symbol $\upharpoonright_{\F^{(n)}}$ indicates the restriction to the $n$-particle subspace $\F^{(n)}$. If $O$ has the integral kernel $O ({\bf x}; {\bf y})$, we can represent $d\Gamma (O)$ as:
\begin{equation*}
d\Gamma (O) = \int  d {\bf x} d {\bf y} \, O ({\bf x}; {\bf y}) a_{\bf x}^* a_{\bf y}\;.
\end{equation*}
In the next lemma we collect some bounds for the second quantization of one-particle operators, that will play an important role in the rest of the paper. Here and in the following, we denote by $\| \cdot \|$ the operator norm, by $\| \cdot \|_{\mathrm{HS}}$ the Hilbert-Schmidt norm and by $\| \cdot \|_{\mathrm{tr}}$ the trace norm. We refer the reader to {\it e.g.} \cite[Lemma 3.1]{BPS} for the proof.
\begin{lemma}[Estimates for Fock space operators.]\label{lem:bounds}
Let $O$ be a bounded operator on $\frak{h}$. We have, for any $\psi \in \mathcal{F}$:
\begin{equation*}
|\langle \psi, d\Gamma (O) \psi \rangle| \leq \| O \| \langle \psi, \cN \psi \rangle \;, \qquad \left\|d\Gamma (O) \psi \right\| \leq \left\| O \right\| \left\| \cN \psi\right\|\;.
\end{equation*}
Let $O$ be a Hilbert-Schmidt operator. We have, for any $\psi \in \mathcal{F}$:
\begin{equation*}
\begin{split} 
\left\|d\Gamma (O) \psi \right\| &\leq \left\| O \right\|_{\mathrm{HS}} \left\| \mathcal{N}^{1/2} \psi \right\| \\
\left\| \int d {\bf x} d {\bf x}'\, O ({\bf x}; {\bf x}') a_{\bf x} a_{{\bf x}'} \psi \right\| & \leq \| O \|_{\mathrm{HS}} \left\| \mathcal{N}^{1/2} \psi \right\| \\
\left\| \int d {\bf x} d {\bf x}'\, O ({\bf x}; {\bf x}') a^*_{\bf x} a^*_{\bf x'} \psi \right\| & \leq 2\| O \|_{\mathrm{HS}} \left\| (\mathcal{N}+1)^{1/2} \psi \right\|\;. 
\end{split}
\end{equation*}
Let $O$ be a trace class operator. We have, for any $\psi \in \mathcal{F}$:
\begin{equation*}
\begin{split}
\left\|d\Gamma (O) \psi \right\| &\leq 2  \left\| O \right\|_{\mathrm{tr}}\|\psi\| \\
\left\| \int d {\bf x} d {\bf y}\, O ({\bf x}; {\bf x}') a_{\bf x} a_{\bf x'} \psi \right\| & \leq 2  \left\| O \right\|_{\mathrm{tr}} \|\psi\| \\
\left\| \int d {\bf x} d {\bf y}\,  O ({\bf x}; {\bf x}') a^*_{\bf x} a^*_{\bf x'} \psi \right\| & \leq 2  \left\| O \right\|_{\mathrm{tr}} \|\psi\|  \;.
\end{split}
\end{equation*}
\end{lemma}
Given a state $\psi \in \mathcal{F}$, we define its one-particle density matrix as the non-negative trace class operator $\gamma^{(1)}_{N}$ on $\frak{h}$ with integral kernel: 
\begin{equation}
\label{eq:gamma-FS} 
\gamma^{(1)}_{N} ({\bf x};{\bf y}) = \langle \psi, a_{\bf y}^* a_{\bf x} \psi \rangle \;.
\end{equation}
Thus, by (\ref{eq:Nav}), $N = \tr\, \gamma^{(1)}_{N}$. Furthermore, given a one-particle observable $O$, we have:
\begin{equation}
\label{expectation_one_body}
\langle \psi, d\Gamma (O) \psi \rangle = \int d {\bf x} d {\bf y} \, O ({\bf x};{\bf y}) \, \langle \psi, a_{\bf x}^* a_{\bf y} \psi \rangle = \tr \, O \gamma^{(1)}_N\;.
\end{equation}
Also, we define the pairing matrix as the operator $\pi^{(1)}_{N}$ on $\frak{h}$ with integral kernel:
\begin{equation}
\pi^{(1)}_{N}({\bf x}; {\bf y}) = \langle \psi, a_{{\bf y}} a_{\bf x} \psi \rangle\;.
\end{equation}
In the special case in which the vector $\psi$ in an $N$-particle state, $\mathcal{N} \psi = N \psi$, the pairing matrix $\pi^{(1)}_{N}$ is trivially vanishing. In general, if the state $\psi$ does not have a definite number of particles, as it will be the case in the present work,  $\pi^{(1)}_{N}$ will be nonzero. From their definitions, the one-particle density matrix and the pairing matrix satisfy:
\begin{equation}\label{eq:24}
0\leq \gamma^{(1)}_{N} \leq \mathbbm{1}\;, \qquad \| \pi^{(1)}_{N} \| \leq 1\;,\qquad \pi^{(1)}_{N}({\bf x}; {\bf y}) = - \pi^{(1)}_{N}({\bf y}; {\bf x})\;.
\end{equation}
Next, we lift the many-body Hamiltonian to the Fock space, as follows. We define the Hamiltonian in Fock space as $\mathcal{H}_{N} = 0 \oplus \bigoplus_{n\geq 1} \mathcal{H}_{N}^{(n)}$, where
\begin{equation*}
\mathcal{H}_{N}^{(n)} = \sum_{i=1}^{n} -\varepsilon^{2} \Delta_{i} + \frac{1}{N} \sum_{i<j}^{n} V(x_{i} - x_{j})\;,
\end{equation*}
with $V$ a real-valued and even function (further regularity assumptions will follow). In terms of the operator-valued distributions, we can rewrite the Fock space Hamiltonian as:
\begin{equation}
\label{eq:FocKHN} 
\cH_{N} = \veps^2 \int d {\bf x} \,  \nabla  a_{{\bf x}}^{*} \nabla a_{{\bf x}} + \frac{1}{2N} \int d {\bf x} d {\bf y} \, V(x-y) a_{{\bf x}}^{*} a_{{\bf y}}^{*} a_{{\bf y}} a_{{\bf x}} \;.
\end{equation}
Finally, the time evolution of a state in the Fock space is defined as $\psi_{t} = e^{-i\mathcal{H}_{N} t / \varepsilon} \psi$. That is, each sector evolves according to the $n$-particle Schr\"odinger equation,
\begin{equation}
i\varepsilon \partial_{t} \psi^{(n)}_{t} = H^{(n)}_{N} \psi^{(n)}_{t}\;,\qquad \psi_{0}^{(n)} = \psi^{(n)}\;.
\end{equation}
\subsection{Bogoliubov transformations}
Here we shall give a brief introduction to the (fermionic) Bogoliubov transformations, which will play a crucial role in our analysis. We refer the reader to {\it e.g.} \cite{Solovej, DG} for a more detailed introduction to the topic.

Given $f, g \in \frak{h}$, let us define the operator $A(f,g)$ as:
\begin{equation}
A(f,g) := a(f) + a^{*}(\overline{g})\;,\qquad A^{*}(f,g) := (A(f,g))^{*} = a^{*}(f) + a(\overline{g})\;,
\end{equation}
where $\overline{f}$ denotes the usual complex conjugate of a function $f$ on $L^{2}(\mathbb{R}^{3} \times \{1,\ldots, M\})$. That is:
\begin{equation}
A^{*}(f,g) = A(\overline{g}, \overline{f})\;.
\end{equation}
These operators can be viewed as generalizing the notion of creation and annihilation operators, to which they reduce for $g = 0$. They satisfy the anticommutation relations:
\begin{equation}
\{ A(f_{1}, g_{1}), A^{*}(f_{2}, g_{2}) \} = \langle (f_{1}, g_{1}), (f_{2}, g_{2}) \rangle_{\frak{h} \oplus \frak{h}}\;.
\end{equation}
Notice that $\{ A(f_{1}, g_{1}), A(f_{2}, g_{2}) \}$, $\{ A^{*}(f_{1}, g_{1}), A^{*}(f_{2}, g_{2}) \}$ do not vanish, in general. A fermionic Bogoliubov transformation is a linear map $\nu: \frak{h} \oplus \frak{h} \to \frak{h} \oplus \frak{h}$, such that, for all $f_{1}, f_{2}, g_{1}, g_{2}$ in $\frak{h}$:
\begin{equation}\label{eq:uni}
\Big\{ A(\nu(f_{1}, g_{1})), A^{*}(\nu(f_{2}, g_{2})) \Big\} = \Big\{ A(f_{1}, g_{1}), A^{*}(f_{2}, g_{2}) \Big\}\;,
\end{equation}
and such that, for all $f, g$ in $\frak{h}$:
\begin{equation}\label{eq:conj}
A^{*}(\nu(f,g)) = A(\nu(\overline{g},\overline{f}))\;.
\end{equation}
In other words, the operators $B(f,g) = A(\nu(f,g))$ satisfy the same anticommutation relations of $A(f,g)$. The condition (\ref{eq:uni}) is equivalent to $\nu^{*} \nu = \mathbbm{1}$, while the condition (\ref{eq:conj}) is equivalent to, calling $\mathcal{C}$ the complex conjugation operator:
\begin{equation}
\nu \mathcal{C} = \mathcal{C} \nu\;.
\end{equation}
From these relations, one deduces that $\nu$ is a Bogoliubov transformation if and only if it has the form:
\begin{equation}\label{eq:nu}
\nu = \begin{pmatrix} u & \overline{v} \\ v & \overline{u} \end{pmatrix}\;,
\end{equation}
where $u, v$ are linear operators on $\frak{h}$ such that:
\begin{equation}\label{eq:bogo2}
\begin{split}
&u^* u + v^* v= \mathbbm{1},  \quad u^* \overline{v} + v^* \overline{u}= 0 \\
&u u^* + \overline{v} v^T = \mathbbm{1}, \quad u v^* + \overline{v} u^T =0\;,
\end{split}
\end{equation}
with $\overline{O} = \mathcal{C} O \mathcal{C}$. We say that a Bogoliubov transformation is implementable if there exists a unitary operator $R$ on $\mathcal{F}$ such that:
\begin{equation}\label{eq:def-BogTransf}
R^{*} A(f,g) R = A(\nu (f,g))\;,\qquad \text{for all $f,g\in \frak{h}$.}
\end{equation}
It is well-known that a Bogoliubov transformation is implementable if and only if $v$ is a Hilbert-Schmidt operator; this fact goes under the name of Shale-Stinespring condition. The operator $R$ is called the implementor of the Bogoliubov transformation; with a slight abuse of terminology, in the following we will refer to $R$ as the Bogoliubov transformation.

We shall extensively use the following transformation property of the operator-valued distributions:
\begin{equation}\label{eq:RaR}
R^{*} a_{{\bf x}} R = a(u_{{\bf x}}) + a^{*}(\overline{v}_{{\bf x}})\;,\qquad R^{*} a^{*}_{{\bf x}} R = a^{*}(u_{{\bf x}}) + a(\overline{v}_{{\bf x}})\;,
\end{equation}
where $u_{{\bf x}}:=u(\cdot;{\bf x})$.  Furthermore, Eq. \eqref{eq:def-BogTransf} implies that $R^{*}$ is the implementor of the Bogoliubov transformation associated with $\nu^{*} \equiv \nu^{-1}$. In particular, the following relations hold:
\begin{equation}\label{eq:RaR2}
R a_{\bf x} R^* = a(u^*_{{\bf x}}) + a^*(v^*_{{\bf x}}), \qquad R a^*_{\bf x} R^* = a^*(u^*_{{\bf x}}) + a(v^*_{{\bf x}})\;.
\end{equation}
\subsection{Quasi-free states and Hartree-Fock-Bogoliubov equation}\label{sec:qfstates}
Given a normalized state $\psi \in \mathcal{F}$, we define its generalized one-particle density matrix $\Gamma_{\psi}$ as the operator on $\frak{h} \oplus \frak{h}$ such that:
\begin{equation}
\langle (f_{1}, g_{1}), \Gamma_{\psi} (f_{2}, g_{2}) \rangle_{\frak{h} \oplus \frak{h}} = \langle \psi, A^{*}(f_{2}, g_{2}) A(f_{1}, g_{1}) \psi  \rangle_{\mathcal{F}}\;.
\end{equation}
A simple computation shows that $\Gamma_{\psi}$ can be expressed in terms of the one-particle density matrix and of the pairing matrix as:
\begin{equation}
\Gamma_{\psi} = \begin{pmatrix} \gamma^{(1)}_{N} & \pi^{(1)}_{N} \\ - \mathcal{C} \pi_{N}^{(1)} \mathcal{C} & 1 - \mathcal{C}\gamma_{N}^{(1)}\mathcal{C} \end{pmatrix}\;,
\end{equation}
where we used that $\pi^{(1)}_{N}$ is antisymmetric (recall Eq. (\ref{eq:24})).

Pure quasi-free states correspond to vectors in the Fock space of the form $R \Omega$, with $R$ a Bogoliubov transformation. Suppose that $\langle R\Omega, \mathcal{N} R \Omega \rangle = N$. We shall denote by $\omega_{N}$ and by $\alpha_{N}$ the one-particle density matrix and the pairing matrix of $R\Omega$, respectively: 
\begin{equation}
\Gamma_{R\Omega} = \begin{pmatrix} \omega_{N} & \alpha_{N} \\ -\overline{\alpha_{N}} & 1- \overline{\omega_{N}} \end{pmatrix}\;.
\end{equation}
The density matrix of $R\Omega$ can be easily represented in terms of the operators $u$ and $v$ entering in the definition of $\nu$, recall Eq. (\ref{eq:nu}). We have:
\begin{equation}\label{eq:omegaalpha}
\omega_{N} = v^{*} v\;,\qquad \alpha_{N} = v^{*} \overline{u}\;.
\end{equation}
As a sanity check, notice that $\omega_{N} \leq \mathbbm{1}$ as a consequence of the first identity of (\ref{eq:bogo2}), while $\alpha_{N}^{T} = -\alpha_{N}$:
\begin{equation}\label{eq:alphaT}
\alpha_{N}^{T} = \overline{\alpha_{N}^{*}} = \overline{ \overline{u}^{*} v} = u^{*} \overline{v} = -v^{*} \overline{u} = -\alpha_{N}\;,
\end{equation}
where the fourth equality follows from the second identity of (\ref{eq:bogo2}). One can also check that the generalized one-particle density matrix of $R\Omega$ is a projection:
\begin{equation}\label{eq:square}
\Gamma_{R\Omega} = \Gamma_{R\Omega}^{2}\;.
\end{equation}
In particular, the following relations hold:
\begin{equation}\label{eq:consproj}
\omega_{N} (1 - \omega_{N}) = \alpha_{N} \alpha_{N}^{*}\;,\qquad \omega_{N} \alpha_{N} - \alpha_{N} \overline{\omega_{N}} = 0\;.
\end{equation}
%t
Quasi-free states are completely determined by their generalized one-particle density matrix: all correlation functions can be computed starting from the one-particle density matrix and the pairing matrix via the Wick rule \cite{Solovej}. Let us use the notation $a^{\sharp}$ to denote the operators $a$ and $a^{*}$. It turns out that:
\begin{equation}\label{eq:wick}
\Big\langle R\Omega, a^{\sharp_{1}}_{{\bf x}_{1}} a^{\sharp_{2}}_{{\bf x}_{2}} \cdots a^{\sharp_{2j}}_{{\bf x}_{2j}}  R\Omega \Big\rangle = \sum_{\pi \in P_{2j}}  \text{sgn}(\pi) \Big\langle R\Omega, a^{\sharp_{\sigma(1)}}_{{\bf x}_{\sigma(1)}} a^{\sharp_{\sigma(2)}}_{{\bf x}_{\sigma(2)}} R \Omega \Big\rangle \cdots \Big\langle R\Omega, a^{\sharp_{\sigma(2j-1)}}_{{\bf x}_{\sigma(2j-1)}} a^{\sharp_{\sigma(2j)}}_{{\bf x}_{\sigma(2j)}} R \Omega \Big\rangle\;,
\end{equation}
where $P_{2j}$ is the set of pairings:
\begin{equation}
P_{2j} = \Big\{ \sigma \in S_{2j} \; \big|\; \sigma(2\ell-1) < \sigma(2\ell + 1)\;,\; \ell=1, \ldots, j-1\;,\; \sigma(2\ell-1) < \sigma(2\ell)\;,\; \ell = 1,\ldots, j \Big\}\;,
\end{equation}
with $S_{2j}$ the set of permutations of $2j$ elements. The identity (\ref{eq:wick}) can be checked starting from the relation (\ref{eq:RaR}), combined with the canonical anticommutation relations and the property $a_{{\bf x}}\Omega = 0$.

We are interested in the many-body evolution of initial data close to quasi-free states, generated by the Hamiltonian (\ref{eq:FocKHN}). It turns out that if one {\it assumes} that the many-body evolution preserves quasi-freeness for all times, one finds the following coupled evolution equations for the one-particle density matrix and the pairing matrix:
\begin{equation}\label{eq:HFB}
\begin{split}
i \varepsilon \partial_t \omega_{N,t} &= \big[ h_{\text{HF}}(t),  \omega_{N,t} \big] + \Pi_{t} \alpha^*_{N,t} - \alpha_{N,t} \Pi_{t} ^*\;, \\
i \varepsilon \partial_t \alpha_{N,t} &= h_{\text{HF}}(t)\alpha_{N,t} + \alpha_{N,t} \overline{h_{\text{HF}}(t)} + \Pi_{t} \big( 1- \overline{\omega}_{N,t} \big) - \omega_{N,t} \Pi_{t}\;,
\end{split}
\end{equation}
where $h_{\text{HF}}(t)$ is the Hartree-Fock Hamiltonian,
\begin{equation}
h_{\text{HF}}(t) = -\varepsilon^{2}\Delta + \rho_{t} * V - X_{t}\;,
\end{equation}
with $\rho_{t}$ the normalized density and $X_{t}, \Pi_{t}$ linear operators on $\frak{h}$,
\begin{equation}
\rho_{t}({\bf x}) = \frac{1}{N} \omega_{N,t}({\bf x};{\bf x})\;,\; X_{t}({\bf x}; {\bf y}) = \frac{1}{N} V(x-y)\omega_{N,t}({\bf x}; {\bf y})\;,\;\Pi_{t}({\bf x},{\bf y}):= \frac{1}{N} V(x-y)\alpha_{N,t}({\bf x},{\bf y})\;.
\end{equation}
See {\it e.g.} \cite{BSS}, for a derivation of (\ref{eq:HFB}) via the Dirac-Frenkel principle, which gives a geometric interpretation to the reduction to the set of quasi-free states. Thus, Eqs. (\ref{eq:HFB}) define the evolution of the generalized one-particle density matrix:
\begin{equation}
\Gamma_{N,t} = \begin{pmatrix} \omega_{N,t} & \alpha_{N,t} \\ \alpha_{N,t}^{*} & 1 - \overline{\omega_{N,t}} \end{pmatrix}\;.
\end{equation}
One can rewrite Eqs. (\ref{eq:HFB}) as:
\begin{equation}\label{eq:HFBdef}
i\varepsilon \partial_{t} \Gamma_{N,t} = [ H(t), \Gamma_{N,t} ]\;,
\end{equation}
with:
\begin{equation}
H(t) = \begin{pmatrix} h_{\text{HF}}(t) & \Pi_{t} \\ \Pi_{t}^{*} & - \overline{h_{\text{HF}}(t)} \end{pmatrix}\;.
\end{equation}
That is, $H(t)$ generates a unitary dynamics for $\Gamma_{N,t}$. Hence, since $\Gamma_{N} = \Gamma_{N}^{2}$, recall Eq. (\ref{eq:square}), we also have $\Gamma_{N,t} = \Gamma_{N,t}^{2}$; in particular, the relations (\ref{eq:consproj}) hold true for $\omega_{N,t}$ and $\alpha_{N,t}$. Eq. (\ref{eq:HFBdef}) is called the Hartree-Fock-Bogoliubov equation, and it is used to describe interesting physical phenomena, such as the time evolution of superconductors. 
The well-posedness of the HFB equation (\ref{eq:HFB}) for a class of potentials that includes those considered in the present work, is well-known; see {\it e.g.} \cite{BSS}. The main goal of the present work will be to give a rigorous derivation of the HFB equation starting from the full many-body quantum dynamics. 

\section{Main result}\label{sec:main}
Before stating our main result, let us introduce the assumptions on the initial data that we will consider in this work.
\begin{assumption}[Assumptions on the initial datum]\label{ass:sc} Let $\Gamma_{N}: \frak{h} \oplus \frak{h} \to \frak{h} \oplus \frak{h}$  be the generalized one-particle density matrix of a pure, quasi-free state:
\begin{equation}
\Gamma_{N} = \begin{pmatrix} \omega_{N} & \alpha_{N} \\ \alpha_{N}^{*} & 1 - \overline{\omega_{N}} \end{pmatrix}\;,\qquad 0 \leq \Gamma_{N} \leq 1\;,\qquad \Gamma_{N} = \Gamma_{N}^{2}\;.
\end{equation}
We shall suppose that:
\begin{equation}\label{eq:assalpha}
\tr\, \omega_{N} = N\;,\qquad \| \alpha_{N} \|_{\text{HS}} \leq CN^{1/3}\;.
\end{equation}
Furthermore, we shall suppose that $\omega_{N}$, $\alpha_{N}$ satisfy the following commutator estimates:
\begin{equation}\label{eq:sc}
\sup_{p\in \mathbb{R}^{3}} \frac{1}{1+|p|} \| [ \omega_{N}, e^{ip\cdot \hat x} ] \|_{\text{HS}} \leq C  N^{1/3}\;,\quad \| [ \omega_{N}, \varepsilon \nabla ] \|_{\text{HS}} \leq C N^{1/3}\;,\quad \| [ \alpha_{N}, \varepsilon \nabla ] \|_{\text{HS}} \leq C N\;.
\end{equation}
\begin{remark}
\begin{itemize}
\item[(i)] We refer the reader to the introduction for a discussion of these assumptions. As we will show in Appendix \ref{app:BCS}, all the above assumptions are satisfied for translation-invariant quasi-free states on a finite torus, at low energy.
\item[(ii)] We shall refer to the commutator estimates in (\ref{eq:sc}) as the \emph{semiclassical structure} of the initial datum.
\end{itemize}
\end{remark}
\end{assumption}
The next theorem is the main result of the paper.
\begin{theorem}[Main result]\label{thm:main} Let $V\in L^{1}(\mathbb{R}^{3})$ such that 
\begin{equation}\label{eq:assV}
V(x) \in \mathbb{R}\;,\qquad V(x) = V(-x)\;,\qquad \int dp\, |\hat V(p)| (1 + |p|^{2}) \leq C\;.
\end{equation}
Let $\omega_{N}$, $\alpha_{N}$ be the one-particle density matrix and the pairing matrix of a quasi-free state satisfying Assumption~\ref{ass:sc}. Let $\psi = R_{0} \xi$ with $R_{0}$ the Bogoliubov transformation associated to $(\omega_{N}, \alpha_{N})$ and $\xi \in \mathcal{F}$, $\|\xi\| = 1$, such that, for some $k\in \mathbb{N}$:
\begin{equation}
\langle \psi, \mathcal{N} \psi \rangle = N\;,\qquad \langle \xi, \mathcal{N}^{\alpha(k)} \xi \rangle \leq C_{k}
\end{equation}
with $\alpha(k) = (13/2)k$ if $k$ is even and $\alpha(k) = (13/2)k + 3/2$ if $k$ is odd. Let $\psi_{t} = e^{-i\mathcal{H}_{N} t /\varepsilon}\psi$ be the solution of the many-body Schr\"odinger equation, and let $R_{t}$ be the Bogoliubov transformation associated with the solution $\omega_{N,t}, \alpha_{N,t}$ of the Hartree-Fock-Bogoliubov equation (\ref{eq:HFB}). Then, for all $t\geq 0$, and for all $j\leq k$:
\begin{equation}\label{eq:main}
\langle \psi_{t}, a^{\sharp_{1}}_{{\bf x}_{1}} \ldots a^{\sharp_{2j}}_{{\bf x}_{2j}} \psi_{t} \rangle = \langle R_{t} \Omega, a^{\sharp_{1}}_{{\bf x}_{1}} \ldots a^{\sharp_{2j}}_{{\bf x}_{2j}} R_{t}\Omega \rangle + E_{2j;t}^{\sharp_{1}, \ldots, \sharp_{2j}}({\bf x}_{1}, \ldots, {\bf x}_{2j})
\end{equation}
where $a^{\sharp}_{{\bf x}}$ can be either $a_{{\bf x}}$ or $a^{*}_{{\bf x}}$, $a^{\sharp_{1}}_{{\bf x}_{1}} \ldots a^{\sharp_{2j}}_{{\bf x}_{2j}}$ is a normal-ordered monomial, and
\begin{equation}\label{eq:HSrem}
\int d{\bf x}_{1}\ldots d{\bf x}_{2j}\, \big| E_{2j;t}^{\sharp_{1}, \ldots, \sharp_{2j}}({\bf x}_{1}, \ldots, {\bf x}_{2j}) \big|^{2} \leq C(j) \exp(d(j)\exp c|t|)  N^{j-2/3 - \max( 1/3,\, p/6 )}\;,
\end{equation}
where the constants $C(j)$, $d(j)$ depend on $j$ and are independent of $N$ and of $t$, while $p$ is equal to the absolute value of the difference of the number of creation and annihilation operators in the left-hand side of (\ref{eq:main}).
\end{theorem}
\begin{remark}
\begin{itemize}
\item[(i)] Theorem \ref{thm:main} proves convergence of the many-body Schr\"odinger equation to the Hartree-Fock-Bogoliubov equation, in the sense of closeness of density matrices, in the Hilbert-Schmidt topology. In fact, the left-hand side of (\ref{eq:HSrem}) is the Hilbert-Schmidt norm squared of the operator $E_{2j;t}^{\sharp_{1}, \ldots, \sharp_{2j}}$ with integral kernel $E_{2j;t}^{\sharp_{1}, \ldots, \sharp_{2j}}({\bf x}_{1}, \ldots, {\bf x}_{2j})$. For the particular cases of the one-particle density matrix ($j=1,p=0$) and of the pairing matrix ($j=1,p=2$), the theorem gives:
\begin{equation}
\begin{split}
\| \gamma_\psi^{(1)} - \omega_N \|_{\mathrm{HS}} &\leq C \exp(d \exp(ct))\;, \\
 \| \pi_\psi^{(1)} - \alpha_N \|_{\mathrm{HS}} &\leq C \exp(d \exp(ct))\;.
\end{split}
\end{equation}
for some constants $C, d, c$ independent of $N$. These bounds have to be compared with the trivial estimates $\| \omega_{N} \|_{\text{HS}} \leq N^{1/2}$ and $\| \alpha_{N} \|_{\text{HS}} \leq CN^{1/3}$.

\item[(ii)] More generally, the Hilbert-Schmidt norm squared of the first term in the right-hand side of (\ref{eq:main}) is much larger than the right-hand side of (\ref{eq:HSrem}) for all times of order $1$. This can be seen writing it in terms of the one-particle density matrix using the Wick rule.  For instance:
\begin{equation}
\begin{split}
\langle R_{t} \Omega, a^{*}_{{\bf x}_{1}} a^{*}_{{\bf x}_{2}} a_{{\bf x}_{3}} a_{{\bf x}_{4}} R_{t}\Omega \rangle &= \omega_{N,t}({\bf x}_{3}; {\bf x}_{2}) \omega_{N,t}({\bf x}_{4}; {\bf x}_{1}) - \omega_{N,t}({\bf x}_{3}; {\bf x}_{1}) \omega_{N,t}({\bf x}_{4} ; {\bf x}_{2}) \\ 
&\quad + \overline{\alpha_{N,t}({\bf x}_{1}; {\bf x}_{2})} \alpha_{N,t}({\bf x}_{4}; {\bf x}_{3})\;.
\end{split}
\end{equation}
If the bounds on $\| \omega_{N} \|_{\text{HS}}$ and $\| \alpha_{N} \|_{\text{HS}}$ are saturated, namely $\|\alpha_{N}\|^{2}_{\text{HS}} \sim N^{2/3}$ and $\| \omega_{N} \|_{\text{HS}}^{2} \sim N$, they remain so even at later times $t$. Actually, since $\|\omega_{N}\|_{\text{HS}}^{2} = N - \| \alpha_{N} \|_{\text{HS}}^{2}$, it is enough to show that $\|\alpha_{N,t}\|^{2}$ stays approximately constant along the HFB flow; this can be done following the proof of Proposition \ref{prp:sc}. Therefore, the largest contribution to the two-particle density matrix of the quasi-free state comes from the terms with the largest number of $\omega_{N,t}$.  For gauge-invariant correlation functions,  {\it i.e.} with an equal number of creation and annihilation operators, the Hilbert-Schmidt norm of the largest contribution is of order $N^j$.  In general,  if $m$ is the number of annihilation operators,  say smaller than the number of creation operators $m+p$,  the largest term has $m=j-p/2$ factors $\omega_{N,t}$ and $p/2$ factors $\alpha_{N,t}$. Thus, its Hilbert-Schmidt norm squared is of order $N^{j-p/6}$, which is much larger than the estimate for the error term in (\ref{eq:HSrem}).
\end{itemize}
\end{remark}
\begin{remark}[Resolution of the pairing terms] It is interesting to compare the error term with the smallest contribution to the quasi-free approximation. If the error is subleading, we say that the quasi-free approximation can be fully resolved. Let us focus on gauge-invariant correlation functions ($p=0$).  The smallest Wick pairing in Hilbert-Schmidt norm is the one that has the largest number of contractions of ``$a a$'' or ``$a^{*} a^{*}$'' monomials, which correspond to kernels of $\alpha_{N,t}$ operators (or its complex conjugate). The number of such contractions is less or equal than $j-1$ if the number of ``$a$'' operators is odd, and it is less or equal than $j$ if the number of ``$a$'' operators is even.  In the $p=0$ case,  the number of ``$a$'' operators is exactly $j$, therefore the Hilbert-Schmidt norm squared of the Wick pairing with the largest number of $\alpha$-contractions is:
\begin{equation}
\begin{split}
(N^{2/3})^{j-1} \times N\qquad &\text{if $ j$ is odd} \\
(N^{2/3})^{j}\qquad &\text{if $j$ is even.}
\end{split}
\end{equation}
Comparing these estimates with the bound (\ref{eq:HSrem}) for $p=0$, we have that, for $N \gg 1$:
\begin{equation}
\begin{split}
(N^{2/3})^{j-1} \times N \gg \|E_{2j;t}^{\sharp_{1}, \ldots, \sharp_{2j}}\|_{\text{HS}}^{2} &\qquad \text{if $j = 1,3$} \\
(N^{2/3})^{j} \gg \|E_{2j;t}^{\sharp_{1}, \ldots, \sharp_{2j}}\|_{\text{HS}}^{2} &\qquad \text{if $j = 2$.}
\end{split}
\end{equation}
Therefore, all Wick pairings are individually much bigger than the error term, and hence they can be resolved, for the $j$-particle density matrices up to $j=3$.  More generally, it is clear that if $p>0$ the bound (\ref{eq:HSrem}) also allows to fully resolve correlations of higher order. 
\end{remark}
\begin{remark}[Comparison with the Hartree-Bogoliubov dynamics.] As mentioned in the Introduction, the proof of Theorem \ref{thm:main} could be simplified if one replaces in the assumptions (\ref{eq:assalpha}), (\ref{eq:sc}) the Hilbert-Schmidt norm with the square root of the trace norm. Furthermore, as for the case of $N$-particle states \cite{BPS}, under these stronger assumptions it is possible to prove that neglecting the exchange operators $X_{t}$, $\Pi_{t}$ in the Hartree-Fock-Bogoliubov equation (\ref{eq:HFB}) does not deteriorate the error estimate (\ref{eq:HSrem}) in the main result. Thus, for this class of initial data, we cannot distinguish the Hartree-Fock-Bogoliubov dynamics from the Hartree-Bogoliubov dynamics. We prefer to work under weaker assumptions that we can verify in concrete cases, see Appendix \ref{app:BCS}. For this class of initial data, we do not know how to prove the indistinguishability of the Hartree-Fock-Bogoliubov and Hartree-Bogoliubov dynamics.
\end{remark}
The rest of the paper is devoted to the proof of Theorem \ref{thm:main}.

\section{The fluctuation dynamics}\label{sec:fluct}

Suppose that $(\omega_{N,t}, \alpha_{N,t})$ solves the time-dependent Hartree-Fock-Bogoliubov equation (\ref{eq:HFB}). Let us introduce the fluctuation dynamics as:
\begin{equation}
\mathcal{U}_{N}(t;s) := R^{*}_{t} e^{-i\mathcal{H}_{N} (t - s) / \varepsilon} R_{s}\;.
\end{equation}
The main goal of this section will be to prove an estimate for the evolution of powers of the number operators $\mathcal{N}$ under the fluctuation dynamics. This will be the key technical tool behind the proof of Theorem~\ref{thm:main}.

\subsection{Growth of fluctuations}

The starting points is the following proposition.
\begin{proposition}[Growth of fluctuations]\label{prp:growth} Under the same assumptions of Theorem~\ref{thm:main}, the following is true. For any $\xi \in \mathcal{F}$, $\|\xi\| = 1$, and $k\in \mathbb{N}$,
\begin{equation}\label{eq:growthN}
\begin{split}
i \varepsilon \frac{d}{dt} &\Big \langle \mathcal{U}_{N}(t;0) \xi,  (\mathcal{N} +1)^k \mathcal{U}_{N}(t;0) \xi \Big\rangle =  - \frac{4i}{N}  \sum_{j=1}^k \, \mathrm{Im} \int d{\bf x}d {\bf y} \, V(x-y) \\
& \quad \times \Big[ \Big\langle \mathcal{U}_{N}(t;0) \xi,  (\mathcal{N} +1)^{j-1} a^*(u_{t,{\bf x}}) a(\overline{v}_{t,{\bf y}}) a(u_{t,{\bf y}})a(u_{t,{\bf x}}) (\mathcal{N} +1)^{k-j} \mathcal{U}_{N}(t;0) \xi \Big\rangle  \\
&\quad\quad + \Big\langle \mathcal{U}_{N}(t;0) \xi,  (\mathcal{N} +1)^{j-1} a(\overline{v}_{t,{\bf x}}) a(\overline{v}_{t,{\bf y}}) a(u_{t,{\bf y}})a(u_{t,{\bf x}}) (\mathcal{N} +1)^{k-j} \mathcal{U}_{N}(t;0) \xi \Big\rangle  \\
&\quad\quad + \Big\langle \mathcal{U}_{N}(t;0) \xi,  (\mathcal{N} +1)^{j-1} a^*(u_{t,{\bf y}}) a^*(\overline{v}_{t,{\bf y}}) a^*(\overline{v}_{t,{\bf x}}) a(\overline{v}_{t,{\bf x}}) (\mathcal{N} +1)^{k-j} \mathcal{U}_{N}(t;0) \xi \Big\rangle  \Big].
\end{split}
\end{equation}
\end{proposition}
\begin{remark} The same identity holds true for the Hartree-Fock case \cite{BPS}, see Proposition~3.3 there. There, the operators $u_{t}$ and $\overline{v}_{t}$ are orthogonal, in the sense that $u_{t} \overline{v}_{t} = 0$. Here, $u_{t} \overline{v}_{t} \neq 0$, in general.
\end{remark}
\begin{proof} To begin, we write the following identity that can be checked using the explicit action of the Bogoliubov transformation (\ref{eq:RaR2}) and the fact that $\tr\, \omega_{N,t} = N$:
\begin{equation}\label{eq:numt}
R_{t} \mathcal{N} R^{*}_{t} = \mathcal{N} - 2 d\Gamma ( \omega_{N,t}) + N + \int  d{\bf x} d {\bf y} \left( a^*_{\bf x} a^*_{\bf y} \alpha_{N,t} ({\bf y},{\bf x}) + a_{\bf x} a_{\bf y} \overline{\alpha}_{N,t} ({\bf x},{\bf y})   \right)\;.
\end{equation}
Therefore, since $\mathcal{U}_{N}(t;s) = R^{*}_{t} e^{-i\mathcal{H}_{N} (t-s) / \varepsilon} R_{s}$, we have:
\begin{equation}
\begin{split}
i \varepsilon \frac{d}{dt} &\Big( \mathcal{U}_{N}(t;0)^{*} \mathcal{N} \mathcal{U}_{N}(t;0)\Big) = i \varepsilon \frac{d}{dt}  \Big(R_{0}^{*} e^{i\mathcal{H}_{N} t / \varepsilon} R_{t} \mathcal{N} R_{t}^{*} e^{-i\mathcal{H}_{N} t / \varepsilon} R_{0}\Big) \\
&= i \varepsilon \frac{d}{dt}  \Big( R_{0}^{*} e^{i\mathcal{H}_{N} t / \varepsilon} \Big( - 2d\Gamma ( \omega_{N,t}) + \Big( \int  d{\bf x} d {\bf y} a^*_{\bf x} a^*_{\bf y} \alpha_{N,t} ({\bf y},{\bf x}) + \text{h.c.}\Big)\Big) e^{-i\mathcal{H}_{N} t / \varepsilon} R_{0}\Big)  \\
& = R_{0}^{*} e^{i\mathcal{H}_{N} t / \varepsilon} \Big( - 2d\Gamma ( i \varepsilon \partial_{t} \omega_{N,t}) + \Big(\int  d{\bf x} d {\bf y} a^*_{\bf x} a^*_{\bf y} (i\varepsilon \partial_{t} \alpha_{N,t}) ({\bf y},{\bf x}) - \text{h.c.}\Big)\Big) e^{-i\mathcal{H}_{N} t / \varepsilon} R_{0} \\
& \quad + R_{0}^{*} e^{i\mathcal{H}_{N} t / \varepsilon} \Big( 2\Big[\mathcal{H}_{N}, d\Gamma ( \omega_{N,t})\Big] - \Big(\int  d{\bf x} d {\bf y} \Big[ \mathcal{H}_{N}, a^*_{\bf x} a^*_{\bf y}\Big] \alpha_{N,t} ({\bf y},{\bf x}) - \text{h.c.}\Big)\Big) e^{-i\mathcal{H}_{N} t / \varepsilon} R_{0}
\end{split}
\end{equation}
where we used that the first and the third term in (\ref{eq:numt}) do not contribute to the time derivative and the notation $\text{h.c.}$ stands for hermitian conjugate. Then, by unitarity of the Bogoliubov transformation.
\begin{equation}\label{eq:ddt}
\begin{split}
&i \varepsilon \frac{d}{dt} \Big(\mathcal{U}_{N}(t;0)^{*} \mathcal{N} \mathcal{U}_{N}(t;0)\Big) = 2\mathcal{U}_{N}(t;0)^{*} R^{*}_{t}\Big( \Big[\mathcal{H}_{N}, d\Gamma ( \omega_{N,t})\Big] - d\Gamma ( i \varepsilon \partial_{t} \omega_{N,t}) \Big) R_{t}\, \mathcal{U}_{N}(t;0) \\
&+ \mathcal{U}_{N}(t;0)^{*} R^{*}_{t}\Big( \int  d{\bf x} d {\bf y} \Big( a^*_{\bf x} a^*_{\bf y} (i\varepsilon \partial_{t} \alpha_{N,t}) ({\bf y},{\bf x}) - \Big[ \mathcal{H}_{N}, a^*_{\bf x} a^*_{\bf y}\Big] \alpha_{N,t} ({\bf y},{\bf x})\Big) - \text{h.c.} \Big) R_{t}\, \mathcal{U}_{N}(t;0)\;.
\end{split}
\end{equation}
In this expression, the terms involving the Laplacian cancel out exactly. Next, let us consider the commutators involving the many-body interaction. Denoting by $\mathcal{V}_{N}$ the part of $\mathcal{H}_{N}$ describing the potential, we compute:
\begin{equation}\label{eq:quartic_gamma}
2\int d{\bf x}d {\bf y} \,\Big[ \mathcal{V}_{N} ,  a^*_{\bf x} a_{\bf y} \Big]  \omega_{N,t}({\bf x},{\bf y})= \frac{2}{N}  \int d{\bf x}d{\bf y} d {\bf z}\, V(x-y) \omega_{N,t}({\bf z},{\bf y})\, a^*_{\bf z} a^*_{\bf x} a_{\bf y} a_{\bf x} - \mathrm{h.c.} , 
\end{equation}
and
\begin{equation}\label{eq:quartic_alpha}
\begin{split}
\int d{\bf x} d{\bf y} \,\Big[ \mathcal{V}_{N},  a^*_{\bf x} a^*_{\bf y} \Big]   \alpha_{N,t} ({\bf y}, {\bf x}) &= \frac{1}{N} \int d{\bf x} d {\bf y}\, V(x-y) \alpha_{N,t} ({\bf y}, {\bf x})\, a^*_{\bf x} a^*_{\bf y} \\&\quad + \frac{2}{N} \int d {\bf x} d {\bf y} d{\bf z}\, V(x-y) \alpha_{N,t} ({\bf z},{\bf y})\, a^*_{\bf x} a^*_{\bf y} a^*_{\bf z} a_{\bf x}\;.
\end{split}
\end{equation}
Therefore,  focusing on the following term,
\begin{equation}
\text{A} := R^{*}_{t} \Big( 2\Big[\mathcal{V}_{N}, d\Gamma ( \omega_{N,t})\Big] - \Big( \int  d{\bf x} d {\bf y} \Big[ \mathcal{V}_{N}, a^*_{\bf x} a^*_{\bf y}\Big] \alpha_{N,t} ({\bf y},{\bf x}) - \text{h.c.}\Big)\Big) R_{t}\;,
\end{equation}
by Eqs. (\ref{eq:quartic_gamma}), (\ref{eq:quartic_alpha}), we have:
\begin{equation}\label{eq:A}
\begin{split}
\text{A} &= \frac{2}{N} \int d {\bf x} d {\bf y} \,V(x-y) R^*_t \int d {\bf z} \Big( \omega_{N,t}({\bf z},{\bf y}) a^*_{\bf z} a^*_{\bf x} a_{\bf y} a_{\bf x} - \alpha_{N,t}({\bf z},{\bf y}) a^*_{\bf x} a^*_{\bf y} a^*_{\bf z} a_{\bf x} \Big) R_t \\
&\quad - \frac{1}{N} \int d {\bf x} d {\bf y} \,V(x-y) R^*_t \Big( \alpha_{N,t} ({\bf y},{\bf x}) a^*_{\bf x} a^*_{\bf y} \Big) R_t - \text{h.c.}.
\end{split}
\end{equation}
Let us denote by $\text{A}_2$ and $\text{A}_4$ the quadratic and the quartic terms in Eq. (\ref{eq:A}),
\begin{equation}\label{eq:A2A4}
\begin{split}
\text{A}_2 &:= -\frac{1}{N}  \int d{\bf x} d{\bf y} \, V(x-y) \alpha_{N,t}({\bf y}, {\bf x}) R^*_t a^*_{\bf x} a^*_{\bf y} R_t - \text{h.c.} \\
\text{A}_4 &:= \frac{2}{N} \int d {\bf x} d {\bf y} \,V(x-y) R^*_t \int d{\bf z} \Big( \omega_{N,t}({\bf z},{\bf y}) a^*_{\bf z} a^*_{\bf x} a_{\bf y} a_{\bf x} - \alpha_{N,t}({\bf z}, {\bf y}) a^*_{\bf x} a^*_{\bf y} a^*_{\bf z} a_{\bf x}\Big) R_t - \text{h.c.}.
\end{split}
\end{equation}
Consider $\text{A}_{4}$. We rewrite it as:
\begin{equation}\label{eq:Q4}
\begin{split}
\text{A}_4 &= \frac{2}{N} \int d{\bf x} d{\bf y} \,V(x-y) R^*_t a^*_{\bf x} \int d {\bf z} \, \Big(-\omega_{N,t}({\bf z},{\bf y}) a^*_{\bf z} + \overline{\alpha}_{N,t}({\bf z},{\bf y}) a_{\bf z} \Big) a_{\bf y} a_{\bf x} R_t - \text{h.c.} \nonumber \\
&= \frac{2}{N} \int d{\bf x} d{\bf y} \,V(x-y) R^*_t a^*_{\bf x} R_t \\ &\cdot \int d{\bf z} \, \Big[-\omega_{N,t}({\bf z},{\bf y}) \big( a^*(u_{t,{\bf z}}) + a(\overline{v}_{t,{\bf z}}) \big) + \overline{\alpha}_{N,t}({\bf z},{\bf y}) \big( a(u_{t,{\bf z}}) + a^*(\overline{v}_{t,{\bf z}}) \big) \Big] R^*_t a_{\bf y} a_{\bf x} R_t - \text{h.c.}.
\end{split}
\end{equation}
Performing the ${\bf z}$ integration we obtain:
\begin{equation}\label{eq:cance}
\begin{split}
&\int d{\bf z} \, \Big[-\omega_{N,t}({\bf z},{\bf y}) \Big( a^*(u_{t,{\bf z}}) + a(\overline{v}_{t,{\bf z}}) \Big) + \overline{\alpha}_{N,t}({\bf z},{\bf y}) \Big( a(u_{t,{\bf z}}) + a^*(\overline{v}_{t,{\bf z}}) \Big) \Big] \\
&= \int d{\bf z} \int d {\bf r} \Big[ a^*_{\bf r}  \Big(-u_t({\bf r},{\bf z}) \omega_{N,t}({\bf z},{\bf y}) + \overline{v}_t({\bf r},{\bf z}) \overline{\alpha}_{N,t}({\bf z},{\bf y}) \Big) \\&\qquad + a_{\bf r} \Big(-v_t ({\bf r}, {\bf z}) \omega_{N,t}({\bf z},{\bf y}) + \overline{u}_t({\bf r},{\bf z}) \overline{\alpha}_{N,t}({\bf z}, {\bf y}) \Big) \Big] \\
&= \int d {\bf r} \Big( a^*_{\bf r} \big( \overline{v}_t \overline{\alpha}_{N,t} - u_t \omega_{N,t} \big)({\bf r},{\bf y}) + a_{\bf r} \big( \overline{u}_t \overline{\alpha}_{N,t} - v_t \omega_{N,t} \big)({\bf r},{\bf y}) \Big) \\
&= - a(\overline{v}_{t,{\bf y}})\;.
\end{split}
\end{equation}
To prove the last identity we used that:
\begin{equation}\label{eq:simp1}
\overline{v}_t \overline{\alpha}_{N,t} - u_t \omega_{N,t} = \overline{v}_t v^T_t u_t  - u_t v^*_t v_t=    \overline{v}_t v^T_t u_t + \overline{v}_t u^T_t v_{t} = \overline{v}_t (v^T_t u_t + u^T_t v_t) =0\;,
\end{equation} 
where: the first step follows from the representation of $\omega_{N,t}$ and $\alpha_{N,t}$ in terms of $u_{t}$ and $v_{t}$, Eq. (\ref{eq:omegaalpha}); the second step follows from the second property in the second line of (\ref{eq:bogo2}); and the third step follows from the second property in the first line of (\ref{eq:bogo2}), after taking a complex conjugate. Eq. (\ref{eq:simp1}) proves the vanishing of the creation contribution to the right-hand side of (\ref{eq:cance}). To obtain the last identity in (\ref{eq:cance}), we used:
\begin{equation}\label{eq:simp2}
\overline{u}_t \overline{\alpha}_{N,t} - v_t \omega_{N,t} = \overline{u}_t v^T_t u_t - v_t v^*_t v_t = - v_t u^*_t u_t - v_t v^*_t v_t =-v_t (u^*_t u_t + v^*_t v_t) = -v_t\;.
\end{equation}
The first step follows from the representation of $\omega_{N,t}$ and $\alpha_{N,t}$ in terms of $u_{t}$ and $v_{t}$; the second step follows from the second property in the second equation in (\ref{eq:bogo2}), after taking a complex conjugate; the last step follows from the first property in the first equation of (\ref{eq:bogo2}). This concludes the check of (\ref{eq:cance}).

Let us now go back to $\text{A}_{4}$ in (\ref{eq:Q4}). Using (\ref{eq:cance}), we have:
\begin{equation}
\begin{split}
\text{A}_{4} &= -\frac{2}{N} \int d{\bf x} d{\bf y} \,V(x-y) R^*_t a^*_{\bf x} R_t a(\overline{v}_{t,{\bf y}}) R^*_t a_{\bf y} a_{\bf x} R_t - \text{h.c.} \\
&= -\frac{2}{N} \int d{\bf x} d{\bf y} \,V(x-y) ( a^{*}(u_{t,{\bf x}}) + a(\overline{v}_{t,{\bf x}})) a(\overline{v}_{t,{\bf y}}) ( a(u_{t,{\bf y}}) + a^{*}(\overline{v}_{t,{\bf y}}) ) ( a(u_{t,{\bf x}}) + a^{*}(\overline{v}_{t,{\bf x}}) )\\&\quad - \text{h.c.},
\end{split}
\end{equation}
and expanding the products:
\begin{equation}\label{eq:Q4rew}
\begin{split}
\text{A}_{4} &= -\frac{2}{N} \int d{\bf x} d{\bf y} \,V(x-y) \Big[ a^*(u_{t,{\bf x}}) a(\overline{v}_{t,{\bf y}}) a^*(\overline{v}_{t,{\bf y}}) a^*(\overline{v}_{t,{\bf x}}) + a^*(u_{t,{\bf x}}) a(\overline{v}_{t, {\bf y}}) a^*(\overline{v}_{t, {\bf y}}) a(u_{t,{\bf x}}) \\& \quad + a^*(u_{t, {\bf x}}) a(\overline{v}_{t, {\bf y}}) a(u_{t,{\bf y}})a(u_{t, {\bf x}})   + a^*(u_{t, {\bf x}}) a(\overline{v}_{t, {\bf y}}) a(u_{t, {\bf y}}) a^*(\overline{v}_{t, {\bf x}}) \\ &\quad + a(\overline{v}_{t, {\bf x}}) a(\overline{v}_{t, {\bf y}})  a^*(\overline{v}_{t, {\bf y}}) a^*(\overline{v}_{t, {\bf x}}) + a(\overline{v}_{t, {\bf x}}) a(\overline{v}_{t, {\bf y}}) a^*(\overline{v}_{t, {\bf y}}) a(u_{t, {\bf x}}) \\
&\quad +  a(\overline{v}_{t, {\bf x}}) a(\overline{v}_{t, {\bf y}}) a(u_{t, {\bf y}})a(u_{t, {\bf x}}) + a(\overline{v}_{t, {\bf x}}) a(\overline{v}_{t, {\bf y}}) a(u_{t, {\bf y}}) a^*(\overline{v}_{t, {\bf x}}) \Big] - \text{h.c.}.
\end{split}
\end{equation}
The second and the fifth term in the right-hand side of (\ref{eq:Q4rew}) are self-adjoint, and hence cancel with minus the hermitian conjugate.  The sum of the first and the sixth term is also self-adjoint (each one is the adjoint of the other),  so they cancel as well when subtracting the hermitian conjugate.  In order to proceed,  we put the remaining terms in the right-hand side into normal order. We get:
\begin{equation}\label{eq:Q4norm}
\begin{split}
\text{A}_4 &= -\frac{2}{N}\int d{\bf x} d{\bf y} \,V(x-y) \Big[  a^*(u_{t,{\bf x}}) a(\overline{v}_{t, {\bf y}}) a(u_{t, {\bf y}})a(u_{t, {\bf x}}) + a^*(u_{t, {\bf x}}) a^*(\overline{v}_{t, {\bf x}}) a(\overline{v}_{t, {\bf y}}) a(u_{t, {\bf y}}) \\&\quad + a^*(u_{t, {\bf x}}) a(\overline{v}_{t, {\bf y}}) \alpha_{N,t}({\bf x},{\bf y}) - a^*(u_{t,{\bf x}}) a(u_{t, {\bf y}}) \omega_{N,t}({\bf x},{\bf y}) + a(\overline{v}_{t,{\bf x}}) a(\overline{v}_{t,{\bf y}}) a(u_{t,{\bf y}})a(u_{t,{\bf x}}) \\&\quad - a^*(\overline{v}_{t,{\bf x}}) a(\overline{v}_{t,{\bf x}}) a(\overline{v}_{t,{\bf y}}) a(u_{t,{\bf y}}) + a(\overline{v}_{t,{\bf x}}) a(\overline{v}_{t,{\bf y}}) \alpha_{N,t}({\bf x},{\bf y}) - a(\overline{v}_{t,{\bf x}}) a(u_{t,{\bf y}}) \omega_{N,t}({\bf x},{\bf y}) \\&\quad + a(\overline{v}_{t,{\bf y}}) a(u_{t,{\bf y}}) \omega_{N,t}({\bf x},{\bf x}) \Big] - \text{h.c.}.
\end{split}
\end{equation}
Using that $V(x-y) = V(y-x)$, the second and the fourth term are self-adjoint, and hence they cancel with minus their hermitian conjugates. Then, we rewrite (\ref{eq:Q4norm}) as:
\begin{equation}\label{eq:Q4norm2}
\begin{split}
\text{A}_4 &= -\frac{2}{N}\int d{\bf x} d{\bf y} \,V(x-y) \Big[  a^*(u_{t,{\bf x}}) a(\overline{v}_{t, {\bf y}}) a(u_{t, {\bf y}})a(u_{t, {\bf x}}) + a(\overline{v}_{t,{\bf x}}) a(\overline{v}_{t,{\bf y}}) a(u_{t,{\bf y}})a(u_{t,{\bf x}}) \\
&\qquad - a^*(\overline{v}_{t,{\bf x}}) a(\overline{v}_{t,{\bf x}}) a(\overline{v}_{t,{\bf y}}) a(u_{t,{\bf y}})\Big] \\
&\qquad - 2\int d{\bf x} d{\bf y}\, (a^*(u_{t, {\bf x}}) + a(\overline{v}_{t,{\bf x}})) \Pi_{t}({\bf x}; {\bf y}) a(\overline{v}_{t, {\bf y}}) + 2\int d{\bf x} d{\bf y}\, a(\overline{v}_{t,{\bf x}}) X_{t}({\bf x},{\bf y}) a(u_{t,{\bf y}}) \\
&\qquad - 2\int d{\bf y}\, a(\overline{v}_{t,{\bf y}}) ( \rho_{t} * V )({\bf y}) a(u_{t,{\bf y}}) - \text{h.c.}.
\end{split}
\end{equation}
We claim that the quadratic terms appearing in (\ref{eq:Q4norm2}) cancel out with $\text{A}_{2}$ and with the quadratic terms in Eq. (\ref{eq:ddt}). Consider the quadratic terms in Eq. (\ref{eq:ddt}), minus the kinetic term. The sum of these terms, plus the term $\text{A}_{2}$ in Eq. (\ref{eq:A2A4}), gives:
\begin{equation}\label{eq:quadratic}
- R^{*}_{t} \Big( d\Gamma ( C_{t} ) + \int  d{\bf x} d {\bf y}\, a^*_{\bf x} a^*_{\bf y} D_{t}({\bf x}; {\bf y})  \Big)R_{t} - \text{h.c.}
\end{equation}
where the operators $C_{t}$ and $D_{t}$ are:
\begin{equation}\label{eq:CD}
\begin{split}
C_{t} &= \Big[ \rho_{t} * V - X_{t}\,,  \omega_{N,t} \Big] + \Pi_t \alpha^*_{N,t} - \alpha_{N,t} \Pi_t^* \\
D_{t} &= \big(\rho_{t} * V - X_{t}\big) \alpha_{N,t} + \alpha_{N,t} \big( \rho_{t} * V - \overline{X_{t}} \big) - \Pi_t \overline{\omega}_{N,t} - \omega_{N,t} \Pi_t\;;
\end{split}
\end{equation}
these operators take into account the interaction part of the equations (\ref{eq:HFB}), after subtracting the linear term in $\alpha_{N,t}$ (thanks to $\text{A}_{2}$). Let us now plug these expressions in (\ref{eq:quadratic}). Consider the contribution to (\ref{eq:quadratic}) due to $\rho_{t}*V$. We have:
\begin{equation}\label{eq:rhoV}
\begin{split}
& - R^{*}_{t} \Big( \int d{\bf x} d{\bf y}\, a^{*}_{{\bf x}} \Big( (\rho_{t} * V)({\bf x}) - (\rho_{t} * V)({\bf y})\Big) \omega_{N,t}({\bf x}; {\bf y})   a_{{\bf y}} \\& \qquad + \int d{\bf x} d{\bf y}\, a^{*}_{{\bf x}} \Big( (\rho_{t} * V)({\bf x}) + (\rho_{t} * V)({\bf y})\Big) \alpha_{N,t}({\bf x}; {\bf y}) a^{*}_{{\bf y}} \Big) R_{t} - \text{h.c.} \\
&= - 2 R^{*}_{t} \Big( \int d{\bf x} d{\bf y}\, a^{*}_{{\bf x}} (\rho_{t} * V)({\bf x}) \omega_{N,t}({\bf x}; {\bf y}) a_{{\bf y}} + \int d{\bf x} d{\bf y}\, a^{*}_{{\bf x}} (\rho_{t} * V)({\bf x}) \alpha_{N,t}({\bf x}; {\bf y}) a^{*}_{{\bf y}} \Big) R_{t} - \text{h.c.}
\end{split}
\end{equation}
Acting with the Bogoliubov transformation, we get:
\begin{equation}
\begin{split}
(\ref{eq:rhoV}) &= - 2 \int d{\bf x} d{\bf y}\, (a^{*}(u_{t,{\bf x}}) + a(\overline{v}_{t,{\bf x}})) (\rho_{t} * V)({\bf x}) \omega_{N,t}({\bf x}; {\bf y}) (a(u_{t,{\bf y}}) + a^{*}(\overline{v}_{t,{\bf y}}))\\& \quad - 2\int d{\bf x} d{\bf y}\, ( a^{*}(u_{t,{\bf x}}) + a(\overline{v}_{t,{\bf x}}) ) (\rho_{t} * V)({\bf x}) \alpha_{N,t}({\bf x}; {\bf y}) ( a^{*}(u_{t,{\bf y}}) + a(\overline{v}_{t,{\bf y}}) ) - \text{h.c.} \\
& = - 2 \int d{\bf x} \, (a^{*}(u_{t,{\bf x}}) + a(\overline{v}_{t,{\bf x}})) (\rho_{t} * V)({\bf x}) \big(a((u_{t} \omega_{t})_{{\bf x}}) + a^{*}((\overline{v}_{t} \overline{\omega_{t}})_{{\bf x}}) \big)\\& \quad + 2\int d{\bf x}\, ( a^{*}(u_{t,{\bf x}}) + a(\overline{v}_{t,{\bf x}}) ) (\rho_{t} * V)({\bf x}) \big( a^{*}((u_{t} \alpha_{t})_{{\bf x}}) + a((\overline{v}_{t} \overline{\alpha_{t}})_{{\bf x}}) \big) - \text{h.c.}. \\
\end{split}
\end{equation}
Using the relations (\ref{eq:simp1}), (\ref{eq:simp2}) we obtain:
\begin{equation}\label{eq:rhoV2}
\begin{split}
(\ref{eq:rhoV}) &= - 2 \int d{\bf x} \, (a^{*}(u_{t,{\bf x}}) + a(\overline{v}_{t,{\bf x}})) (\rho_{t} * V)({\bf x}) a^{*}(\overline{v}_{t, {\bf x}}) - \text{h.c.} \\
&= - 2 \int d{\bf x} \, a^{*}(u_{t,{\bf x}}) (\rho_{t} * V)({\bf x}) a^{*}(\overline{v}_{t, {\bf x}}) - \text{h.c.}\;.
\end{split}
\end{equation}
Consider now the terms in (\ref{eq:CD}) corresponding to $X_{t}$. Their contribution to (\ref{eq:quadratic}) is:
\begin{equation}\label{eq:Xt}
\begin{split}
&R^{*}_{t} \Big( d\Gamma( X_{t} \omega_{t} ) + \int d{\bf x} d{\bf y}\, a^{*}_{{\bf x}} a^{*}_{{\bf y}} ( X_{t} \alpha_{t} )({\bf x}; {\bf y}) \Big) R_{t} \\
&\quad - R^{*}_{t} \Big( d\Gamma( \omega_{t} X_{t} ) - \int d{\bf x} d{\bf y}\, a^{*}_{{\bf x}} a^{*}_{{\bf y}} ( \alpha_{t} \overline{X_{t}})({\bf x}; {\bf y}) \Big) R_{t}  - \text{h.c.} \\
& \equiv 2R^{*}_{t} \Big( d\Gamma( X_{t} \omega_{t} ) + \int d{\bf x} d{\bf y}\, a^{*}_{{\bf x}} a^{*}_{{\bf y}} ( X_{t} \alpha_{t} )({\bf x}; {\bf y}) \Big) R_{t} - \text{h.c.}\;,
\end{split}
\end{equation}
where we used the antisymmetry of $\alpha_{N,t}$. We have:
\begin{equation}
\begin{split}
&2 R^{*}_{t} \Big( d\Gamma( X_{t} \omega_{t} ) + \int d{\bf x} d{\bf y}\, a^{*}_{{\bf x}} a^{*}_{{\bf y}} ( X_{t} \alpha_{t} )({\bf x}; {\bf y}) \Big) R_{t} \\
&= 2\int d{\bf x} d{\bf y}\, ( a^{*}(u_{t,{\bf x}}) + a(\overline{v}_{t,{\bf x}}) ) (X_{t} \omega_{N,t})({\bf x}; {\bf y}) ( a(u_{t,{\bf y}}) + a^{*}(\overline{v}_{t,{\bf y}}) ) \\
&\quad + 2\int d{\bf x} d{\bf y}\, ( a^{*}(u_{t,{\bf x}}) + a(\overline{v}_{t,{\bf x}}) ) (X_{t} \alpha_{N,t})({\bf x}; {\bf y}) ( a^{*}(u_{t,{\bf y}}) + a(\overline{v}_{t,{\bf y}}) ) \\
& = 2\int d{\bf x} d{\bf z}\, ( a^{*}(u_{t,{\bf x}}) + a(\overline{v}_{t,{\bf x}}) ) X_{t}({\bf x}; {\bf z}) ( a((u_{t} \omega_{t})_{{\bf z}}) + a^{*}((\overline{v}_{t} \overline{\omega}_{t})_{{\bf z}} ) \\
&\quad - 2\int d{\bf x} d{\bf y}\, ( a^{*}(u_{t,{\bf x}}) + a(\overline{v}_{t,{\bf x}}) ) X_{t}({\bf x}; {\bf z}) ( a^{*}((u_{t} \alpha_{t})_{{\bf z}}) + a((\overline{v}_{t} \overline{\alpha}_{t})_{{\bf z}} )\;. 
\end{split}
\end{equation}
Using the relations (\ref{eq:simp1}), (\ref{eq:simp2}), we get
\begin{equation}\label{eq:328}
\begin{split}
&2R^{*}_{t} \Big( d\Gamma( X_{t} \omega_{t} ) + \int d{\bf x} d{\bf y}\, a^{*}_{{\bf x}} a^{*}_{{\bf y}} ( X_{t} \alpha_{t} )({\bf x}; {\bf y}) \Big) R_{t} - \text{h.c.}\\
&\qquad = 2\int d{\bf x} d{\bf z}\, ( a^{*}(u_{t,{\bf x}}) + a(\overline{v}_{t,{\bf x}}) ) X_{t}({\bf x}; {\bf z}) a^{*}(\overline{v}_{t,{\bf z}}) - {\text{h.c.}} \\
&\qquad =  2\int d{\bf x} d{\bf z}\, a^{*}(u_{t,{\bf x}}) X_{t}({\bf x}; {\bf z}) a^{*}(\overline{v}_{t,{\bf z}}) - {\text{h.c.}}\;.
\end{split}
\end{equation}
Finally, consider the terms corresponding to $\Pi_{t}$ in (\ref{eq:CD}). Their contribution to (\ref{eq:quadratic}) is:
\begin{equation}
\begin{split}
&- R^{*}_{t} d\Gamma( \Pi_{t} \alpha_{t}^{*} - \alpha_{t} \Pi_{t}^{*} )  R_{t}  + R^{*}_{t} \int d{\bf x} d{\bf y}\, a^{*}_{{\bf x}} a^{*}_{{\bf y}} \Big( \Pi_{t} \overline{\omega}_{t} + \omega_{t} \Pi_{t} \Big) ({\bf x}; {\bf y}) R_{t} - \text{h.c.} \\
&\quad \equiv 2 R^{*}_{t} d\Gamma( \Pi_{t} \overline{\alpha}_{t}) R_{t} + 2 R^{*}_{t} \int d{\bf x} d{\bf y}\, a^{*}_{{\bf x}} a^{*}_{{\bf y}} \big( \Pi_{t} \overline{\omega}_{t}\big)({\bf x}; {\bf y}) R_{t} - \text{h.c.}\;,
\end{split}
\end{equation}
where we used the antisymmetry of $\alpha_{N,t}$ and of $\Pi_{t}$. We have:
\begin{equation}
\begin{split}
&2 R^{*}_{t} \Big(d\Gamma( \Pi_{t} \overline{\alpha}_{t}) + \int d{\bf x} d{\bf y}\, a^{*}_{{\bf x}} a^{*}_{{\bf y}} \big( \Pi_{t} \overline{\omega}_{t}\big)({\bf x}; {\bf y})\Big) R_{t} - \text{h.c.} \\
& \quad = 2\int d{\bf x} d{\bf y}\, ( a^{*}(u_{t, {\bf x}}) + a(\overline{v}_{t,{\bf x}}) ) \big( \Pi_{t} \overline{\alpha}_{t} \big)({\bf x}; {\bf y}) ( a(u_{t,{\bf y}}) + a^{*}(\overline{v}_{t,{\bf y}}) ) \\
&\qquad + 2\int d{\bf x} d{\bf y}\, ( a^{*}(u_{t, {\bf x}}) + a(\overline{v}_{t,{\bf x}}) ) \big( \Pi_{t} \overline{\omega}_{t}\big)({\bf x}; {\bf y})  ( a^{*}(u_{t, {\bf y}}) + a(\overline{v}_{t,{\bf y}}) ) \\
&\quad = -2\int d{\bf x} d{\bf z}\, ( a^{*}(u_{t, {\bf x}}) + a(\overline{v}_{t,{\bf x}}) ) \Pi_{t}({\bf x}; {\bf z}) ( a((u_{t} \alpha_{t})_{{\bf z}}) + a^{*}((\overline{v}_{t} \overline{\alpha}_{t})_{{\bf z}}) ) \\
&\qquad + 2 \int d{\bf x} d{\bf z}\, ( a^{*}(u_{t, {\bf x}}) + a(\overline{v}_{t,{\bf x}}) ) \Pi_{t}({\bf x}; {\bf z}) ( a^{*}((u_{t} \omega_{t})_{{\bf z}}) + a((\overline{v}_{t} \overline{\omega}_{t})_{{\bf z}}) )\;.
\end{split}
\end{equation}
Using again the relations (\ref{eq:simp1}), (\ref{eq:simp2}), we get:
\begin{equation}\label{eq:PiV}
\begin{split}
&2 R^{*}_{t} \Big(d\Gamma( \Pi_{t} \overline{\alpha}_{t}) + \int d{\bf x} d{\bf y}\, a^{*}_{{\bf x}} a^{*}_{{\bf y}} \big( \Pi_{t} \overline{\omega}_{t}\big)({\bf x}; {\bf y})\Big) R_{t} - \text{h.c.} \\
&\quad = 2\int d{\bf x} d{\bf z}\, ( a^{*}(u_{t, {\bf x}}) + a(\overline{v}_{t,{\bf x}}) ) \Pi_{t}({\bf x}; {\bf z}) a(\overline{v}_{t,{\bf z}}) - \text{h.c.}
\end{split}
\end{equation}
Putting together (\ref{eq:rhoV2}), (\ref{eq:328}), (\ref{eq:PiV}), we obtained:
\begin{equation}\label{eq:CD3}
\begin{split}
&- R^{*}_{t} \Big( d\Gamma ( C_{t} ) + \int  d{\bf x} d {\bf y}\, a^*_{\bf x} a^*_{\bf y} D_{t}({\bf x}; {\bf y})  \Big)R_{t} - \text{h.c.} \\
&\quad = - 2 \int d{\bf x} \, a^{*}(u_{t,{\bf x}}) (\rho_{t} * V)({\bf x}) a^{*}(\overline{v}_{t, {\bf x}}) + 2\int d{\bf x} d{\bf z}\, a^{*}(u_{t,{\bf x}}) X_{t}({\bf x}; {\bf z}) a^{*}(\overline{v}_{t,{\bf z}})\\&\qquad + 2\int d{\bf x} d{\bf z}\, ( a^{*}(u_{t, {\bf x}}) + a(\overline{v}_{t,{\bf x}}) ) \Pi_{t}({\bf x}; {\bf z}) a(\overline{v}_{t,{\bf z}}) - \text{h.c.}
\end{split}
\end{equation}
Thus, as claimed, we found that the sum of the quadratic terms in (\ref{eq:Q4norm2}) and in (\ref{eq:CD3}) is vanishing. This implies:
\begin{equation}
\begin{split}
&i \varepsilon \frac{d}{dt} \Big( \mathcal{U}_{N}(t;0)^{*} \mathcal{N} \mathcal{U}_{N}(t;0)\Big) \\
& =  -\frac{1}{N}\int d{\bf x} d{\bf y} \,V(x-y) \mathcal{U}_{N}(t;0)^{*} \Big[  a^*(u_{t,{\bf x}}) a(\overline{v}_{t, {\bf y}}) a(u_{t, {\bf y}})a(u_{t, {\bf x}}) + a(\overline{v}_{t,{\bf x}}) a(\overline{v}_{t,{\bf y}}) a(u_{t,{\bf y}})a(u_{t,{\bf x}}) \\
&\quad - a^*(\overline{v}_{t,{\bf x}}) a(\overline{v}_{t,{\bf x}}) a(\overline{v}_{t,{\bf y}}) a(u_{t,{\bf y}})\Big] \mathcal{U}_{N}(t;0)  - \text{h.c.}
\end{split}
\end{equation}
The final statement, Eq. (\ref{eq:growthN}), follows from the identity:
\begin{equation}
\begin{split}
&i \varepsilon \frac{d}{dt} \Big(\mathcal{U}_{N}(t;0)^{*} (\mathcal{N} +1)^k \mathcal{U}_{N}(t;0)\Big) \\
&\quad = \sum_{j=1}^k \mathcal{U}_{N}(t;0)^{*} (\mathcal{N} +1)^{j-1} \mathcal{U}_{N}(t;0) \\&\quad\qquad \times \mathcal{U}_{N}(t;0)^{*} \Big[i \varepsilon  \frac{d}{dt}\Big( \mathcal{U}_{N}(t;0)^{*} \mathcal{N} \mathcal{U}_{N}(t;0) \Big) \Big] \mathcal{U}_{N}(t;0)^{*} (\mathcal{N} +1)^{k-j} \mathcal{U}_{N}(t;0)\;.
\end{split}
\end{equation}
This concludes the proof of Proposition \ref{prp:growth}.
\end{proof}
\subsection{Propagation of the semiclassical structure}
In order to bound the growth of fluctuations, we will make crucial use of the propagation of the semiclassical structure of Assumption \ref{ass:sc} along the flow of the HFB equation.
\begin{proposition}[Propagation of the semiclassical structure]\label{prp:sc} Under the same assumptions on $V$ as in Theorem \ref{thm:main}, the following is true. Let $(\omega_{N,t}, \alpha_{N,t})$ be the solution of the Hartree-Fock-Bogoliubov equation (\ref{eq:HFB}). Then, for all times $t \geq 0$:
\begin{equation}\label{eq:scgro}
\begin{split}
&\sup_{p} \frac{1}{1 + |p|} \big\| \big[e^{ip\cdot \hat x},  \omega_{N,t} \big] \big\|_{\text{HS}} + \big\| \big[\varepsilon \nabla ,  \omega_{N,t} \big] \big\|_{\text{HS}} + \frac{1}{N\varepsilon} \big\| \big[\varepsilon \nabla ,  \alpha_{N,t} \big] \big\|_{\text{HS}} \\
&\leq \Big( \sup_{p} \frac{1}{1 + |p|} \big\| \big[e^{ip\cdot \hat x},  \omega_{N} \big] \big\|_{\text{HS}} + \big\| \big[\varepsilon \nabla ,  \omega_{N} \big] \big\|_{\text{HS}}\\&\qquad + \frac{1}{N\varepsilon} \big\| \big[\varepsilon \nabla ,  \alpha_{N} \big] \big\|_{\text{HS}} + |t| e^{\frac{C t}{N\varepsilon}} \frac{1}{N\varepsilon} \|\alpha_{N}\|_{\text{HS}} \Big) e^{C|t|}\;.
\end{split}
\end{equation}
Furthermore,
\begin{equation}\label{eq:traces}
\| \omega_{N,t} \|_{\text{HS}}^2 \leq \| \omega_{N} \|_{\text{tr}}\;,\qquad \| \alpha_{N,t} \|_{\text{HS}} \leq  e^{C |t| / (N\varepsilon)}\| \alpha_{N} \|_{\text{HS}}\;.
\end{equation}
%t
\end{proposition}
\begin{remark}\label{rem:propass} In particular, under the Assumption \ref{ass:sc}, the bound (\ref{eq:scgro}) allows to propagate the estimates (\ref{eq:assalpha}), (\ref{eq:sc}) for all times:
\begin{equation}\label{eq:scgro2}
\begin{split}
\sup_{p} \frac{1}{1 + |p|} \big\| \big[e^{ip\cdot \hat x},  \omega_{N,t} \big] \big\|_{\text{HS}} &\leq C e^{C|t|} N^{\frac{1}{3}}\;,\quad \big\| \big[\varepsilon \nabla ,  \omega_{N,t} \big] \big\|_{\text{HS}}\leq C e^{C|t|} N^{\frac{1}{3}}\;,\\&\big\| \big[\varepsilon \nabla ,  \alpha_{N,t} \big] \big\|_{\text{HS}} \leq C e^{C|t|} N\;.
\end{split}
\end{equation}
\end{remark}
\begin{proof}
The proof is based on an extension of \cite{BPS} from the Hartree-Fock equation to the Hartree-Fock-Bogoliubov equation. To begin, let us prove the invariance of the trace-norm of $\omega_{N,t}$. Since $0\leq \omega_{N} \leq 1$, we have:
\begin{equation}
\| \omega_{N,t} \|^{2}_{\text{HS}} \leq \tr\, \omega_{N,t}\;.
\end{equation}
Then, from (\ref{eq:HFB}):
\begin{equation}\label{eq:traceomega}
\begin{split}
i \varepsilon \partial_{t} \tr\, \omega_{N,t} = \tr \big( \Pi_{t} \alpha^*_{N,t} - \alpha_{N,t} \Pi_{t}^* \big)
\end{split}
\end{equation}
where we used that the trace of the commutator is zero. Next, we rewrite:
\begin{equation}
\begin{split}\label{eq:traceomega2}
\tr\, \Pi_{t} \alpha_{N,t}^{*} &= \int d{\bf x} d{\bf y}\, \Pi_{t}({\bf x}; {\bf y}) \alpha_{N,t}^{*}({\bf y}; {\bf x}) \\
&= -\frac{1}{N} \int d{\bf x} d{\bf y}\, V(x-y) |\alpha_{N,t}({\bf x}; {\bf y})|^{2}
\end{split}
\end{equation}
thanks to the antisymmetry of $\alpha_{N,t}$. Being the right-hand side of Eq. (\ref{eq:traceomega}) equal to the imaginary part of $\tr\, \Pi_{t} \alpha_{N,t}^{*}$, and since the right-hand side of (\ref{eq:traceomega2}) is real, we immediately conclude that $\tr\, \omega_{N,t}$ is constant in time. This proves the first inequality in (\ref{eq:traces}).

Next, let us prove the propagation of the commutator estimates. Let us start by computing, using the HFB equation for $\omega_{N,t}$:
\begin{equation}\label{eq:scprop1}
\begin{split}
i \varepsilon \frac{d}{dt} \big[e^{ip\cdot \hat x},  \omega_{N,t} \big] &= \big[ e^{ip\cdot \hat x} , \big[ h_{\text{HF}}(t),  \omega_{N,t} \big] \big] + \big[ e^{ip\cdot \hat x}  , \Pi_{t} \alpha_{N,t}^* - \alpha_{N,t} \Pi_{t}^* \big] \\
&= \big[ h_{\text{HF}}(t) , \big[ e^{ip\cdot \hat x},  \omega_{N,t} \big] \big] + \big[  \omega_{N,t}  , \big[ h_{\text{HF}}(t) , e^{ip\cdot \hat x} \big] \big] + \big[ e^{ip\cdot \hat x}  , \Pi_{t} \alpha_{N,t}^* - \alpha_{N,t} \Pi_{t}^* \big] \\
&= \big[ h_{\text{HF}}(t) , \big[ e^{ip\cdot \hat x},  \omega_{N,t} \big] \big] -  \big[  \omega_{N,t}  , \big[ \varepsilon^2 \Delta , e^{ip\cdot \hat x} \big] \big] - \big[  \omega_{N,t}  , \big[ X_t , e^{ip\cdot \hat x} \big] \big] \\&\quad + \big[ e^{ip\cdot \hat x}  , \Pi_{t} \alpha_{N,t}^* - \alpha_{N,t} \Pi_{t}^* \big]\;, 
\end{split}
\end{equation}
where in the second step we used the Jacobi identity, and in the third step we used that the direct term $\rho_{t}* V$ commutes with $e^{ip\cdot \hat x}$. Consider the term involving the commutator with the Laplacian. Noticing that:
\begin{equation}
\big[ \varepsilon^2 \Delta , e^{ip\cdot \hat x} \big] =  i \varepsilon \nabla \cdot \varepsilon p \,e^{ip\cdot \hat x} + e^{ip\cdot \hat x} \varepsilon p \cdot i \varepsilon \nabla\;,
\end{equation}
we obtain:
\begin{equation}
\begin{split}
\big[  \omega_{N,t}  , \big[ \varepsilon^2 \Delta , e^{ip\cdot \hat x} \big] \big] &= \big[ \omega_{N,t} , i \varepsilon \nabla \big] \cdot \varepsilon p \, e^{ip\cdot \hat x} + i \varepsilon \nabla \cdot \varepsilon p \big[ \omega_{N,t} , e^{ip\cdot \hat x} \big] \\&\quad + \varepsilon p \, e^{ip\cdot \hat x}\cdot \big[ \omega_{N,t} , i \varepsilon \nabla \big]  +\big[ \omega_{N,t} , e^{ip\cdot \hat x} \big] \varepsilon p \cdot i \varepsilon \nabla\;.
\end{split}
\end{equation}
The terms involving the gradient operator outside of the commutator, which are unbounded, can be reabsorbed thanks to a suitable unitary conjugation, together with the first term in the right-hand side of the last line of (\ref{eq:scprop1}). Consider the unitary groups $U_{1}(t;s)$ and $U_{2}(t;s)$, generated by the following differential equations:
\begin{equation}
i\varepsilon \partial_t U_1(t,s) = A(t) U_1(t,s)\;,\qquad i\varepsilon \partial_t U_2 (t,s)= B(t) U_2(t,s)\;,\qquad U_{1}(s;s) = U_{2}(s;s) = 1
\end{equation}
where:
\begin{equation}
A(t)= h_{\text{HF}}(t) + i \varepsilon^2 \nabla \cdot p\;,\qquad B(t)= h_{\text{HF}}(t) - i \varepsilon^2 \nabla \cdot p\;.
\end{equation}
We then get:
\begin{equation}
\begin{split}
&i \varepsilon \frac{d}{dt} \Big( U_1^*(t,0)  \big[e^{ip\cdot \hat x},  \omega_{N,t} \big] U_2(t,0) \Big) \\
&\quad = U_1^*(t,0) \Big( - A(t)  \big[e^{ip\cdot \hat x},  \omega_{N,t} \big] +  \big[e^{ip\cdot \hat x},  \omega_{N,t} \big] B(t) + i \varepsilon \frac{d}{dt}  \big[e^{ip\cdot \hat x},  \omega_{N,t} \big]  \Big) U_2(t,0) \\
&\quad = U_1^*(t,0) \Big( - \big[ \omega_{N,t} , i \varepsilon \nabla \big] \cdot \varepsilon p \, e^{ip\cdot \hat x} -  \varepsilon p \, e^{ip\cdot \hat x} \cdot \big[ \omega_{N,t} , i \varepsilon \nabla \big]  - \big[  \omega_{N,t}  , \big[ X_t , e^{ip\cdot \hat x} \big] \big] \\
&\quad\quad + \big[ e^{ip\cdot \hat x}  , \Pi_{t} \alpha_{N,t}^* - \alpha_{N,t} \Pi_{t}^* \big]   \Big) U_2(t,0).
\end{split}
\end{equation}
Integrating in time this identity, and using $\| U_1^*(t,0)  \big[e^{ip\cdot \hat x},  \omega_{N,t} \big] U_2(t,0) \|_{\text{HS}} = \| \big[e^{ip\cdot \hat x},  \omega_{N,t} \big] \|_{\text{HS}}$, we find:
\begin{equation}\label{eq:intform}
\begin{split}
\big\| &\big[e^{ip\cdot \hat x}, \omega_{N,t} \big] \big\|_{\text{HS}} \\ &\leq \big\| \big[e^{ip\cdot \hat x},  \omega_{N,0} \big] \big\|_{\text{HS}} + 2 |p| \int_0^t ds\, \big\| \big[\varepsilon \nabla,  \omega_{N,s} \big] \big\|_{\text{HS}} + \frac{1}{\varepsilon} \int_0^t ds\, \big\| \big[  \omega_{N,s}  , \big[ X_s , e^{ip\cdot \hat x} \big] \big] \big\|_{\text{HS}}  \\
&\quad  + \frac{1}{\varepsilon} \int_{0}^{t} ds\, \big\| \big[ e^{ip\cdot \hat x}  , \Pi_s  \alpha_{N,s}^* - \alpha_{N,s} \Pi_s^* \big] \big\|_{\text{HS}}\;.
\end{split}
\end{equation}
Next, it is useful to rewrite the direct term $V * \rho_{s}$, the exchange term $X_s$ and the term $\Pi_s$ as, for $\hat {\bf x} = (\hat x, \hat \sigma)$, acting as $(\hat {\bf x} f)(x,\sigma) = (x,\sigma) f(x,\sigma)$:
\begin{equation}\label{eq:direct}
\begin{split}
(V * \rho_s)(\hat {\bf x}) &= \frac{1}{N} \int dq  \, \hat{V}(q) \hat{\rho}_s(q, \hat \sigma)  e^{iq\cdot \hat x} \\
X_s &= \frac{1}{N} \int dq\, \hat{V}(q) e^{iq\cdot \hat x} \omega_{N,s} e^{-iq\cdot \hat x} \\
\Pi_{s} &= \frac{1}{N} \int dq\, \hat{V}(q) e^{iq\cdot \hat x} \alpha_{N,s} e^{-iq\cdot \hat x}\;,
\end{split}
\end{equation}
where $\hat{\rho}_s(q, \sigma) = \int dx\, e^{iq\cdot x} \omega_{N,s}(x,\sigma; x,\sigma)$. Using these representations, and recalling that $\| \omega_{N,s} \|\leq 1$, $\| \alpha_{N,s} \| \leq 1$ (as a consequence of (\ref{eq:square}) and of unitarity of the HFB dynamics), one can bound the third and the fourth terms in the right-hand side of Eq. (\ref{eq:intform}) as:
\begin{equation}
\begin{split}
\big\| \big[  \omega_{N,s}  , \big[ X_s , e^{ip\cdot \hat x} \big] \big] \big\|_{\text{HS}}  &\leq \frac{2}{N} \int dq\, |\hat{V}(q)|\big\| \big[e^{ip\cdot \hat x},  \omega_{N,s}\big] \big\|_{\text{HS}} \\
&\leq \frac{C}{N}  \big\| \big[e^{ip\cdot \hat x},  \omega_{N,s}\big] \big\|_{\text{HS}} \\
\big\| \big[ e^{ip\cdot \hat x}  , \Pi_{s} \alpha_{N,s}^* - \alpha_{N,s} \Pi_{s}^{*} \big] \big\|_{\text{HS}} &\leq \frac{4}{N} \int dq\,  |\hat{V}(q)|  \big\|  \alpha_{N,s} \big\|_{\text{HS}} \\
&\leq \frac{C}{N} \big\|  \alpha_{N,s} \big\|_{\text{HS}}\;,
\end{split}
\end{equation}
where we used the assumptions (\ref{eq:assV}) on the potential. Therefore, plugging these bounds in (\ref{eq:intform}) we obtain:
\begin{equation}\label{eq:gro1}
\begin{split}
\big\| &\big[e^{ip\cdot \hat x},  \omega_{N,t} \big] \big\|_{\text{HS}} \\ &\leq \big\| \big[e^{ip\cdot \hat x},  \omega_{N,0} \big] \big\|_{\text{HS}} + C \int_0^t ds\, \Big(|p|\big\| \big[\varepsilon \nabla,  \omega_{N,s} \big] \big\|_{\text{HS}} + \frac{1}{N\varepsilon} \big\| \big[e^{ip\cdot \hat x},  \omega_{N,s}\big] \big\|_{\text{HS}} + \frac{1}{N\varepsilon}\|\alpha_{N,s}\|_{\text{HS}}\Big)\;.
\end{split}
\end{equation}
Consider now the commutator with the gradient in Eq. (\ref{eq:gro1}). We compute:
\begin{equation}
\begin{split}
i \varepsilon \frac{d}{dt} \big[\varepsilon \nabla , \omega_{N,t} \big] &= \big[ \varepsilon \nabla , \big[ h_{\text{HF}}(t),  \omega_{N,t} \big] + \Pi_{t} \alpha_{N,t}^* - \alpha_{N,t} \Pi_{t}^* \big] \\
&= \big[ h_{\text{HF}}(t) , \big[ \varepsilon \nabla,  \omega_{N,t} \big] \big] + \big[  \omega_{N,t}  , \big[ h_{\text{HF}}(t) ,  \varepsilon \nabla \big] \big] + \big[ \varepsilon \nabla  , \Pi_{t} \alpha_{N,t}^* - \alpha_{N,t} \Pi_{t}^* \big] \\
&= \big[ h_{\text{HF}}(t) , \big[ \varepsilon \nabla ,   \omega_{N,t} \big] \big] + \big[  \omega_{N,t}  , \big[ V * \rho_t , \varepsilon \nabla \big] \big] - \big[  \omega_{N,t}  , \big[ X_t , \varepsilon \nabla \big] \big] \\ &\quad + \big[ \varepsilon \nabla  , \Pi_{t} \alpha_{N,t}^* - \alpha_{N,t} \Pi_{t}^* \big]\; .
\end{split}
\end{equation}
In order to get rid of the first term, we conjugate the commutator with the unitary group $U_{3}(t;s)$, generated by the differential equation:
\begin{equation}\label{eq:U3}
i\varepsilon \partial_t U_3(t,s)= h_{\text{HF}}(t) U_3(t,s)\;,\qquad U_{3}(s;s) = 1\;.
\end{equation}
Then:
\begin{equation}
\begin{split}
&i \varepsilon \frac{d}{dt} \Big( U_3^*(t,0) \big[\varepsilon \nabla , \omega_{N,t} \big] U_3(t,0) \Big) \\
&\quad = U_3^*(t,0) \Big(   \big[  \omega_{N,t}  , \big[ V * \rho_t , \varepsilon \nabla \big] \big] - \big[  \omega_{N,t}  , \big[ X_t , \varepsilon \nabla \big] \big] + \big[ \varepsilon \nabla  , \Pi_{t} \alpha_{N,t}^* - \alpha_{N,t} \Pi_{t}^* \big]  \Big) U_3(t,0)\;.
\end{split}
\end{equation}
Integrating in time, and using the invariance of the Hilbert-Schmidt norm under unitary conjugations, we obtain:
\begin{equation}\label{eq:commnabla}
\begin{split}
\big\| &\big[\varepsilon \nabla ,  \omega_{N,t} \big] \big\|_{\text{HS}} \\ &\leq \big\| \big[\varepsilon \nabla ,  \omega_{N} \big] \big\|_{\text{HS}} + \frac{1}{\varepsilon}  \int_0^t ds\, \big\| \big[  \omega_{N,s}, \big[ V * \rho_s , \varepsilon \nabla \big] \big] \big\|_{\text{HS}} \\& \quad + \frac{1}{\varepsilon} \int_0^t ds \, \big\| \big[  \omega_{N,s}, \big[ X_s , \varepsilon \nabla \big] \big] \big\|_{\text{HS}} + \frac{1}{\varepsilon} \int_0^t ds\, \big\| \big[ \varepsilon \nabla  , \Pi_{s}\alpha_{N,s}^* - \alpha_{N,s} \Pi_{s}^{*} \big]  \big\|_{\text{HS}}\;.
\end{split}
\end{equation}
Consider the second term in Eq. (\ref{eq:commnabla}). We have:
\begin{equation}\label{eq:Vrhocomm}
\begin{split}
\big\| \big[  \omega_{N,s}  , \big[ V * \rho_s , \varepsilon \nabla \big] \big] \big\|_{\text{HS}} &= \varepsilon \big\| \big[  \omega_{N,s},  \nabla ( V * \rho_s )\big] \big\|_{\text{HS}} \\
&\leq \varepsilon \sum_{\sigma} \int dp  \,  |\hat{V}(p)| |\hat{\rho}_s(p, \sigma)| |p| \big\| \big[  e^{ip\cdot \hat x}  , \omega_{N,s} \big] \big\|_{\text{HS}} \\
&\leq  \varepsilon  \Big(\int dq  \, |\hat{V}(q)| (1+|q|)^2\Big) \sup_{p} \frac{1}{1+|p|} \big\| \big[  e^{ip\cdot \hat x}  , \omega_{N,s}  \big] \big\|_{\text{HS}}\,,
\end{split}
\end{equation}
where we used Eq. (\ref{eq:direct}) and 
\begin{equation}
|\hat{\rho}_s(p,\sigma)| \leq \frac{1}{N} \sum_{\sigma} \int dx\, \omega_{N,t}(x,\sigma; x,\sigma) =1\;.
\end{equation}
Consider now the third term in Eq. (\ref{eq:commnabla}). We have, thanks to the representation (\ref{eq:direct}):
\begin{equation}\label{eq:Xnablacomm}
\begin{split}
\big\| \big[  \omega_{N,s}, \big[ X_s , \varepsilon \nabla \big] \big] \big\|_{\text{HS}} &\leq \frac{1}{N} \int dq  \,  |\hat{V}(q)| \big\|\big[  \omega_{N,s} , \big[ e^{iq\cdot \hat x} \omega_{N,s} e^{-iq\cdot \hat x}, \varepsilon \nabla \big] \big] \big\|_{\text{HS}} \\
&\leq \frac{2}{N} \int dq\,  |\hat{V}(q)| \big\| e^{iq\cdot \hat x} \big[  \omega_{N,s} , \varepsilon \nabla \big] e^{-iq\cdot \hat x} \big\|_{\text{HS}}\\
&\leq \frac{2}{N}  \int dq  \, |\hat{V}(q)|  \big\| \big[\varepsilon \nabla ,  \omega_{N,s}\big] \big\|_{\text{HS}}\;,
\end{split}
\end{equation}
where we used that $e^{-iq\hat x} \nabla e^{iq\cdot \hat x}= \nabla +iq$ in the second line. Next, consider the last term in Eq. (\ref{eq:commnabla}). In order to estimate it, we shall use the following identity:
\begin{equation}
\begin{split}
&\big[\varepsilon \nabla ,  e^{iq\cdot \hat x} \alpha_{N,s} e^{-iq\cdot \hat x} \alpha_{N,s}^* \big] - \big[\varepsilon \nabla ,  \alpha_{N,s} e^{iq\cdot \hat x} \alpha_{N,s}^* e^{-iq\cdot \hat x}  \big]  \\
&\quad = e^{iq\cdot \hat x} \alpha_{N,s} e^{-iq\cdot \hat x} \big[ \varepsilon \nabla,   \alpha_{N,s}^* \big] + \big[ \varepsilon \nabla,  e^{iq\cdot \hat x} \alpha_{N,s}  e^{-iq\cdot \hat x} \big] \alpha_{N,s}^* \\&\qquad  - \alpha_{N,s}  \big[ \varepsilon \nabla, e^{iq\cdot \hat x} \alpha_{N,s}^* e^{-iq\cdot \hat x} \big] - \big[ \varepsilon \nabla, \alpha_{N,s}  \big] e^{iq\cdot \hat x} \alpha_{N,s}^* e^{-iq\cdot \hat x} \\
&\quad=  e^{iq\cdot \hat x} \alpha_{N,s} e^{-iq\cdot \hat x} \big[ \varepsilon \nabla,   \alpha_{N,s}^* \big] +  e^{iq\cdot \hat x} \big[ \varepsilon \nabla, \alpha_{N,s}   \big] e^{-iq\cdot \hat x} \alpha_{N,s}^* \\ &\qquad - \alpha_{N,s}  e^{iq\cdot \hat x}  \big[ \varepsilon \nabla, \alpha_{N,s}^* \big]e^{-iq\cdot \hat x}  - \big[ \varepsilon \nabla, \alpha_{N,s}  \big] e^{iq\cdot \hat x} \alpha_{N,s}^* e^{-iq\cdot \hat x}\;.
\end{split}
\end{equation}
Plugging this identity into the expression for the last term in Eq. (\ref{eq:commnabla}), we get:
\begin{equation}\label{eq:nablapicomm}
\begin{split}
\big\| \big[ \varepsilon \nabla  , \Pi_{s} \alpha_{N,s}^* - \alpha_{N,s} \Pi_{s}^{*} \big]  \big\|_{\text{HS}} &\leq \frac{4}{N} \int dq\, |\hat{V}(q)|  \big\|\big[\varepsilon \nabla ,  \alpha_{N,s} \big] \big\|_{\text{HS}}\;.
\end{split}
\end{equation}
Using the bounds (\ref{eq:Vrhocomm}), (\ref{eq:Xnablacomm}), (\ref{eq:nablapicomm}) to estimate the right-hand side of (\ref{eq:commnabla}), we find:
\begin{equation}\label{eq:gro2}
\begin{split}
&\big\| \big[\varepsilon \nabla ,  \omega_{N,t} \big] \big\|_{\text{HS}} \leq \big\| \big[\varepsilon \nabla ,  \omega_{N} \big] \big\|_{\text{HS}} \\ &\quad + C\int_{0}^{t}ds\, \Big(\sup_{p} \frac{1}{1+|p|} \big\| \big[  e^{ip\cdot \hat x}  , \omega_{N,s}  \big] \big\|_{\text{HS}} + \frac{1}{N\varepsilon} \big\| \big[\varepsilon \nabla ,  \omega_{N,s}\big] \big\|_{\text{HS}} + \frac{1}{N\varepsilon} \big\|\big[\varepsilon \nabla ,  \alpha_{N,s} \big] \big\|_{\text{HS}}\Big)\;.
\end{split}
\end{equation}
In order to close the Gronwall argument, we need to estimate the time derivative of the commutator of $\alpha_{N,s}$ with the gradient. We compute:
\begin{equation}\label{eq:dalpha}
\begin{split}
&i \varepsilon \frac{d}{dt} \big[\varepsilon \nabla , \alpha_{N,t} \big] = \big[ \varepsilon \nabla ,  h_{\text{HF}}(t) \alpha_{N,t}  + \alpha_{N,t} \overline{h_{\text{HF}}(t)} + \Pi_{t} (1- \overline{\omega_{N,t}}) - \omega_{N,t}  \Pi_{t}  \big] \\
&\quad = h_{\text{HF}}(t) \big[ \varepsilon \nabla,  \alpha_{N,t} \big]  + \big[ \varepsilon \nabla  ,  h_{\text{HF}}(t)  \big] \alpha_{N,t} + \alpha_{N,t} \big[ \varepsilon \nabla  ,  \overline{h_{\text{HF}}(t)}  \big] +  \big[ \varepsilon \nabla,  \alpha_{N,t} \big] \overline{h_{\text{HF}}(t)} \\
&\qquad + \Pi_{t}  \big[ \varepsilon \nabla,  1- \overline{\omega_{N,t}}\big] +  \big[ \varepsilon \nabla,  \Pi_{t} \big]  (1- \overline{\omega_{N,t}}) - \omega_{N,t}  \big[ \varepsilon \nabla,  \Pi_{t} \big] - \big[ \varepsilon \nabla,  \omega_{N,t} \big] \Pi_{t}\;.
\end{split}
\end{equation}
The first and the fourth term can be reabsorbed thanks to a unitary conjugation of the commutator with $U_{3}(t;0)$, recall the definition (\ref{eq:U3}). Notice that this time the relative sign between these terms is plus, and that the fourth term comes with a complex conjugate. For this reason, we shall consider $U_3^*(t;0) \big[\varepsilon \nabla , \alpha_{N,t} \big] \overline{U_3}(t;0)$, with $\overline{U_3}(t;0) := \mathcal{C} U_{3}(t;0) \mathcal{C}$ and $\mathcal{C}$ the complex conjugation operator. We then have:
\begin{equation}
\begin{split}
&i \varepsilon \frac{d}{dt} \Big( U_3^*(t;0) \big[\varepsilon \nabla , \alpha_{N,t} \big] \overline{U_3}(t;0) \Big) \\&\quad = U_3^*(t;0) \Big(-  h_{\text{HF}}(t) \big[  \varepsilon \nabla , \alpha_{N,t} \big] - \big[  \varepsilon \nabla, \alpha_{N,t} \big]  \overline{h_{\text{HF}}(t)} + i \varepsilon \frac{d}{dt} \big[\varepsilon \nabla , \alpha_{N,t} \big]   \Big) \overline{U_3}(t;0)\;,
\end{split}
\end{equation}
and the first two terms cancel the first and the fourth term in Eq. (\ref{eq:dalpha}). Thus, integrating in time and using the invariance of the trace norm under unitary conjugation, we get:
\begin{equation}\label{eq:nablaa}
\begin{split}
\big\| \big[\varepsilon \nabla ,  \alpha_{N,t} \big] \big\|_{\text{HS}} &\leq  \big\| \big[\varepsilon \nabla ,  \alpha_{N} \big] \big\|_{\text{HS}} \\&\quad + \frac{1}{\varepsilon}  \int_0^t ds\, \big\| \big[  \varepsilon \nabla  , h_{\text{HF}}(s) \big]\alpha_{N,s} \big\|_{\text{HS}} + \frac{1}{\varepsilon}  \int_0^t ds\, \big\| \alpha_{N,s} \big[  \varepsilon \nabla  , \overline{h_{\text{HF}}(s)} \big] \big\|_{\text{HS}} \\
 &\quad + \frac{2}{N\varepsilon} \int dq\, |\hat{V}(q)| \int_0^t ds\, \left( \big\| \big[  \varepsilon \nabla  , \omega_{N,s} \big] \big\|_{\text{HS}} + \big\| \big[  \varepsilon \nabla  , \alpha_{N,s} \big] \big\|_{\text{HS}} \right).
\end{split}
\end{equation}
Consider the second term. We have:
\begin{equation}\label{eq:nablahcomm}
\big\| \big[  \varepsilon \nabla  , h_{\text{HF}}(s) \big]\alpha_{N,s} \big\|_{\text{HS}} \leq \big\| (\varepsilon \nabla \rho_{s} * V)\alpha_{N,s} \big\|_{\text{HS}} + \big\| \big[  \varepsilon \nabla  , X_{s} \big]\alpha_{N,s} \big\|_{\text{HS}}\;;
\end{equation}
the first term is bounded as, setting $\hat \rho_{s}(p) = \sum_{\sigma} \hat \rho_{s}(p, \sigma)$:
\begin{equation}\label{eq:nablarhoalpha}
\begin{split}
\big\| (\varepsilon \nabla \rho_{s} * V)\alpha_{N,s} \big\|_{\text{HS}} &\leq \| (\varepsilon  \rho_{s} * \nabla V) \|_{\infty}  \big\| \alpha_{N,s} \big\|_{\text{HS}} \\
&\leq \varepsilon \int dp\, |\hat \rho_{s}(p)| |\hat V(p)| \big\| \alpha_{N,s} \big\|_{\text{HS}} \\
&\leq C\varepsilon \big\| \alpha_{N,s} \big\|_{\text{HS}}\;,
\end{split}
\end{equation}
while the second term is bounded as:
\begin{equation}\label{eq:nablaXalpha}
\begin{split}
\big\| \big[\varepsilon \nabla, X_s \big]  \alpha_{N,s} \big\|_{\text{HS}} &\leq \frac{1}{N} \int dq\,  |\hat{V}(q)| \big\| \big[\varepsilon \nabla,  e^{iq\cdot \hat x} \omega_{N,s} e^{-iq\cdot \hat x} \big] \alpha_{N,s} \big\|_{\text{HS}} \\
&\leq  \frac{1}{N} \int dq\,  |\hat{V}(q)| \big\| e^{iq\cdot \hat x} \big[  \varepsilon \nabla, \omega_{N,s} \big] e^{-iq\cdot \hat x}  \big\|_{\text{HS}} \\
&\leq  \frac{1}{N} \int dq\, |\hat{V}(q)| \big\| \big[\varepsilon \nabla , \omega_{N,s}\big] \big\|_{\text{HS}}\;,
\end{split}
\end{equation}
where in the second step we used that $\|\alpha_{N,s}\|\leq 1$. Therefore, plugging the bounds (\ref{eq:nablarhoalpha}), (\ref{eq:nablaXalpha}) in the right-hand side of (\ref{eq:nablahcomm}):
\begin{equation}
\big\| \big[  \varepsilon \nabla  , h_{\text{HF}}(s) \big]\alpha_{N,s} \big\|_{\text{HS}} \leq C\varepsilon \big\| \alpha_{N,s} \big\|_{\text{HS}} + \frac{C}{N} \big\| \big[\varepsilon \nabla , \omega_{N,s}\big] \big\|_{\text{HS}}\;,
\end{equation}
and the same estimate holds true for the third term in Eq. (\ref{eq:nablaa}). Thus,
\begin{equation}\label{eq:gro3}
\begin{split}
\big\| \big[\varepsilon \nabla ,  \alpha_{N,t} \big] \big\|_{\text{HS}} &\leq \big\| \big[\varepsilon \nabla ,  \alpha_{N} \big] \big\|_{\text{HS}} \\ &\quad + C\int_{0}^{t}ds\, \Big( \big\| \alpha_{N,s} \big\|_{\text{HS}} + \frac{C}{N\varepsilon} \big\| \big[\varepsilon \nabla , \omega_{N,s}\big] \big\|_{\text{HS}} +  \frac{C}{N\varepsilon}\big\| \big[  \varepsilon \nabla  , \alpha_{N,s} \big] \big\|_{\text{HS}} \Big)\;.
\end{split}
\end{equation}
Putting together Eqs. (\ref{eq:gro1}), (\ref{eq:gro2}), (\ref{eq:gro3}), we have:
\begin{equation}\label{eq:gg}
\begin{split}
&\sup \frac{1}{1 + |p|} \big\| \big[e^{ip\cdot \hat x},  \omega_{N,t} \big] \big\|_{\text{HS}} + \big\| \big[\varepsilon \nabla ,  \omega_{N,t} \big] \big\|_{\text{HS}} + \frac{1}{N\varepsilon}\big\| \big[\varepsilon \nabla ,  \alpha_{N,t} \big] \big\|_{\text{HS}} \\
&\quad \leq \sup \frac{1}{1 + |p|} \big\| \big[e^{ip\cdot \hat x},  \omega_{N} \big] \big\|_{\text{HS}} + \big\| \big[\varepsilon \nabla ,  \omega_{N} \big] \big\|_{\text{HS}} + \frac{1}{N\varepsilon} \big\| \big[\varepsilon \nabla ,  \alpha_{N} \big] \big\|_{\text{HS}} \\
&\qquad + C\int_{0}^{t}ds\, \Big( \frac{1}{N\varepsilon} \big\| \alpha_{N,s} \big\|_{\text{HS}} + \big\| \big[\varepsilon \nabla , \omega_{N,s}\big] \big\|_{\text{HS}} \\&\quad\qquad + \sup \frac{1}{1 + |p|} \big\| \big[e^{ip\cdot \hat x},  \omega_{N,s} \big] \big\|_{\text{HS}} + \frac{1}{N\varepsilon} \big\| \big[  \varepsilon \nabla  , \alpha_{N,s} \big] \big\|_{\text{HS}} \Big)\;.
\end{split}
\end{equation}
To conclude, let us consider $\|\alpha_{N,t}\|_{\text{HS}}$. We have:
\begin{equation}
i\varepsilon \frac{d}{dt} U_{3}^{*}(t;0) \alpha_{N,t} \overline{U_{3}}(t;0) = U_{3}^{*}(t;0) \Big( \Pi_{t} (1 - \overline{\omega_{N,t}}) - \omega_{N,t} \Pi_{t}\Big) \overline{U_{3}}(t;0)\;.
\end{equation}
Using the invariance of the norm under unitary conjugation:
\begin{equation}\label{eq:alphatr}
\begin{split}
\| \alpha_{N,t} \|_{\text{HS}} &\leq \| \alpha_{N} \|_{\text{HS}} + \frac{1}{\varepsilon}\int_{0}^{t} ds\, \big( \big\| \Pi_{t} (1 - \overline{\omega_{N,t}}) \big\|_{\text{HS}} + \big\| \omega_{N,t} \Pi_{t} \big\|_{\text{HS}} \big) \\
&\leq \| \alpha_{N} \|_{\text{HS}} + \frac{2}{\varepsilon} \int_{0}^{t}ds\, \| \Pi_{s} \|_{\text{HS}}\;.
\end{split}
\end{equation}
From the representation (\ref{eq:direct}), we get:
\begin{equation}
\begin{split}
\| \Pi_{s} \|_{\text{HS}} &= \Big\| \frac{1}{N} \int dp\, \hat{V}(p) e^{ip\cdot \hat x} \alpha_{N,s} e^{-ip\cdot \hat x} \Big\|_{\text{HS}} \\
&\leq \frac{1}{N} \int dp\, | \hat V(p) | \| \alpha_{N,s} \|_{\text{HS}} \\
&\leq \frac{C}{N} \| \alpha_{N,s} \|_{\text{HS}}\;.
\end{split}
\end{equation}
Plugging this estimate in (\ref{eq:alphatr}), we get:
\begin{equation}
\| \alpha_{N,t} \|_{\text{HS}} \leq \| \alpha_{N} \|_{\text{HS}} + \frac{C}{N\varepsilon} \int_{0}^{t} ds\, \| \alpha_{N,s} \|_{\text{HS}}\;;
\end{equation}
by Gronwall's lemma,
\begin{equation}
\| \alpha_{N,t} \|_{\text{HS}} \leq  e^{\frac{C t}{N\varepsilon}} \| \alpha_{N} \|_{\text{HS}}\;,
\end{equation}
which proves the last inequality in (\ref{eq:traces}). Coming back to (\ref{eq:gg}), we obtained:
\begin{equation}
\begin{split}
&\sup_{p} \frac{1}{1 + |p|} \big\| \big[e^{ip\cdot \hat x},  \omega_{N,t} \big] \big\|_{\text{HS}} + \big\| \big[\varepsilon \nabla ,  \omega_{N,t} \big] \big\|_{\text{HS}} + \frac{1}{N\varepsilon}\big\| \big[\varepsilon \nabla ,  \alpha_{N,t} \big] \big\|_{\text{HS}} \\
&\leq \sup_{p} \frac{1}{1 + |p|} \big\| \big[e^{ip\cdot \hat x},  \omega_{N} \big] \big\|_{\text{HS}} + \big\| \big[\varepsilon \nabla ,  \omega_{N} \big] \big\|_{\text{HS}} + \frac{1}{N\varepsilon} \big\| \big[\varepsilon \nabla ,  \alpha_{N} \big] \big\|_{\text{HS}} + |t| e^{\frac{C t}{N\varepsilon}} \frac{1}{N\varepsilon} \|\alpha_{N}\|_{\text{HS}} \\
&\quad+ C\int_{0}^{t}ds\, \Big( \big\| \big[\varepsilon \nabla , \omega_{N,s}\big] \big\|_{\text{HS}} + \sup_{p} \frac{1}{1 + |p|} \big\| \big[e^{ip\cdot \hat x},  \omega_{N,s} \big] \big\|_{\text{HS}} + \frac{1}{N\varepsilon} \big\| \big[  \varepsilon \nabla  , \alpha_{N,s} \big] \big\|_{\text{HS}} \Big)\;.
\end{split}
\end{equation}
By Gronwall's lemma:
\begin{equation}
\begin{split}
&\sup_{p} \frac{1}{1 + |p|} \big\| \big[e^{ip\cdot \hat x},  \omega_{N,t} \big] \big\|_{\text{HS}} + \big\| \big[\varepsilon \nabla ,  \omega_{N,t} \big] \big\|_{\text{HS}} + \frac{1}{N\varepsilon} \big\| \big[\varepsilon \nabla ,  \alpha_{N,t} \big] \big\|_{\text{HS}} \\
&\leq \Big( \sup_{p} \frac{1}{1 + |p|} \big\| \big[e^{ip\cdot \hat x},  \omega_{N} \big] \big\|_{\text{HS}} + \big\| \big[\varepsilon \nabla ,  \omega_{N} \big] \big\|_{\text{HS}}\\&\qquad + \frac{1}{N\varepsilon} \big\| \big[\varepsilon \nabla ,  \alpha_{N} \big] \big\|_{\text{HS}} + |t| e^{\frac{C t}{N\varepsilon}} \frac{1}{N\varepsilon} \|\alpha_{N}\|_{\text{HS}} \Big) e^{C|t|}\;,
\end{split}
\end{equation}
which proves (\ref{eq:scgro}). This concludes the proof of Proposition \ref{prp:sc}.
\end{proof}
\subsection{The auxiliary dynamics}
In this section we shall introduce a new dynamics, called the auxiliary dynamics and denoted by $\widetilde{\mathcal{U}}_{N}(t;s)$, on which we will be able to control powers of the number operator via a Gronwall-type inequality. Later, we will prove that $\mathcal{U}_{N}(t;s)$ and $\widetilde{\mathcal{U}}_{N}(t;s)$ are close, in a strong sense. In order to define the auxiliary dynamics, let us denote by $\mathcal{L}_{N}(t)$ the generator of the fluctuation dynamics $\mathcal{U}_{N}(t;s)$:
\begin{equation}
i\varepsilon \partial_{t} \mathcal{U}_{N}(t;s) = \mathcal{L}_{N}(t) \mathcal{U}_{N}(t;s)\;,\qquad \mathcal{U}_{N}(s;s) = \mathbbm{1}\;.
\end{equation}
The generator $\mathcal{L}_{N}(t)$ could be explicitly computed (proceeding as in \cite{BJPSS}), but this will not be needed in the present setting. Let us define the self-adjoint operator:
\begin{equation}\label{eq:widetildeL}
\begin{split}
&\widetilde{\mathcal{L}}_{N}(t) := \mathcal{L}_{N}(t) \\&\quad - \frac{2}{N} \text{Re} \int d{\bf x} d{\bf y}\, V(x-y) \big( a^*(u_{t,{\bf x}}) a(\overline{v}_{t,{\bf y}}) a(u_{t,{\bf y}})a(u_{t,{\bf x}}) + a^*(u_{t,{\bf y}}) a^*(\overline{v}_{t,{\bf y}}) a^*(\overline{v}_{t,{\bf x}}) a(\overline{v}_{t,{\bf x}})\big)\;.
\end{split}
\end{equation}
We define the auxiliary dynamics as the solution of the following differential equation:
\begin{equation}
i\varepsilon \partial_{t} \widetilde{\mathcal{U}}_{N}(t;s) = \widetilde{\mathcal{L}}_{N}(t) \widetilde{\mathcal{U}}_{N}(t;s)\;,\qquad \widetilde{\mathcal{U}}_{N}(s;s) = \mathbbm{1}\;.
\end{equation}
Comparing with Eq. (\ref{eq:growthN}), we have:
\begin{equation}\label{eq:growthNau}
\begin{split}
i \varepsilon \frac{d}{dt} &\Big \langle \widetilde{\mathcal{U}}_{N}(t;0) \xi,  (\mathcal{N} +1)^k \widetilde{\mathcal{U}}_{N}(t;0)) \xi \Big\rangle =  - \frac{4i}{N}  \sum_{j=1}^k \, \mathrm{Im} \int d{\bf x}d {\bf y} \, V(x-y) \\
& \quad \times \Big\langle \widetilde{\mathcal{U}}_{N}(t;0) \xi,  (\mathcal{N} +1)^{j-1} a(\overline{v}_{t,{\bf x}}) a(\overline{v}_{t,{\bf y}}) a(u_{t,{\bf y}})a(u_{t,{\bf x}}) (\mathcal{N} +1)^{k-j} \widetilde{\mathcal{U}}_{N}(t;0) \xi \Big\rangle\;.
\end{split}
\end{equation}
The next proposition allows to bound all powers of the number operator on the auxiliary dynamics.
\begin{proposition}[Bound on the growth of fluctuations on the auxiliary dynamics]\label{prp:Nk} Under the same assumptions as Theorem \ref{thm:main}, the following is true. There exists constants $C, c>0$ independent of $N$ such that, for all $k\in \mathbb{N}$, $t \in \mathbb{R}$, $\xi \in \mathcal{F}$:
\begin{equation}\label{eq:Nkbd}
\langle \widetilde{\mathcal{U}}_{N}(t;0) \xi, (\mathcal{N}+1)^{k} \widetilde{\mathcal{U}}_{N}(t;0) \xi \rangle \leq C (\exp (k \exp (c|t|))) \langle \xi, (\mathcal{N} + 1)^{k} \xi \rangle\;.
\end{equation}
\end{proposition}
\begin{proof} 
Let us define:
\begin{equation}\label{eq:Ikj}
\begin{split}
\text{I}_{k,j} &:= \frac{1}{N} \int d{\bf x} d{\bf y}\, V(x-y) \\
&\quad \cdot \Big\langle \widetilde{\mathcal{U}}_{N}(t;0) \xi,  (\mathcal{N} +1)^{j-1} a(\overline{v}_{t,{\bf x}}) a(\overline{v}_{t,{\bf y}}) a(u_{t,{\bf y}})a(u_{t,{\bf x}}) (\mathcal{N} +1)^{k-j} \widetilde{\mathcal{U}}_{N}(t;0) \xi \Big\rangle\;.
\end{split}
\end{equation}
It is convenient to introduce the notation:
\begin{equation}\label{eq:Btp}
\begin{split}
B_{t,p} &:= \int d{\bf y}\, e^{-ip\cdot y} a(\overline{v}_{t,{\bf y}}) a(u_{t, {\bf y}}) \\
&= \int d{\bf r}_{1} d{\bf r}_{2}\, a_{{\bf r}_{1}} a_{{\bf r}_{2}} \big(v_{t} e^{-ip\cdot \hat x} u^{*}_{t}\big)({\bf r}_{1}; {\bf r}_{2})\;.
\end{split}
\end{equation}
The operator $v_{t} e^{-ip\cdot \hat x} u^{*}_{t}$ can be conveniently expressed in terms of $\omega_{N,t}$ and $\alpha_{N,t}$ via the following key rewriting:
\begin{equation}\label{eq:subtleidentity}
v_t e^{-ip \cdot \hat x} u_t^* = v_t [ \omega_{t}, e^{-ip\cdot \hat x}] u_t^* + \overline{u}_t \alpha_{t}^*  e^{-ip\cdot \hat x} u_t^* - v_t  e^{-ip\cdot \hat x} \alpha_{t} v_t^T\;;
\end{equation}
this identity can be checked opening the commutator in the right-hand side, and using the properties (\ref{eq:simp1}), (\ref{eq:simp2}). In particular, using that $\| u_{t} \| \leq 1$, $\| v_{t} \| \leq 1$, we obtain the following estimate:
\begin{equation}\label{eq:scuv}
\begin{split}
\big\| v_t e^{-ip \cdot \hat x} u_t^* \big\|_{\text{HS}} &\leq \big\| [ \omega_{t}, e^{-ip\cdot \hat x}] \big\|_{\text{HS}} + 2 \| \alpha_{t} \|_{\text{HS}} \\
&\leq C e^{c|t|} N^{1/3} (1 + |p|)\;, 
\end{split}
\end{equation}
where the last step follows from Proposition \ref{prp:sc}, combined with Assumption \ref{ass:sc}. Thus, by Lemma \ref{lem:bounds} we have:
\begin{equation}\label{eq:Bbound}
\begin{split}
\| B_{t,p} \xi \| &\leq C e^{c|t|} N^{1/3} (1 + |p|) \| \mathcal{N}^{1/2} \xi \| \\
\| B^{*}_{t,p} \xi \| &\leq C e^{c|t|} N^{1/3} (1 + |p|) \| (\mathcal{N}+1)^{1/2} \xi \|\qquad \forall \xi \in \mathcal{F}\;.
\end{split}
\end{equation}
Let us estimate $\text{I}_{k,j}$ in Eq. (\ref{eq:Ikj}). Inserting $1= (\mathcal{N}+5)^{k/2 + 1 - j} (\mathcal{N}+5)^{-k/2 - 1 +j}$, and using the Cauchy-Schwarz inequality together with the fact that $[B_{t,p} , \mathcal{N}] =2 B_{t,p}$, we get:
\begin{equation}
\begin{split}
&| \text{I}_{k,j} | \leq \frac{1}{N} \int dp\, | \hat V(p) | \\& \cdot \Big\| (\mathcal{N}+5)^{k/2 + 1 - j} (\mathcal{N} +1)^{j-1} \widetilde{\mathcal{U}}_{N}(t;0) \xi \Big\| \Big\| (\mathcal{N}+5)^{-k/2 - 1 +j} B_{t,p} B_{t,-p} (\mathcal{N} +1)^{k-j} \widetilde{\mathcal{U}}_{N}(t;0) \xi \Big\| \\
&\leq \frac{C}{N}  \int dp\, | \hat V(p) | \Big\| (\mathcal{N}+1)^{k/2} \widetilde{\mathcal{U}}_{N}(t;0) \xi \Big\| \Big\|  B_{t,p} B_{t,-p} (\mathcal{N} +1)^{k/2 - 1} \widetilde{\mathcal{U}}_{N}(t;0) \xi \Big\|\;.
\end{split}
\end{equation}
Using twice the estimate (\ref{eq:Bbound}), we have:
\begin{equation}
\Big\|  B_{t,p} B_{t,-p} (\mathcal{N} +1)^{k/2 - 1} \widetilde{\mathcal{U}}_{N}(t;0) \xi \Big\| \leq Ce^{c|t|}N\varepsilon (1 + |p|)^{2} \Big\|  (\mathcal{N} +1)^{k/2} \widetilde{\mathcal{U}}_{N}(t;0) \xi \Big\|\;,
\end{equation}
where in the last step we used the bound (\ref{eq:scuv}). Thus, we obtained:
\begin{equation}\label{eq:III}
\begin{split}
|\text{I}_{k,j}| &\leq C e^{c|t|}\varepsilon \int dp\, |\hat V(p)|(1 + |p|)^{2} \langle \widetilde{\mathcal{U}}_{N}(t;0) \xi, (\mathcal{N} + 1)^{k} \widetilde{\mathcal{U}}_{N}(t;0) \xi\rangle \\
&\leq K e^{c|t|} \varepsilon \langle \widetilde{\mathcal{U}}_{N}(t;0) \xi, (\mathcal{N} + 1)^{k} \widetilde{\mathcal{U}}_{N}(t;0) \xi\rangle\;.
\end{split}
\end{equation}
In conclusion, thanks to the identity (\ref{eq:growthNau}):
\begin{equation}
\Big| i \varepsilon \frac{d}{dt} \big \langle \widetilde{\mathcal{U}}_{N}(t;0) \xi,  (\mathcal{N} +1)^k \widetilde{\mathcal{U}}_{N}(t;0) \xi \big\rangle \Big| \leq C k e^{c|t|} \varepsilon \langle \widetilde{\mathcal{U}}_{N}(t;0) \xi, (\mathcal{N} + 1)^{k} \widetilde{\mathcal{U}}_{N}(t;0) \xi\rangle\;.
\end{equation}
The final claim, Eq. (\ref{eq:Nkbd}), follows from the application of Gronwall's lemma.
\end{proof}
We now prove that the fluctuation dynamics and the auxiliary dynamics are close in norm.
\begin{proposition}[Norm approximation for the fluctuation dynamics.]\label{prp:normapprox} Under the same assumptions as Theorem \ref{thm:main}, the following holds. For all $\xi \in \mathcal{F}$ and for all times $t\in \mathbb{R}$:
\begin{equation}
\| (\mathcal{U}_{N}(t;0) - \widetilde{\mathcal{U}}_{N}(t;0))\xi \| \leq C N^{-1/3} \exp (4\exp (c|t|)) \big\| (\mathcal{N} + 1)^{3/2} \xi \big\|\;.
\end{equation}
\end{proposition}
\begin{proof}
The argument is similar to the proof of Proposition 4.5 in \cite{BJPSS}. By unitarity:
\begin{equation}\label{eq:unit}
\| (\mathcal{U}_{N}(t;0) - \widetilde{\mathcal{U}}_{N}(t;0))\xi \| = \| (\mathbbm{1} - \mathcal{U}_{N}(t;0)^{*} \widetilde{\mathcal{U}}_{N}(t;0))\xi \|\;,
\end{equation}
where:
\begin{equation}
\begin{split}
&\mathbbm{1} - \mathcal{U}_{N}(t;0)^{*} \widetilde{\mathcal{U}}_{N}(t;0) \\ &= -\int_{0}^{t} ds\, \partial_{s} \mathcal{U}_{N}(s;0)^{*} \widetilde{\mathcal{U}}_{N}(s;0) \\
&\quad = \frac{i}{\varepsilon} \int_{0}^{t} ds\, \mathcal{U}_{N}(s;0)^{*} ( -\mathcal{L}_{N}(s) + \widetilde{\mathcal{L}}_{N}(s) ) \widetilde{\mathcal{U}}_{N}(s;0) \\
&\quad = -\frac{2i}{\varepsilon N} \int_{0}^{t} ds\, \mathcal{U}_{N}(s;0)^{*} \Big(\text{Re} \int d{\bf x} d{\bf y}\, V(x-y) \big( a^*(u_{s,{\bf x}}) a(\overline{v}_{s,{\bf y}}) a(u_{s,{\bf y}})a(u_{s,{\bf x}})\\&\quad\quad + a^*(u_{s,{\bf y}}) a^*(\overline{v}_{s,{\bf y}}) a^*(\overline{v}_{s,{\bf x}}) a(\overline{v}_{s,{\bf x}})\big)\Big) \widetilde{\mathcal{U}}_{N}(s;0)\;,
\end{split}
\end{equation}
where the last step follows from Eq. (\ref{eq:widetildeL}). Plugging this identity in (\ref{eq:unit}) we obtain:
\begin{equation}\label{eq:normapprox}
\begin{split}
&\| (\mathcal{U}_{N}(t;0) - \widetilde{\mathcal{U}}_{N}(t;0))\xi \| \\&\quad \leq \frac{2}{\varepsilon N} \int_{0}^{t} ds\, \Big\| \Big(\text{Re} \int d{\bf x} d{\bf y}\, V(x-y) \big( a^*(u_{s,{\bf x}}) a(\overline{v}_{s,{\bf y}}) a(u_{s,{\bf y}})a(u_{s,{\bf x}})\\&\quad\quad + a^*(u_{s,{\bf y}}) a^*(\overline{v}_{s,{\bf y}}) a^*(\overline{v}_{s,{\bf x}}) a(\overline{v}_{s,{\bf x}})\big)\Big) \widetilde{\mathcal{U}}_{N}(s;0) \xi \Big\|\;.
\end{split}
\end{equation}
Consider the contribution due to the operator $a^*(u_{s,{\bf x}}) a(\overline{v}_{s,{\bf y}}) a(u_{s,{\bf y}})a(u_{s,{\bf x}})$; the other term is estimated in a similar way. We have:
\begin{equation}
\begin{split}
\text{I} &:= \frac{1}{\varepsilon N} \int_{0}^{t} ds\, \Big\| \int d{\bf x} d{\bf y}\, V(x-y) \big( a^*(u_{s,{\bf x}}) a(\overline{v}_{s,{\bf y}}) a(u_{s,{\bf y}})a(u_{s,{\bf x}}) \widetilde{\mathcal{U}}_{N}(s;0) \xi\Big\| \\
& \leq \frac{1}{\varepsilon N} \int_{0}^{t} ds \int dp\, |\hat V(p)| \Big\| d\Gamma(u_{s} e^{-ip\cdot \hat x} u^{*}_{s}) B_{s,p} \widetilde{\mathcal{U}}_{N}(s;0) \xi\Big\|\;,
\end{split}
\end{equation}
recall the notation (\ref{eq:Btp}). By Lemma \ref{lem:bounds}, using that $\|u_{s} e^{-ip\cdot \hat x} u^{*}_{s}\| \leq 1$, we obtain:
\begin{equation}
\begin{split}
\text{I} &\leq \frac{1}{\varepsilon N} \int_{0}^{t} ds \int dp\, |\hat V(p)| \Big\| \mathcal{N} B_{s,p} \widetilde{\mathcal{U}}_{N}(s;0) \xi\Big\| \\
&\leq \frac{1}{\varepsilon N} \int_{0}^{t} ds \int dp\, |\hat V(p)| \Big\| B_{s,p}  \mathcal{N}  \widetilde{\mathcal{U}}_{N}(s;0) \xi\Big\| \\
&\leq \frac{C e^{c|t|}}{N^{1/3}} \int_{0}^{t} ds\, \| ( \mathcal{N} + 1 )^{3/2} \widetilde{\mathcal{U}}_{N}(s;0) \xi \|\;,
\end{split}
\end{equation}
where the last step follows from the bound (\ref{eq:Bbound}). Thus, using Proposition \ref{prp:Nk} we find:
\begin{equation}
\begin{split}
\text{I} &\leq \frac{K e^{c|t|}}{N^{1/3}} \int_{0}^{t} ds\, \exp (3\exp (c|s|)) \| ( \mathcal{N} + 1 )^{3/2} \xi \| \\
&\leq \frac{\widetilde{K}}{N^{1/3}}  \exp (4\exp (c|t|)) \| ( \mathcal{N} + 1 )^{3/2} \xi \|\;.
\end{split}
\end{equation}
The other term contributing to the right-hand side of (\ref{eq:normapprox}) can be estimated in the same way. This concludes the proof of Proposition \ref{prp:normapprox}.
\end{proof}
\subsection{Bound on the growth of fluctuations}
We shall now use Proposition \ref{prp:Nk}, \ref{prp:normapprox}, to estimate the growth of fluctuations on the fluctuation dynamics $\mathcal{U}_{N}(t;s)$. The next proposition is the main result of Section \ref{sec:fluct}.
\begin{proposition}[Bounds for the growth of fluctuations]\label{prp:Nktrue} Under the same assumptions as Theorem \ref{thm:main} the following holds true. For all $k \in \mathbb{N}$ there exists $d(k) > 0$ such that, for all $\xi \in \mathcal{F}$, $t\in \mathbb{R}$ :
\begin{equation}\label{eq:claim}
\langle \mathcal{U}_{N}(t;0) \xi, (\mathcal{N}+1)^{k} \mathcal{U}_{N}(t;0) \xi \rangle \leq \exp (d(k) \exp (c|t|)) \langle \xi, (\mathcal{N} + 1)^{\alpha(k)} \xi \rangle\;,
\end{equation}
\end{proposition}
with $\alpha(k) = (13/2) k$ for $k$ even and $\alpha(k) = (13/2) k + 3/2$ for $k$ odd.
\begin{proof} Without loss of generality we suppose that $k>0$, otherwise the statement is trivial. The proof follows the argument of Section 4.3 of \cite{BJPSS}. We claim that for every $k, n\in \mathbb{N}$ there exists $c(k) > 0$ such that, for all $\xi \in \mathcal{F}$:
\begin{equation}\label{eq:indu}
\begin{split}
&\langle \mathcal{U}_{N}(t;s)\xi, (\mathcal{N} + 1)^{k} \mathcal{U}_{N}(t;s) \xi \rangle \\&\qquad \leq N^{(k - (2/3)n)_{+}} e^{2(n+1) c |t|} \exp(2(n+1) c(k) \exp (c |t-s|)) \langle \xi, (\mathcal{N} + 1)^{2k + 3n} \xi \rangle\;,
\end{split}
\end{equation}
with $(\cdot)_{+}$ the positive part. The claim (\ref{eq:claim}) follows after choosing the smallest $n$ for which $(k - (2/3) n)_{+} = 0$, that is $n = (3/2) k$ if $k$ is even, or $n = (3/2) (k + 1/3)$ if $k$ is odd.

We are left with proving (\ref{eq:indu}). The proof is by induction. Without loss of generality, we can suppose that $k - (2/3) n \geq -1/3$. Let us check Eq. (\ref{eq:indu}) for $n = 0$. To begin, we claim that:
\begin{equation}\label{eq:n0}
R_{t}\mathcal{N}^{k} R^{*}_{t} \leq C^{k} ( \mathcal{N} + N )^{k}\;.
\end{equation}
Let us prove this inequality. We have, by Eq. (\ref{eq:numt}):
\begin{equation}\label{eq:nk}
\begin{split}
R_{t}\mathcal{N}^{k} R^{*}_{t} &= \Big( \mathcal{N} - 2 d\Gamma ( \omega_{N,t}) + N + \int  d{\bf x} d {\bf y} \left( a^*_{\bf x} a^*_{\bf y} \alpha_{N,t} ({\bf y},{\bf x}) + a_{\bf x} a_{\bf y} \overline{\alpha}_{N,t} ({\bf x},{\bf y})   \right)\Big)^{k} \\
&\equiv \Big( \mathcal{N} + \mathcal{A}_{0} + \mathcal{A}_{2} + \mathcal{A}_{-2} + N\Big)^{k}\;,
\end{split}
\end{equation} 
where $\mathcal{A}_{0} =-2 d\Gamma ( \omega_{N,t})$, $\mathcal{A}_{2} = \int  d{\bf x} d {\bf y} a^*_{\bf x} a^*_{\bf y} \alpha_{N,t} ({\bf y},{\bf x})$ and $\mathcal{A}_{-2} = \mathcal{A}_{2}^{*}$. Observe that:
\begin{equation}\label{eq:aaa}
\mathcal{N} \mathcal{A}_{0} = \mathcal{A}_{0} \mathcal{N}\;,\quad \mathcal{N} \mathcal{A}_{2} = \mathcal{A}_{2} (\mathcal{N} + 2)\;,\quad \mathcal{N} \mathcal{A}_{-2}  = \mathcal{A}_{-2} (\mathcal{N} - 2)\;.
\end{equation}
Also, thanks to Lemma \ref{lem:bounds}:
\begin{equation}\label{eq:Akest}
\| \mathcal{A}_{k} \varphi \| \leq BN^{1/2} \| (\mathcal{N} + 1)^{1/2} \varphi \|\;,\qquad k = 0,2,-2\;.
\end{equation}
Now, let us consider the terms proportional to $N^{j} \mathcal{N}^{\ell}$, with $j+\ell \leq k$, in the right-hand side of (\ref{eq:nk}). Using (\ref{eq:aaa}), the general such term has the form:
\begin{equation}
N^{j}  \mathcal{A}_{i_{1}} \cdots \mathcal{A}_{i_{k-j - \ell}} (\mathcal{N} + c)^{\ell}
\end{equation}
for a suitable constant $c$. Recalling (\ref{eq:aaa}), we know that  $-2k\leq c\leq 2k$. Then, we estimate:
\begin{equation}
\begin{split}
&\Big| N^{j}\langle \varphi, \mathcal{A}_{i_{1}} \cdots \mathcal{A}_{i_{k-j - \ell}}  (\mathcal{N} + c)^{\ell} \varphi \rangle\Big| \\
&\qquad = \Big| N^{j}\Big\langle (\mathcal{N} + 4k)^{\frac{k - j + \ell}{2}} \varphi, \mathcal{A}_{i_{1}} \cdots \mathcal{A}_{i_{k-j - \ell}} (\mathcal{N} + 4k + \tilde c)^{-\frac{k - j + \ell}{2}}  (\mathcal{N} + c)^{\ell}\varphi \Big\rangle\Big| \\
&\qquad \leq N^{j}\Big\| (\mathcal{N} + 4k)^{\frac{k - j + \ell}{4}} \varphi \Big\| \Big\| \mathcal{A}_{i_{1}} \cdots \mathcal{A}_{i_{k-j - \ell}} (\mathcal{N} + 4k + \tilde c)^{-\frac{k - j + \ell}{4}} (\mathcal{N} + c)^{\ell}\varphi \Big\| \\
&\qquad \leq B^{k} N^{\frac{j + k - \ell}{2}}\Big\| (\mathcal{N} + 4k)^{\frac{k - j + \ell}{4}} \varphi \Big\| \Big\| (\mathcal{N} + 2k+1)^{\frac{k-j-\ell}{2}} (\mathcal{N} + 2k)^{-\frac{k - j + \ell}{4}} (\mathcal{N} + 2k)^{\ell}\varphi \Big\| \\
&\qquad \leq \widetilde{B}^{k} N^{\frac{j + k - \ell}{2}} \Big\| (\mathcal{N} + 4k)^{\frac{k - j + \ell}{4}} \varphi \Big\|^{2}\;.
\end{split}
\end{equation}
In the first identity, we inserted $1 = (\mathcal{N} + 4k)^{\frac{k - j + \ell}{2}} (\mathcal{N} + 4k)^{-\frac{k - j + \ell}{2}}$ before $\mathcal{A}_{i_{1}}$ and we moved the second operator after $\mathcal{A}_{i_{k-j-\ell}}$; in doing so, using (\ref{eq:aaa}) we pick up an additive constant $\tilde c$, $-2k\leq \tilde {c} \leq 2k$. Then, the second inequality follows from repeated use of (\ref{eq:Akest}), and moving all the number operators to the right. Therefore, we have:
\begin{equation}\label{eq:RNR}
\begin{split}
\langle R_{t}\varphi, \mathcal{N}^{k} R_{t} \varphi \rangle &\leq \widetilde{C}^{k} \sum_{0\leq j,\ell \leq k} N^{\frac{j + k - \ell}{2}} \langle \varphi, (\mathcal{N} + 4k)^{\frac{k - j + \ell}{2}} \varphi\rangle \\
&\leq C^{k} \langle \varphi, (\mathcal{N} + N + 4k)^{k} \varphi\rangle\;,
\end{split}
\end{equation}
where we used that, for positive and commuting operators $A,B$, and for all $a,b\geq 0$, $A^{a} B^{b} \leq (A + B)^{a+b}$. Eq. (\ref{eq:RNR}) implies (\ref{eq:n0}) after a redefinition of the overall constant. 

Therefore,
\begin{equation}
\begin{split}
\langle \mathcal{U}_{N}(t;s)\xi, (\mathcal{N} + 1)^{k} \mathcal{U}_{N}(t;s) \xi \rangle &\leq  C^{k} \langle R_{s} \xi, (\mathcal{N} + N)^{k} R_{s}\xi \rangle \\
&\leq C^{2k} \langle \xi, (\mathcal{N} + 2N)^{k} \xi \rangle
\end{split}
\end{equation}
where the first inequality follows from (\ref{eq:n0}) and from the fact that the Hamiltonian commutes with the number operator, while the second is proved similarly to (\ref{eq:n0}). This establishes (\ref{eq:indu}) for $n=0$.

Suppose now that (\ref{eq:indu}) holds for $n\geq0$, and let us prove it for $n$ replaced by $n+1$, assuming that $k - (2/3) (n+1) \geq -1/3$. We start by writing:
\begin{equation}
\begin{split}
&\langle \mathcal{U}_{N}(t;s)\xi, (\mathcal{N} + 1)^{k} \mathcal{U}_{N}(t;s) \xi \rangle  \\& \quad = \Big\| (\mathcal{N} + 1)^{k/2} \mathcal{U}_{N}(t;s) \xi \Big\|^2\\
&\quad \leq \left( \Big\| (\mathcal{N} + 1)^{k/2} \widetilde{\mathcal{U}}_{N}(t;s) \xi \Big\| + \Big\| (\mathcal{N} + 1)^{k/2} \Big(\mathcal{U}_{N}(t;s) - \widetilde{\mathcal{U}}_{N}(t;s)\Big) \xi \Big\|  \right)^2  \\
&\quad \leq 2 \langle \widetilde{\mathcal{U}}_{N}(t;s)\xi, (\mathcal{N} + 1)^{k} \widetilde{\mathcal{U}}_{N}(t;s) \xi \rangle\\
&\qquad + 2\langle (\mathcal{U}_{N}(t;s) - \widetilde{\mathcal{U}}_{N}(t;s))\xi, (\mathcal{N} + 1)^{k} (\mathcal{U}_{N}(t;s) - \widetilde{\mathcal{U}}_{N}(t;s)) \xi \rangle\\
&\quad\equiv \text{I} + \text{II}\;.
\end{split}
\end{equation}
The term $\text{I}$ is estimated using Proposition \ref{prp:Nk}: 
\begin{equation}\label{eq:Iind}
|\text{I}| \leq C (\exp (k \exp (c|t|))) \langle \xi, (\mathcal{N} + 1)^{k} \xi \rangle\;.
\end{equation}
To control the term $\text{II}$, we use that, setting $\Delta \mathcal{L}_{N}(t) = \mathcal{L}_{N}(t) - \widetilde{\mathcal{L}}_{N}(t)$:
\begin{equation}
\mathcal{U}_{N}(t;s) - \widetilde{\mathcal{U}}_{N}(t;s) = \frac{1}{i\varepsilon} \int_{s}^{t} ds'\, \mathcal{U}_{N}(t;s')\Delta \mathcal{L}_{N}(s')\widetilde{\mathcal{U}}_{N}(s';s)\;.
\end{equation}
Plugging this identity in the expression for $\text{II}$, we have:
\begin{equation}
\text{II} = \frac{2}{\varepsilon^{2}} \int_{s}^{t} ds' \int_{s}^{t} ds'' \langle \mathcal{U}_{N}(t;s')\Delta \mathcal{L}_{N}(s')\widetilde{\mathcal{U}}_{N}(s';s)\xi, (\mathcal{N} + 1)^{k} \mathcal{U}_{N}(t;s'')\Delta \mathcal{L}_{N}(s'')\widetilde{\mathcal{U}}_{N}(s'';s) \xi \rangle\;,
\end{equation}
which we estimate as, by Cauchy-Schwarz inequality:
\begin{equation}\label{eq:IIIest}
| \text{II} | \leq \frac{4|t-s|}{\varepsilon^{2}} \int_{s}^{t} ds' \langle \mathcal{U}_{N}(t;s')\Delta \mathcal{L}_{N}(s')\widetilde{\mathcal{U}}_{N}(s';s)\xi, (\mathcal{N} + 1)^{k} \mathcal{U}_{N}(t;s')\Delta \mathcal{L}_{N}(s')\widetilde{\mathcal{U}}_{N}(s';s) \xi \rangle\;.
\end{equation}
Now, by the inductive assumption (\ref{eq:indu}), we have:
\begin{equation}\label{eq:107}
\begin{split}
&\langle \mathcal{U}_{N}(t;s')\Delta \mathcal{L}_{N}(s')\widetilde{\mathcal{U}}_{N}(s';s)\xi, (\mathcal{N} + 1)^{k} \mathcal{U}_{N}(t;s')\Delta \mathcal{L}_{N}(s')\widetilde{\mathcal{U}}_{N}(s';s) \xi \rangle \\
&\quad \leq N^{(k - (2/3)n)_{+}} e^{2(n+1) c |t|} \exp( 2(n+1) c(k) \exp (c |t-s'|)) \\&\quad\quad \times \langle \Delta \mathcal{L}_{N}(s')\widetilde{\mathcal{U}}_{N}(s';s)\xi, (\mathcal{N} + 1)^{2k + 3n} \Delta \mathcal{L}_{N}(s')\widetilde{\mathcal{U}}_{N}(s';s)\xi \rangle\;.
\end{split}
\end{equation}
In order to proceed, we rewrite the difference of the generators as $ \Delta \mathcal{L}_{N}(s') = A_{N}(s') + A_{N}^{*}(s')$, with
\begin{equation}
\begin{split}
&A_{N}(s') \\&= \frac{1}{N} \int d{\bf x} d{\bf y}\, V(x-y) \big( a^*(u_{s',{\bf x}}) a(\overline{v}_{s',{\bf y}}) a(u_{s',{\bf y}})a(u_{s',{\bf x}}) +  a^{*}(\overline{v}_{s',{\bf x}}) a(\overline{v}_{s',{\bf x}}) a(\overline{v}_{s',{\bf y}}) a(u_{s',{\bf y}}) \big)\;.
\end{split}
\end{equation}
By the Cauchy-Schwarz inequality and using that $\mathcal{N} A_{N}(s') = A_{N}(s') (\mathcal{N} - 2)$, $\mathcal{N} A^{*}_{N}(s') = A^{*}_{N}(s') (\mathcal{N} + 2)$:
\begin{equation}
\begin{split}
&\langle \Delta \mathcal{L}_{N}(s')\widetilde{\mathcal{U}}_{N}(s';s)\xi, (\mathcal{N} + 1)^{2k + 3n} \Delta \mathcal{L}_{N}(s')\widetilde{\mathcal{U}}_{N}(s';s)\xi \rangle \\
&\quad \leq C \Big\| (\mathcal{N} + 1)^{k + \frac{3n}{2}} A_{N}(s') \widetilde{\mathcal{U}}_{N}(s';s)\xi \Big\|^{2} + C \Big\| (\mathcal{N} + 1)^{k + \frac{3n}{2}}  A^{*}_{N}(s') \widetilde{\mathcal{U}}_{N}(s';s)\xi \Big\|^{2} \\
&\quad \leq K \Big\|  A_{N}(s') (\mathcal{N} + 1)^{k + \frac{3n}{2}} \widetilde{\mathcal{U}}_{N}(s';s)\xi \Big\|^{2} + K \Big\|  A^{*}_{N}(s') (\mathcal{N} + 1)^{k + \frac{3n}{2}} \widetilde{\mathcal{U}}_{N}(s';s)\xi \Big\|^{2}\;.
\end{split}
\end{equation}
Both terms can be estimated as in the proof of Proposition \ref{prp:normapprox}. We find:
\begin{equation}
\begin{split}
&\langle \Delta \mathcal{L}_{N}(s')\widetilde{\mathcal{U}}_{N}(s';s)\xi, (\mathcal{N} + 1)^{2k + 3n} \Delta \mathcal{L}_{N}(s')\widetilde{\mathcal{U}}_{N}(s';s)\xi \rangle \\
&\quad \leq \frac{C e^{2c |s'|}}{N^{4/3}} \Big\|  (\mathcal{N} + 1)^{k + \frac{3(n+1)}{2}} \widetilde{\mathcal{U}}_{N}(s';s)\xi \Big\|^{2} \\
&\quad \leq \frac{C e^{2c |s'|}}{N^{4/3}} \exp(c(k) \exp (c |s'-s|)) \Big\|  (\mathcal{N} + 1)^{k + \frac{3(n+1)}{2}} \xi \Big\|^{2}\;,
\end{split}
\end{equation}
for a $k$-dependent constant $c(k)$ (recall the condition $k - (2/3)(n+1) \geq -1/3$). The last step follows from the control of the growth of powers of the number operator on the auxiliary dynamics, Proposition \ref{prp:Nk}.

Plugging this estimate in (\ref{eq:107}), and recalling (\ref{eq:IIIest}), we obtain:
\begin{equation}\label{eq:IIIind}
\begin{split}
|\text{II}| &\leq C |t-s| N^{(k - (2/3)n)_{+}} N^{-2/3} e^{2(n+2)c |t|} \exp( (2(n+1) + 1) c(k) \exp (c |t-s|)) \\& \quad \times \Big\|  (\mathcal{N} + 1)^{k + \frac{3(n+1)}{2}} \xi \Big\|^{2} \\
&\leq C |t-s| N^{(k - (2/3)(n+1))_{+}} e^{2(n+2) c |t|} \exp( (2(n+1) + 1) c(k) \exp (c |t-s|)) \\ &\quad \times \Big\|  (\mathcal{N} + 1)^{k + \frac{3(n+1)}{2}} \xi \Big\|^{2}\;.
\end{split}
\end{equation}
Putting together (\ref{eq:Iind}) and (\ref{eq:IIIind}) we obtain, possibly for a larger $c(k)$:
\begin{equation}
\begin{split}
&\langle \mathcal{U}_{N}(t;s)\xi, (\mathcal{N} + 1)^{k} \mathcal{U}_{N}(t;s) \xi \rangle\\
&\qquad \leq N^{(k - (2/3)(n+1))_{+}} e^{2(n+2) c |t|} \exp( 2(n+2) c(k) \exp (c |t-s|)) \\ &\quad\quad\quad \times \Big\|  (\mathcal{N} + 1)^{k + \frac{3(n+1)}{2}} \xi \Big\|^{2}\;,
\end{split}
\end{equation}
which reproduces (\ref{eq:indu}) with $n$ replaced by $n+1$. This concludes the proof of Proposition~\ref{prp:Nktrue}.
\end{proof}
\section{Proof of the main result}\label{sec:proofmain}
We are now ready to prove our main result, Theorem \ref{thm:main}. We start by writing, by unitarity of $R_{t}$, for $j\leq k$:
\begin{equation}\label{eq:51}
\begin{split}
&\langle \psi_{t}, a^{\sharp_{1}}_{{\bf x}_{1}} \cdots a^{\sharp_{2j}}_{{\bf x}_{2j}} \psi_{t} \rangle \\
&\qquad = \langle R^{*}_{t}\psi_{t}, R^{*}_{t}a^{\sharp_{1}}_{{\bf x}_{1}} R_{t} \cdots R^{*}_{t}a^{\sharp_{2j}}_{{\bf x}_{2j}} R_{t} R^{*}_{t} \psi_{t} \rangle \\
&\qquad = \langle \mathcal{U}_{N}(t;0) \xi, ( a^{\sharp_{1}}(u_{t,{\bf x}_{1}}) + a^{\sharp_{1}}(\overline{v}_{t, {\bf x}_{1}})^{*})  \cdots ( a^{\sharp_{2j}}(u_{t,{\bf x}_{2j}}) + a^{\sharp_{2j}}(\overline{v}_{t, {\bf x}_{2j}})^{*}) \mathcal{U}_{N}(t;0) \xi \rangle\;.
\end{split}
\end{equation}
We then expand the products, and we put the result into normal order, using the canonical anticommutation relations. We call ``Wick contraction'' any nontrivial anticommutator appearing in this process. In doing this, we produce a constant term, given by:
\begin{equation}
\langle \Omega, R^{*}_{t}a^{\sharp_{1}}_{{\bf x}_{1}} R_{t} \cdots R^{*}_{t}a^{\sharp_{2j}}_{{\bf x}_{2j}} R_{t}  \Omega \rangle
\end{equation}
which is equal to the first term in the right-hand side of (\ref{eq:main}). We are left with estimating the remainder, $E_{2j;t}^{\sharp_{1}, \ldots, \sharp_{2j}}({\bf x}_{1}, \ldots, {\bf x}_{2j})$, which is associated to a sum of normal-ordered monomials in the fermionic creation and annihilation operators, each multiplying an integral kernel produced by $\ell$ contractions, with $0\leq \ell\leq j-1$. The general term contributing to the remainder is bounded as:
\begin{equation}
\begin{split}
&\Big(\prod_{n \in A} | \alpha_{N,t}({\bf x}_{n}; {\bf x}_{\pi(n)}) |\Big) \Big( \prod_{r \in B} | \omega_{N,t}({\bf x}_{r}; {\bf x}_{\pi(r)}) |\Big)  \\
&\qquad \times \big| \langle \mathcal{U}_{N}(t;0) \xi, a^{\sharp_{r_{1}}}(\eta^{r_{1}}_{{\bf x}_{r_{1}}})\cdots a^{\sharp_{r_{2j-2\ell}}}(\eta^{r_{2j-2\ell}}_{{\bf x}_{r_{2j-2\ell}}})   \mathcal{U}_{N}(t;0) \xi\rangle \big|\;,
\end{split}
\end{equation}
where: $A,B$ are disjoint subsets of $(1,2,\ldots, 2j)$ such that $|A| + |B| = \ell$; $\pi$ is a permutation of $(1,\ldots, 2j)$ such that 
\begin{equation}
A\cap \pi(A) = B \cap  \pi(B) = A\cap \pi(B) = B \cap \pi(A) = \emptyset\;;
\end{equation}
the labels $r_{1}, \ldots, r_{2j-2\ell}$ are the labels of the operators that are not touched by the pairings, that is:
\begin{equation}
\begin{split}
(A \cup \pi(A) \cup B \cup \pi(B)) \cup (r_{1}, \ldots, r_{2j-2\ell}) &= (1, 2, \ldots, 2j) \\
(A \cup \pi(A) \cup B \cup \pi(B)) \cap (r_{1}, \ldots, r_{2j-2\ell}) &= \emptyset\;;
\end{split}
\end{equation}
the operators $\eta^{r_{i}}$ are such that $\| \eta^{r_{i}} \|\leq 1$ (they are either $u_{t}$ or $\overline{v}_{t}$). Let us now estimate the Hilbert-Schmidt norm of the generic contribution. We have:
\begin{equation}\label{eq:ao}
\begin{split}
&\int d{\bf x}_{1}\ldots d{\bf x}_{2j}\, \Big(\prod_{n\in A} | \alpha_{N,t}({\bf x}_{n}; {\bf x}_{\pi(n)}) |^{2}\Big) \Big( \prod_{r \in B} | \omega_{N,t}({\bf x}_{r}; x_{\pi(r)}) |^{2}\Big) \\&\qquad \times \big| \langle \mathcal{U}_{N}(t;0) \xi, a^{\sharp_{r_{1}}}(\eta^{r_{1}}_{{\bf x}_{r_{1}}})\cdots a^{\sharp_{r_{2j-2\ell}}}(\eta^{r_{2j-2\ell}}_{{\bf x}_{r_{2j-2\ell}}}) \mathcal{U}_{N}(t;0) \xi \rangle \big|^{2} \\
& = \| \alpha_{N,t} \|_{\text{HS}}^{2|A|} \| \omega_{N,t} \|_{\text{HS}}^{2|B|} \int d{\bf z}_{1}\ldots d{\bf z}_{2j-2\ell}\big| \langle \mathcal{U}_{N}(t;0) \xi, a^{\sharp_{r_{1}}}(\eta^{r_{1}}_{{\bf z}_{1}})\cdots a^{\sharp_{r_{2j-2\ell}}}(\eta^{r_{2j-2\ell}}_{{\bf z}_{2j-2\ell}}) \mathcal{U}_{N}(t;0) \xi \rangle \big|^{2}\;.
\end{split}
\end{equation}
Let us denote by $m$ the number of creation operators in the scalar product, $0\leq m \leq 2j-2\ell$. Being the monomial normal ordered, we easily get:
\begin{equation}\label{eq:estNkj}
\begin{split}
&\int d{\bf z}_{1}\ldots d{\bf z}_{2j-2\ell}\big| \langle \mathcal{U}_{N}(t;0) \xi, a^{\sharp_{r_{1}}}(\eta^{r_{1}}_{{\bf z}_{1}})\cdots a^{\sharp_{r_{2j-2\ell}}}(\eta^{r_{2j-2\ell}}_{{\bf z}_{2j-2\ell}}) \mathcal{U}_{N}(t;0) \xi \rangle \big|^{2} \\&\qquad \leq \| \mathcal{N}^{m/2} \mathcal{U}_{N}(t;0) \xi\|^{2} \| \mathcal{N}^{j-\ell - m/2} \mathcal{U}_{N}(t;0) \xi\|^{2}\;.
\end{split}
\end{equation}
The right-hand side of (\ref{eq:estNkj}) can be bounded thanks to Proposition \ref{prp:Nktrue}. We get:
\begin{equation}
|(\ref{eq:ao})| \leq \| \alpha_{N,t} \|_{\text{HS}}^{2|A|} \| \omega_{N,t} \|_{\text{HS}}^{2|B|} C_{2j}^{2} \exp( d(j)\exp (c|t|) )\;.
\end{equation} 
Since $\|\omega_{N,t}\|_{\text{HS}}^{2} \leq N$ and $\| \alpha_{N,t} \|_{\text{HS}}^{2} \leq CN^{2/3}$, we obtain:
\begin{equation}\label{eq:est1term}
|(\ref{eq:ao})| \leq N^{\frac{2}{3} |A|} N^{|B|}C_{2j}^{2} \exp( d(j)\exp (c|t|) )\;,
\end{equation}
and we recall that $|A| + |B| = \ell \leq j-1$. Let $p$ be the absolute value of the difference between the number of $a$'s and $a^{*}$'s in the left-hand side of (\ref{eq:51}). If $p=0$, it is possible to have $|A| = 0$ and $|B| = j-1$. Using that the number of terms contributing to $E_{2j;t}^{\sharp_{1}, \ldots, \sharp_{2j}}({\bf x}_{1}, \ldots, {\bf x}_{2h})$ is bounded by $C^{j} (2j)!$, we easily get:
\begin{equation}\label{eq:pis0}
\| E_{2j,t}^{\sharp_{1}, \ldots, \sharp_{2j}} \|_{\text{HS}}^{2} \leq C(j) N^{j-1} \exp( d(j)\exp (c|t|) )\qquad \text{if $p=0$.}
\end{equation}
Suppose now that $p\geq 2$. Then, $|B| \leq j-p/2$, and $|A| \leq (j-1) - |B|$. Summing over the allowed $|A|$ and $|B|$, we obtain:
\begin{equation}\label{eq:pis2}
\| E_{2j,t}^{\sharp_{1}, \ldots, \sharp_{2j}} \|_{\text{HS}}^{2} \leq C(j) N^{j-2/3 - p/6} \exp( d(j)\exp (c|t|) )\qquad \text{if $p\geq 2$.}
\end{equation}
Putting together (\ref{eq:pis0}) and (\ref{eq:pis2}), the final claim (\ref{eq:HSrem}) follows. This concludes the proof of Theorem \ref{thm:main}.\qed

\appendix 

\section{Check of Assumption \ref{ass:sc} for quasi-free states of the torus}\label{app:BCS}

Here we shall check the validity of Assumption \ref{ass:sc} for translation-invariant fermionic states on the torus $\mathbb{T}^{3}$ of side length $2\pi$. We shall consider states with $N_{\sigma}$ particles with spin $\sigma \in \{\uparrow, \downarrow\}$, $N = N_{\uparrow} + N_{\downarrow}$. Let $\omega_{N}, \alpha_{N}$ be the one-particle density matrix and the pairing matrix of a pure quasi-free state. Let $\omega_{N}(x, \sigma; y, \sigma'), \alpha_{N}(x, \sigma; y, \sigma')$ be the integral kernels of $\omega_{N},\alpha_{N}$. For given $x,y$, we also define:
\begin{equation}\label{eq:2x2}
\omega_{N}(x;y) = \begin{pmatrix} \omega_{N}(x, \uparrow; y, \uparrow) & \omega_{N}(x, \uparrow; y, \downarrow) \\ 
\omega_{N}(x, \downarrow; y, \uparrow) & \omega_{N}(x, \downarrow; y, \downarrow)\end{pmatrix}\;,\quad \alpha_{N}(x;y) = \begin{pmatrix} \alpha_{N}(x, \uparrow; y, \uparrow) & \alpha_{N}(x, \uparrow; y, \downarrow) \\ 
\alpha_{N}(x, \downarrow; y, \uparrow) & \alpha_{N}(x, \downarrow; y, \downarrow)\end{pmatrix}\;.
\end{equation}
Translation invariance means that:
\begin{equation}
\omega_{N}(x, \sigma; y, \sigma') = \omega_{N}(x + a, \sigma; y + a, \sigma')\qquad \forall a\in \mathbb{T}^{3}\;.
\end{equation}
Therefore,  the density is constant:
\begin{equation}
\rho_{\omega}(x,\sigma) := \omega_{N}(x,\sigma;x,\sigma) = \frac{N_{\sigma}}{(2\pi)^{3}}\;,
\end{equation}
where the last identity follows from $\int_{\mathbb{T}^{3}} dx\, \omega_{N}(x,\sigma;x,\sigma) = N_{\sigma}$. We introduce the Fourier symbol associated to $\omega_{N}(x, \sigma; y, \sigma)$ as:
\begin{equation}
\hat \omega_{N}(k,\sigma) = \int_{\mathbb{T}^{3}} dz\, e^{i z\cdot k} \omega_{N}(z, \sigma; 0, \sigma)\;,\qquad k\in \mathbb{Z}^{3}\;.
\end{equation}
It follows that $0 \leq \hat \omega_{N}(k,\sigma) \leq 1$ and that $\sum_{k\in \mathbb{Z}^{3}}  \hat \omega_{N}(k, \sigma) = N_{\sigma}$. In terms of this object, the kinetic energy can be written in momentum space as:
\begin{equation}
 -\tr\, \varepsilon^{2} \Delta \omega_{N} = \sum_{\sigma = \uparrow, \downarrow}\sum_{k\in \mathbb{Z}^{3}} \varepsilon^{2} |k|^{2} \hat \omega_{N}(k, \sigma)\;.
\end{equation}
A simple example of a translation invariant state is the free Fermi gas:
\begin{equation}
\omega^{\text{FFG}}_{N} = \mathbbm{1}(-\varepsilon^{2} \Delta \leq \mu)\;,
\end{equation}
where $\mu$ has to be thought of as a diagonal matrix in spin space, $(\mu f)(x, \sigma) = \mu_{\sigma} f(x, \sigma)$. In this appendix we shall suppose that the Fermi ball is completely filled:
\begin{equation}
N_{\sigma} = \sum_{k: |k| \leq k^{\sigma}_{F} } 1\;,
\end{equation}
for a suitable Fermi momentum $k_{F}^{\sigma} \sim N_{\sigma}^{1/3}$. Without loss of generality, we shall suppose that $(k_{F}^{\sigma})^2$ lies halfway between the highest $|k|^2$ for $k \in B_{F}^{\sigma}$ and the smallest $|k|^2$ for $k\notin B_{F}^{\sigma}$:
\begin{equation}
\text{dist}( \sigma(-\Delta), (k_{F}^{\sigma})^2) = \frac{1}{2}\;.
\end{equation}
Given a general quasi-free state $(\omega_{N}, \alpha_{N})$, its Hartree-Fock-Bogoliubov energy is:
\begin{equation}
\begin{split}
\mathcal{E}_{N}^{\text{HFB}}(\omega_{N}, \alpha_{N}) &= - \tr\, \varepsilon^{2} \Delta \omega_{N} + \frac{1}{2N} \int d{\bf x} d{\bf y}\, V(x-y) \rho_{\omega}({\bf x}) \rho_{\omega}({\bf y})\\&\qquad - \frac{1}{2N} \int d{\bf x} d{\bf y}\, V(x- y) (| \omega_{N}({\bf x};{\bf y}) |^{2} - |\alpha_{N}({\bf x};{\bf y})|^{2})\;.
\end{split}
\end{equation}
If the state is translation-invariant, the HFB functional can be written as:
\begin{equation}\label{eq:HFBtrans}
\begin{split}
\mathcal{E}_{N}^{\text{HFB}}(\omega_{N}, \alpha_{N}) &= \sum_{\sigma}\sum_{k} \varepsilon^{2} |k|^{2} \hat \omega_{N}(k,\sigma) +  \frac{N^{2} \hat V(0)}{2N} \\
&\quad - \frac{1}{2N} \int d{\bf x} d{\bf y}\, V(x- y) (| \omega_{N}({\bf x};{\bf y}) |^{2} - |\alpha_{N}({\bf x};{\bf y})|^{2})\;.
\end{split}
\end{equation}
Notice that the direct energy is constant, for all translation-invariant states with $N$ particles. As in the rest of the paper, we shall assume that the state is pure,
\begin{equation}\label{eq:A9}
|\alpha^{*}_{N}|^{2} = \omega_{N}(1 - \omega_{N})\;.
\end{equation}
We shall be interested in the properties of quasi-free states which are energetically close to the free Fermi gas. 
\begin{proposition}[Bounds for $\alpha_{N}$]\label{prp:pairbd} Let $V$ be bounded. Suppose that $(\omega_{N}, \alpha_{N})$ is a pure, quasi-free and translation invariant state, and that:
\begin{equation}\label{eq:upperFFG}
\mathcal{E}_{N}^{\text{HFB}}(\omega_{N}, \alpha_{N}) \leq \mathcal{E}_{N}^{\text{HFB}}(\omega^{\text{FFG}}_{N}, 0) + CN^{1/3}\;.
\end{equation}
Then, there exists $C>0$ such that:
\begin{equation}\label{eq:alphaHS}
\| \alpha_{N} \|_{\text{HS}}^{2} \leq C N^{2/3}\;,\qquad \| [ \varepsilon \nabla, \alpha_{N} ] \|_{\text{HS}}^{2} \leq CN^{2/3}\;.
\end{equation}
\end{proposition}
\begin{remark} Observe that the second bound in (\ref{eq:alphaHS}) is actually stronger than the corresponding bound assumed in Eq. (\ref{eq:sc}).
\end{remark}
\begin{proof} Let us start by proving the first of (\ref{eq:alphaHS}). By translation invariance and boundedness of the potential:
\begin{equation}\label{eq:HFBlow}
\mathcal{E}_{N}^{\text{HFB}}(\omega_{N}, \alpha_{N}) \geq \sum_{k, \sigma} \varepsilon^{2} |k|^{2} \hat \omega_{N}(k, \sigma) + \frac{N \hat V(0)}{2} - C\|V\|_{\infty}\;.
\end{equation}
The last term is a lower bound of the last line in (\ref{eq:HFBtrans}). Consider the kinetic energy. We rewrite it as:
\begin{equation}
\sum_{k, \sigma} \varepsilon^{2} |k|^{2} \hat \omega_{N}(k, \sigma) = \sum_{k, \sigma} \varepsilon^{2} |k|^{2} \hat \omega^{\text{FFG}}_{N}(k, \sigma) + \sum_{k, \sigma} \varepsilon^{2} |k|^{2} (\hat \omega_{N}(k, \sigma) - \hat \omega^{\text{FFG}}_{N}(k, \sigma) )\;,
\end{equation}
where $\hat \omega_{N}^{\text{FFG}}(k, \sigma) = \chi(k \in B^{\sigma}_{F})$ with $B^{\sigma}_{F}$ the Fermi ball of an $N_{\sigma}$-particle state. Since both states have the same number of particles with given spin,
\begin{equation}\label{eq:A8}
\sum_{k , \sigma} \varepsilon^{2} |k|^{2} \hat \omega_{N}(k, \sigma) = \sum_{k, \sigma} \varepsilon^{2} |k|^{2} \hat \omega^{\text{FFG}}_{N}(k, \sigma) + \sum_{k, \sigma} (\varepsilon^{2} |k|^{2} - \mu_{\sigma}) (\hat \omega_{N}(k, \sigma) - \hat \omega^{\text{FFG}}_{N}(k, \sigma) )\;.
\end{equation}
Consider the last term in the right-hand side of Eq. (\ref{eq:A8}). We rewrite it as:
\begin{equation}\label{eq:A88}
\begin{split}
\sum_{k, \sigma} (\varepsilon^{2} |k|^{2} - \mu_{\sigma}) (\hat \omega_{N}(k, \sigma) - \hat \omega^{\text{FFG}}_{N}(k, \sigma) ) &= \sum_{\sigma, k\in B^{\sigma}_{F}} (\varepsilon^{2} |k|^{2} - \mu_{\sigma}) (\hat \omega_{N}(k, \sigma) - 1 )\\
& \quad + \sum_{\sigma, k \notin B^{\sigma}_{F}} (\varepsilon^{2} |k|^{2} - \mu_{\sigma}) \hat \omega_{N}(k, \sigma) \\
&= \sum_{\sigma, k \in B^{\sigma}_{F}} |\varepsilon^{2} |k|^{2} - \mu_{\sigma}| (1 - \hat \omega_{N}(k, \sigma))\\
& \quad + \sum_{\sigma, k \notin B^{\sigma}_{F}} |\varepsilon^{2} |k|^{2} - \mu_{\sigma}| \hat \omega_{N}(k, \sigma)\;.
\end{split}
\end{equation}
We bound below the right-hand side of (\ref{eq:A88}) as:
\begin{equation}\label{eq:A10}
\begin{split} 
\sum_{k, \sigma} (\varepsilon^{2} |k|^{2} - \mu_{\sigma}) (\hat \omega_{N}(k, \sigma) - &\hat \omega^{\text{FFG}}_{N}(k, \sigma) ) \\ &\geq \sum_{\substack{\sigma, k \in B^{\sigma}_{F} \\ k^{\sigma}_{F} - |k| \geq C}} |\varepsilon^{2} |k|^{2} - \mu_{\sigma}| (1 - \hat \omega_{N}(k, \sigma))\\
& \quad + \sum_{ \substack{ \sigma, k \notin B^{\sigma}_{F} \\ |k| - k^{\sigma}_{F} \geq C}} |\varepsilon^{2} |k|^{2} - \mu_{\sigma}| \hat \omega_{N}(k, \sigma) \\
&\geq C\varepsilon \sum_{\substack{ \sigma, k \in B^{\sigma}_{F} \\ k^{\sigma}_{F} - |k| \geq C}}  (1 - \hat \omega_{N}(k, \sigma)) + C\varepsilon \sum_{ \substack{ \sigma, k \notin B^{\sigma}_{F} \\ |k| - k^{\sigma}_{F} \geq C}} \hat \omega_{N}(k, \sigma)\;.
\end{split}
\end{equation}
Since:
\begin{equation}
\begin{split}
\sum_{\substack{ \sigma, k \in B^{\sigma}_{F} \\ k^{\sigma}_{F} - |k| \geq C}}  (1 - \hat \omega_{N}(k, \sigma)) &\geq \sum_{\sigma,  k \in B^{\sigma}_{F}}  (1 - \hat \omega_{N}(k, \sigma)) - CN^{2/3} = \tr\, \omega_{N}^{\text{FFG}}(1 - \omega_{N}) - CN^{2/3}\\
\sum_{ \substack{ \sigma, k \notin B^{\sigma}_{F} \\ |k| - k^{\sigma}_{F} \geq C}} \hat \omega_{N}(k, \sigma) &\geq \sum_{ \sigma, k \notin B^{\sigma}_{F}} \hat \omega_{N}(k, \sigma) - CN^{2/3} = \tr\, (1 - \omega_{N}^{\text{FFG}}) \omega_{N} - CN^{2/3}\;,
\end{split}
\end{equation}
we obtain the lower bound:
\begin{equation}\label{eq:A888}
\sum_{k, \sigma} (\varepsilon^{2} |k|^{2} - \mu_{\sigma}) (\hat \omega_{N}(k, \sigma) - \hat \omega^{\text{FFG}}_{N}(k, \sigma) ) \geq C \varepsilon \tr\, (1 - \omega_{N}^{\text{FFG}}) \omega_{N} - CN^{1/3}\;.
\end{equation}
Coming back to (\ref{eq:HFBlow}), we have:
\begin{equation}\label{eq:HFBlow2}
\begin{split}
\mathcal{E}_{N}^{\text{HFB}}(\omega_{N}, \alpha_{N}) &\geq \sum_{k, \sigma} \varepsilon^{2} |k|^{2} \hat \omega^{\text{FFG}}_{N}(k,\sigma) + \frac{N \hat V(0)}{2} + C\varepsilon  \tr\, (1 - \omega_{N}^{\text{FFG}}) \omega_{N} - C N^{1/3}\\
&\geq \mathcal{E}_{N}^{\text{HFB}}(\omega_{N}^{\text{FFG}}, 0) + C\varepsilon  \tr\, (1 - \omega_{N}^{\text{FFG}}) \omega_{N} - C N^{1/3}\;,
\end{split}
\end{equation}
where in the last step we used that the sum of the first two terms in the right-hand side of the first line in (\ref{eq:HFBlow2}) reproduces the Hartree energy of the free Fermi gas, and we added and subtracted the exchange energy of the free Fermi gas, which is bounded by a constant. Hence, assumption (\ref{eq:upperFFG}) and Eq. (\ref{eq:HFBlow2}) immediately imply the estimate:
\begin{equation}\label{eq:deltaest2}
\tr\, (1 - \omega_{N}^{\text{FFG}}) \omega_{N} \leq C N^{2/3}\;.
\end{equation}
To conclude, we observe that:
\begin{equation}\label{eq:deltaCS}
\begin{split}
\|\alpha_{N}\|_{\text{HS}}^{2} &= \tr\, \omega_{N}  (1 - \omega_{N}) \\
&= \tr\, \omega_{N}^{\text{FFG}} \omega_{N}  (1 - \omega_{N}) \omega_{N}^{\text{FFG}} +  \tr\, (1-\omega_{N}^{\text{FFG}}) \omega_{N}  (1 - \omega_{N}) (1-\omega_{N}^{\text{FFG}}) \\
&\leq 2 \tr\, \omega_{N}^{\text{FFG}} (1 - \omega_{N})\;;
\end{split}
\end{equation}
combined with (\ref{eq:deltaest2}), this gives the first estimate in (\ref{eq:alphaHS}). Let us now prove the second estimate in (\ref{eq:alphaHS}). We have:
\begin{equation}
\begin{split}
\| [\alpha_{N}, \varepsilon \nabla] \|_{\text{HS}}^{2} &= \tr [\alpha_{N}^{*}, \varepsilon \nabla][\alpha_{N}, \varepsilon \nabla] \\
&= \tr ( \alpha_{N}^{*} \varepsilon \nabla - \varepsilon \nabla\alpha_{N}^{*} ) ( \alpha_{N} \varepsilon \nabla - \varepsilon \nabla \alpha_{N}  ) \\
&= -2 \tr\, \varepsilon^{2}\Delta |\alpha^{*}_N|^{2} + 2\tr\, \alpha_{N}^{*} \varepsilon \nabla \alpha_{N} \varepsilon \nabla \\
&\leq -4 \tr\, \varepsilon^{2}\Delta |\alpha^{*}_N|^{2}\;,
\end{split}
\end{equation}
where the last step follows from Cauchy-Schwarz inequality for traces. Thus, by Eq. (\ref{eq:A9}):
\begin{equation}\label{eq:commkinest}
\| [\alpha_{N}, \varepsilon \nabla] \|_{\text{HS}}^{2} \leq -4 \tr\, \varepsilon^{2}\Delta \omega_{N} (1 - \omega_{N})\;.
\end{equation}
The right-hand side of (\ref{eq:commkinest}) can be bounded proceeding as for the first estimate in (\ref{eq:alphaHS}). Consider the difference of the kinetic energies of $\omega_{N}$ and $\omega_{N}^{\text{FFG}}$. Proceeding as in (\ref{eq:A88}):
\begin{equation}
\begin{split}
\sum_{k, \sigma} \varepsilon^{2} |k|^{2} (\hat \omega_{N}(k, \sigma) - \hat \omega^{\text{FFG}}_{N}(k, \sigma) ) &= \sum_{\sigma, k \in B^{\sigma}_{F}} |\varepsilon^{2} |k|^{2} - \mu_{\sigma}| (1 - \hat \omega_{N}(k, \sigma))\\
& \quad + \sum_{\sigma, k \notin B^{\sigma}_{F}} |\varepsilon^{2} |k|^{2} - \mu_{\sigma}| \hat \omega_{N}(k, \sigma) \\
&= \sum_{\sigma, k} |\varepsilon^{2} |k|^{2} - \mu_{\sigma}| \hat \omega^{\text{FFG}}_{N}(k, \sigma) (1 - \hat \omega_{N}(k, \sigma))\\
& \quad + \sum_{\sigma, k} |\varepsilon^{2} |k|^{2} - \mu_{\sigma}| (1 - \hat \omega^{\text{FFG}}_{N}(k, \sigma)) \hat \omega_{N}(k, \sigma)\;,
\end{split}
\end{equation}
which we bound below by dropping the absolute values:
\begin{equation}\label{eq:bdkinn}
\begin{split}
&\sum_{k, \sigma} \varepsilon^{2} |k|^{2} (\hat \omega_{N}(k, \sigma) - \hat \omega^{\text{FFG}}_{N}(k, \sigma) ) \\& \geq \sum_{\sigma, k} \varepsilon^{2} |k|^{2} \Big( \hat \omega^{\text{FFG}}_{N}(k, \sigma) (1 - \hat \omega_{N}(k, \sigma)) + (1 - \hat \omega^{\text{FFG}}_{N}(k, \sigma)) \hat \omega_{N}(k, \sigma) \Big) \\
&\quad - 2\sum_{k, \sigma} \mu_{\sigma} (1 - \hat \omega^{\text{FFG}}_{N}(k, \sigma)) \hat \omega_{N}(k, \sigma)\;.
\end{split}
\end{equation}
Next, we use that, similarly to (\ref{eq:deltaCS}):
\begin{equation}
\begin{split}
- \tr\, \varepsilon^{2}\Delta \omega_{N} (1 - \omega_{N}) &= \tr\, (i \varepsilon \nabla) \omega_{N} (1 - \omega_{N}) (i \varepsilon \nabla)^{*} \\
&= \tr\, (i \varepsilon \nabla) \omega_{N}^{\text{FFG}} \omega_{N} (1 - \omega_{N}) \omega_{N}^{\text{FFG}} (i \varepsilon \nabla)^{*} \\
&\quad + \tr\, (i \varepsilon \nabla) (1-\omega_{N}^{\text{FFG}}) \omega_{N} (1 - \omega_{N}) (1- \omega_{N}^{\text{FFG}}) (i \varepsilon \nabla)^{*} \\
&\leq \tr\, (i \varepsilon \nabla) \omega_{N}^{\text{FFG}} (1 - \omega_{N}) \omega_{N}^{\text{FFG}} (i \varepsilon \nabla)^{*} \\
&\quad +  \tr\, (i \varepsilon \nabla) (1-\omega_{N}^{\text{FFG}}) \omega_{N} (1- \omega_{N}^{\text{FFG}}) (i \varepsilon \nabla)^{*}\;,
\end{split}
\end{equation}
where we used that $[ \omega_{N}^{\text{FFG}}, \varepsilon \nabla ] = 0$. Thus, we obtained:
\begin{equation}\label{eq:AAA}
- \tr\, \varepsilon^{2}\Delta \omega_{N} (1 - \omega_{N}) \leq -\tr\, \varepsilon^{2}\Delta \big( \omega_{N}^{\text{FFG}} (1 - \omega_{N})  + (1-\omega_{N}^{\text{FFG}}) \omega_{N} \big)\;,
\end{equation}
where the right-hand side of (\ref{eq:AAA}) is what appears in the right-hand side of (\ref{eq:bdkinn}). Hence, plugging this bound in (\ref{eq:bdkinn}), we get:
\begin{equation}\label{eq:kinest3}
\begin{split}
\sum_{k, \sigma} \varepsilon^{2} |k|^{2} (\hat \omega_{N}(k, \sigma) - \hat \omega^{\text{FFG}}_{N}(k, \sigma) ) &\geq - \tr\, \varepsilon^{2}\Delta \omega_{N} (1 - \omega_{N}) - 2 \tr\, \mu (1 - \omega_{N}^{\text{FFG}}) \omega_{N} \\
&\geq - \tr\, \varepsilon^{2}\Delta \omega_{N} (1 - \omega_{N}) - CN^{2/3}\;,
\end{split}
\end{equation}
where in the last step we used (\ref{eq:deltaest2}). Proceeding as in (\ref{eq:A888}), (\ref{eq:HFBlow2}), the bound (\ref{eq:kinest3}) implies:
\begin{equation}
\mathcal{E}_{N}^{\text{HFB}}(\omega_{N}, \alpha_{N}) \geq \mathcal{E}_{N}^{\text{HFB}}(\omega^{\text{FFG}}_{N}, 0) - \tr\, \varepsilon^{2}\Delta \omega_{N} (1 - \omega_{N}) - CN^{2/3}\;;
\end{equation}
combined with assumption (\ref{eq:upperFFG}), we obtain:
\begin{equation}
- \tr\, \varepsilon^{2}\Delta \omega_{N} (1 - \omega_{N}) \leq K N^{2/3}\;.
\end{equation}
Recalling (\ref{eq:commkinest}), the second bound in (\ref{eq:alphaHS}) follows. This concludes the proof of Proposition~\ref{prp:pairbd}.
\end{proof}
The next proposition establishes the validity of the commutator estimates for $\omega_{N}$ in (\ref{eq:sc}), in Hilbert-Schmidt norm, for translation-invariant states energetically close to the free Fermi gas. Notice that for a translation invariant state the second estimate in (\ref{eq:sc}) is trivially true. The next proposition proves the first estimate in (\ref{eq:sc}).
\begin{proposition}[Bounds for $\omega_{N}$]\label{prp:TIomega} Let $V$ be bounded. Suppose that $(\omega_{N}, \alpha_{N})$ is a pure, quasi-free and translation invariant state, and that:
\begin{equation}\label{eq:enconst2}
\mathcal{E}_{N}^{\text{HFB}}(\omega_{N}, \alpha_{N}) \leq \mathcal{E}_{N}^{\text{HFB}}(\omega^{\text{FFG}}_{N}, 0) + CN^{1/3}\;.
\end{equation}
Then, there exists $C>0$ such that:
\begin{equation}\label{eq:scHF}
\sup_{p\in \mathbb{Z}^3} \frac{1}{1+|p|} \| [ \omega_{N}, e^{ip\cdot \hat x} ] \|_{\text{\text{HS}}}^{2} \leq C  N\varepsilon\;.
\end{equation}
\end{proposition}
\begin{proof} We start by observing:
\begin{equation}\label{eq:deltaest}
\begin{split}
\tr\, | [ e^{ip\cdot \hat x}, \omega_{N} ] |^{2} &= \tr ( \omega_{N}^{2} + e^{ip\cdot \hat x} \omega_{N}^{2} e^{-ip\cdot \hat x} - e^{ip\cdot \hat x} \omega_{N} e^{-ip\cdot \hat x} \omega_{N} - \omega_{N} e^{ip\cdot \hat x} \omega_{N} e^{-ip\cdot \hat x} ) \\
&\leq 2\tr\, \omega_{N} (1 - e^{ip\cdot \hat x} \omega_{N} e^{-ip\cdot \hat x})\;,
\end{split}
\end{equation}
where we used that $0\leq \omega_{N} \leq 1$, and the cyclicity of the trace. We will use the energetic constraint (\ref{eq:enconst2}) to show that:
\begin{equation}
\tr\, \omega_{N} (1 - e^{ip\cdot \hat x} \omega_{N} e^{-ip\cdot \hat x}) \leq CN\varepsilon (1 + |p|)\;,
\end{equation}
which combined with (\ref{eq:deltaest}) gives the final claim (for $p=0$, this has been obtained in the proof of Proposition \ref{prp:pairbd}).

Let us denote by $\hat \omega_{N,p}(k, \sigma)$ the Fourier symbol of $(e^{i p\cdot \hat x} \omega_{N} e^{-ip\cdot \hat x})(x,\sigma; y,\sigma)$, that is $\hat \omega_{N,p}(k, \sigma) = \hat \omega_{N}(k-p, \sigma)$. Since $|\varepsilon^{2} |k|^{2} - \mu_{\sigma}| = \varepsilon^{2} (| k | - k^{\sigma}_{F}) (|k| + k^{\sigma}_{F})$ and $0\leq \hat \omega_{N,p}(k) \leq 1$,  we have, from Eq. (\ref{eq:A10}):
\begin{equation}\label{eq:ffgcomm}
\begin{split}
\sum_{k, \sigma} (\varepsilon^{2} |k|^{2} - \mu_{\sigma}) (\hat \omega_{N}(k, \sigma) - \hat \omega^{\text{FFG}}_{N}(k, \sigma) ) &\geq C\varepsilon \sum_{\substack{ \sigma, k \in B^{\sigma}_{F} \\ k^{\sigma}_{F} - |k| \geq C}}  \hat \omega_{N,p}(k, \sigma) (1 - \hat \omega_{N}(k, \sigma))\\&\quad + C\varepsilon \sum_{ \substack{ \sigma, k \notin B_{F} \\ |k| - k^{\sigma}_{F} \geq C}} (1 - \hat \omega_{N,-p}(k, \sigma))\hat \omega_{N}(k, \sigma) \\
&\geq C\varepsilon \sum_{\sigma, k \in B^{\sigma}_{F}}  \hat \omega_{N,p}(k, \sigma) (1 - \hat \omega_{N}(k, \sigma))\\&\quad + C\varepsilon \sum_{ \sigma, k \notin B^{\sigma}_{F}} (1 - \hat \omega_{N,-p}(k, \sigma))\hat \omega_{N}(k, \sigma) - CN^{1/3}\;.
\end{split}
\end{equation}
This implies:
\begin{equation}
\begin{split}
&\sum_{\sigma, k} (\varepsilon^{2} |k|^{2} - \mu_{\sigma}) (\hat \omega_{N}(k, \sigma) - \hat \omega^{\text{FFG}}_{N}(k, \sigma) ) + CN^{1/3}\\&\qquad \geq C\varepsilon \tr\, \omega_{N}^{\text{FFG}} e^{ip\cdot \hat x} \omega_{N} e^{-ip\cdot \hat x} (1 - \omega_{N} ) +C\varepsilon \tr\, (1-\omega_{N}^{\text{FFG}}) \omega_{N} (1 - e^{-ip\cdot \hat x} \omega_{N} e^{ip\cdot \hat x} ) \\
&\qquad = C\varepsilon \tr\, \omega_{N}^{\text{FFG}} e^{ip\cdot \hat x} \omega_{N} e^{-ip\cdot \hat x} (1 - \omega_{N} ) + C\varepsilon \tr\, e^{ip\cdot \hat x} (1-\omega_{N}^{\text{FFG}}) \omega_{N} e^{-ip\cdot \hat x}  (1 - \omega_{N} ) \\
&\qquad = C\varepsilon \tr\,e^{ip\cdot \hat x} \omega_{N} e^{-ip\cdot \hat x} (1 - \omega_{N} ) - C\varepsilon \tr\, \big[ e^{ip\cdot \hat x}, \omega_{N}^{\text{FFG}}\big] \omega_{N} e^{-ip\cdot \hat x}  (1 - \omega_{N} )\;.
\end{split}
\end{equation}
Therefore, we obtain the following lower bound:
\begin{equation}\label{eq:lowkin}
\begin{split}
&\sum_{\sigma, k} (\varepsilon^{2} |k|^{2} - \mu_{\sigma}) (\hat \omega_{N}(k, \sigma) - \hat \omega^{\text{FFG}}_{N}(k, \sigma) ) +  CN^{1/3} \\& \qquad \geq C\varepsilon \tr\,e^{ip\cdot \hat x} \omega_{N} e^{-ip\cdot \hat x} (1 - \omega_{N} ) - C\varepsilon \big\| \big[ e^{ip \cdot \hat x}, \omega_{N}^{\text{FFG}}  \big] \big\|_{\text{tr}}\;.
\end{split}
\end{equation}
Using that the direct energy of $(\omega_{N}, \alpha_{N})$ and of $(\omega_{N}^{\text{FFG}}, 0)$ are equal, and the fact that the exchange energy and the pairing energy are bounded by a constant, Eq. (\ref{eq:lowkin}) gives:
\begin{equation}
\begin{split}
&\mathcal{E}^{\text{HFB}}_{N}(\omega_{N}, \alpha_{N}) + CN^{1/3} \\
&\quad \geq \mathcal{E}^{\text{HFB}}_{N}(\omega_{N}^{\text{FFG}}, 0) - C\|V\|_{\infty} + C\varepsilon\tr\,e^{ip\cdot \hat x} \omega_{N} e^{-ip\cdot \hat x} (1 - \omega_{N} ) - C\varepsilon \big\| \big[ e^{ip \cdot \hat x}, \omega_{N}^{\text{FFG}}  \big] \big\|_{\text{tr}}\;.
\end{split}
\end{equation}
Thus, by the assumption (\ref{eq:enconst2}), we get:
\begin{equation}
\begin{split}
\tr\,e^{ip\cdot \hat x} \omega_{N} e^{-ip\cdot \hat x} (1 - \omega_{N} ) &\leq \big\| \big[ e^{ip \cdot \hat x}, \omega_{N}^{\text{FFG}}  \big] \big\|_{\text{tr}} + C N^{2/3} \\
&\leq C N^{2/3}( 1+ \varepsilon |p|)\;,
\end{split}
\end{equation}
where the last inequality follows from an explicit computation of the trace norm, see {\it e.g.} \cite{BPS}. Combining this estimate with (\ref{eq:deltaest}), we immediately get:
\begin{equation}
\tr\, | [ e^{ip\cdot \hat x}, \omega_{N} ] |^{2} \leq K N^{2/3} (1 + |p|)\;,
\end{equation}
which concludes the proof of (\ref{eq:scHF}).
\end{proof}
\begin{remark} 
We can construct examples of translation invariant states $(\omega_{N}, \alpha_{N})$ with $\alpha_{N} \neq 0$ satisfying the assumption (\ref{eq:upperFFG}) of Propositions \ref{prp:pairbd}, \ref{prp:TIomega}, in the following way. For simplicity, suppose that $N$ is even and $N_{\uparrow} = N_{\downarrow} = N/2$. Let $\lambda(k)$ be the restriction to $\mathbb{Z}^{3}$ of an even, smooth, non-increasing function $0\leq \lambda(k) \leq 1$, such that:
\begin{equation}
\lambda(k) = \chi(k \in B_{F})\qquad \text{for $||k| - k_{F}| > C$}
\end{equation} 
and $\lambda(k)$ smoothly interpolates between $1$ and $0$ for $| |k| - k_{F} | \leq C$. We shall suppose that:
\begin{equation}
\sum_{k\in \mathbb{Z}^{3}} \lambda(k) = \frac{N}{2}\;.
\end{equation}
Then, we consider $(\omega_{N}, \alpha_{N})$ such that (recall the representation (\ref{eq:2x2})):
\begin{equation}
\omega_{N} = \begin{pmatrix} \omega_{N,\uparrow\uparrow} & 0 \\ 0 & \omega_{N,\downarrow\downarrow}  \end{pmatrix}\;,\qquad \alpha_{N} = \begin{pmatrix} 0 & \alpha_{N,\uparrow\downarrow} \\ \alpha_{N,\downarrow\uparrow} & 0 \end{pmatrix}
\end{equation}
with (using the convention $\mathrm{sign}(\uparrow)=+$, $\mathrm{sign}(\downarrow)=-$ and $\sigma' \neq \sigma$):
\begin{equation}
\omega_{N, \sigma \sigma} = \sum_{k\in \mathbb{Z}^{3}} \lambda(k) |f_{k} \rangle \langle f_{k}|\;,\qquad \alpha_{N,\sigma \sigma'} = \mathrm{sign}(\sigma)\sum_{k\in \mathbb{Z}^{3}} \sqrt{ \lambda(k) (1 - \lambda(k)) } | f_{k} \rangle \langle \overline{f_{k}} |\;.
\end{equation}
We observe that $0\leq \omega_{N} \leq 1$, $\tr\, \omega_{N} = N$, $\alpha_{N}^{T} = -\alpha_{N}$, and that:
\begin{equation}
\begin{pmatrix} \omega_{N} & \alpha_{N} \\ \alpha_{N}^{*} & 1 - \overline{\omega_{N}} \end{pmatrix}^{2} = \begin{pmatrix} \omega_{N} & \alpha_{N} \\ \alpha_{N}^{*} & 1 - \overline{\omega_{N}} \end{pmatrix}\;.
\end{equation}
Thus, $(\omega_{N}, \alpha_{N})$ parametrizes a pure, quasi-free state, with nonzero pairing. Let us estimate the HFB energy of such state. From (\ref{eq:HFBtrans}), we have:
\begin{equation}\label{eq:HFBup0}
\mathcal{E}_{N}^{\text{HFB}}(\omega_{N}, \alpha_{N}) \leq \sum_{\sigma}\sum_{k} \varepsilon^{2} |k|^{2} \hat \omega_{N}(k,\sigma) + \frac{N\hat V(0)}{2} + C\|V\|_{\infty}\;,
\end{equation}
where the last term is an upper bound for the exchange energy and for the pairing energy. Next, we estimate:
\begin{equation}\label{eq:A18}
\begin{split}
\sum_{\sigma}\sum_{k} \varepsilon^{2} |k|^{2} &\hat \omega_{N}(k,\sigma) = \sum_{\sigma}\sum_{k} \varepsilon^{2} |k|^{2} \hat \omega^{\text{FFG}}_{N}(k,\sigma) + \sum_{\sigma}\sum_{k} \varepsilon^{2} |k|^{2} (\hat \omega_{N}(k,\sigma) - \hat\omega_{N}^{\text{FFG}}(k, \sigma)) \\
&= \sum_{\sigma}\sum_{k} \varepsilon^{2} |k|^{2} \hat \omega^{\text{FFG}}_{N}(k,\sigma) + \sum_{\sigma}\sum_{k} (\varepsilon^{2} |k|^{2} - \mu) (\hat \omega_{N}(k,\sigma) - \hat\omega_{N}^{\text{FFG}}(k, \sigma))\;,
\end{split}
\end{equation}
where in the last step we used that $\omega_{N}$ and $\omega_{N}^{\text{FFG}}$ have the same number of particles with given spin. The first term in the right-hand side of (\ref{eq:A18}) is the kinetic energy of the free Fermi gas. The second term is bounded above as:
\begin{equation}\label{eq:upkino}
\begin{split}
\sum_{\sigma}\sum_{k} (\varepsilon^{2} |k|^{2} - \mu) &(\hat \omega_{N}(k,\sigma) - \hat\omega_{N}^{\text{FFG}}(k, \sigma)) \\& = \varepsilon^{2}\sum_{\sigma}\sum_{k} (|k| + k_{F}) (|k| - k_{F}) (\hat \omega_{N}(k,\sigma) - \hat\omega_{N}^{\text{FFG}}(k, \sigma)) \\
& = \varepsilon^{2}\sum_{\sigma}\sum_{k: ||k| - k_{F}| \leq C } (|k| + k_{F}) (|k| - k_{F}) (\hat \omega_{N}(k,\sigma) - \hat\omega_{N}^{\text{FFG}}(k, \sigma)) \\
&\leq K\varepsilon \sum_{k: ||k| - k_{F}| \leq C } 1 \leq \widetilde{K} N^{1/3}\;.
\end{split}
\end{equation}
Thus, recalling that the direct energy is the same for all translation-invariant states with equal number of particles with a given spin, Eqs. (\ref{eq:HFBup0}), (\ref{eq:upkino}) imply the bound
\begin{equation}
\mathcal{E}_{N}^{\text{HFB}}(\omega_{N}, \alpha_{N}) \leq \mathcal{E}_{N}^{\text{HFB}}(\omega^{\text{FFG}}_{N}, 0) + \widetilde{K} N^{1/3} + 2\|V\|_{\infty}\;,
\end{equation}
where the factor $2$ takes into account the bound for the exchange energy of the free Fermi gas. This concludes the check of (\ref{eq:upperFFG}).
\end{remark}
Finally, we conclude the section by showing that the only translation invariant ground state of the HFB functional on the torus is the free Fermi gas.
\begin{proposition}[Translation invariant ground state]\label{prp:FFGgs} Let $V$ be bounded. Suppose that $(\omega_{N}, \alpha_{N})$ is a pure, quasi-free and translation invariant state. Then, for $N$ large enough:
\begin{equation}
\mathcal{E}^{\text{HFB}}_{N}(\omega_{N}, \alpha_{N}) \geq \mathcal{E}^{\text{HFB}}_{N}(\omega_{N}^{\text{FFG}}, 0)\;.
\end{equation}
Equality holds if and only if $(\omega_{N}, \alpha_{N}) = (\omega_{N}^{\text{FFG}}, 0)$.
\end{proposition}
\begin{proof} We introduce the notation:
\begin{equation}
\text{Ex}(\omega_{N}) := -\frac{1}{2N} \int d{\bf x} d{\bf y}\, V( x- y) | \omega_{N}({\bf x};{\bf y}) |^{2}\;.%\qquad X_{\omega}({\bf x};{\bf y}) = \frac{1}{N} V(x-y) \omega_{N}({\bf x};{\bf y})\;.
\end{equation}
Then, is it not difficult to see that:
\begin{equation}\label{eq:2obt}
\begin{split}
\mathcal{E}^{\text{HFB}}_{N}(\omega_{N}, \alpha_{N}) - \mathcal{E}^{\text{HFB}}_{N}(\omega_{N}^{\text{FFG}}, 0) &= \tr\, h_{\omega} ( \omega_{N} - \omega_{N}^{\text{FFG}} ) + \text{Ex}(\omega_{N} - \omega_{N}^{\text{FFG}}) - \text{Ex}(\alpha_{N}) \\
&= \tr\, (h_{\omega} - \mu) ( \omega_{N} - \omega_{N}^{\text{FFG}} ) + \text{Ex}(\omega_{N} - \omega_{N}^{\text{FFG}}) - \text{Ex}(\alpha_{N})
\end{split}
\end{equation}
where $h_{\omega} = -\varepsilon^{2}\Delta - X_{\omega}$, with $X_{\omega}({\bf x};{\bf y}) = \frac{1}{N} V(x-y) \omega_{N}({\bf x};{\bf y})$, and where we recall that $\mu$ has to be thought of as a diagonal matrix in spin state: $(\mu f)(x, \sigma) = \mu_{\sigma} f(x,\sigma)$. We will use the first term in the right-hand side to control the other two. We write:
\begin{equation}\label{eq:lower}
\begin{split}
& \tr\, (h_{\omega} - \mu) ( \omega_{N} - \omega_{N}^{\text{FFG}} ) \\& \quad =  \tr\, (h_{\omega} - \mu) \mathbbm{1}( h_{\omega} - \mu \geq 0 )( \omega_{N} - \omega_{N}^{\text{FFG}} ) \mathbbm{1}( h_{\omega} - \mu \geq 0 ) \\&\qquad  + \tr\, (h_{\omega} - \mu) \mathbbm{1}( h_{\omega} - \mu < 0 )( \omega_{N} - \omega_{N}^{\text{FFG}} ) \mathbbm{1}( h_{\omega} - \mu < 0 )\;.
 \end{split}
\end{equation}
Since by assumption $\text{dist}(\sigma(-\Delta), k_{F}^{\sigma 2}) = 1/2$, and using the bound $\| X_{\omega} \| \leq C/N$ together with the fact that $h_\omega$ is invariant by translation and thus commutes with $-\Delta$, we have for $N$ large enough:
\begin{equation}
\omega_{N}^{\text{FFG}} = \mathbbm{1}( -\varepsilon^{2} \Delta - \mu\leq 0) \equiv  \mathbbm{1}( h_{\omega} - \mu \leq 0 )\;.
\end{equation}
Thus, we can rewrite (\ref{eq:lower}) as:
\begin{equation}
\begin{split}
& \tr\, (h_{\omega} - \mu) ( \omega_{N} - \omega_{N}^{\text{FFG}} ) \\&\quad = \tr\, | h_{\omega} - \mu| (1 - \omega^{\text{FFG}}_{N}) \omega_{N} (1 - \omega^{\text{FFG}}_{N}) + \tr\, | h_{\omega} - \mu| \omega^{\text{FFG}}_{N} ( 1-  \omega_{N}) \omega^{\text{FFG}}_{N} \\
 &\quad =\tr\, | h_{\omega} - \mu| \Big( ( \omega_{N} - \omega_{N}^{\text{FFG}} )^{2} + \omega_{N} (1 - \omega_{N}) \Big)\;.
 \end{split}
\end{equation}
Since:
\begin{equation}
\begin{split}
\Big| \text{Ex}( \omega_{N} - \omega_{N}^{\text{FFG}} ) \Big| &\leq \frac{\|V\|_{\infty}}{2N} \| \omega_{N} - \omega_{N}^{\text{FFG}} \|_{\text{HS}}^{2} \\\Big| \text{Ex}( \alpha_{N}  ) \Big| &\leq \frac{\|V\|_{\infty}}{2N} \| \alpha_{N} \|_{\text{HS}}^{2} = \frac{\|V\|_{\infty}}{2N} \tr\, \omega_{N} (1 - \omega_{N})\;,
\end{split}
\end{equation}
we obtain:
\begin{equation}
\mathcal{E}^{\text{HFB}}_{N}(\omega_{N}, \alpha_{N}) - \mathcal{E}^{\text{HFB}}_{N}(\omega_{N}^{\text{FFG}}, 0) \geq \tr\, \Big( | h_{\omega} - \mu| - \frac{\|V\|_{\infty}}{N}\Big) \Big( ( \omega_{N} - \omega_{N}^{\text{FFG}} )^{2} + \omega_{N} (1 - \omega_{N}) \Big)\;.
\end{equation}
Next, using that $ | h_{\omega} - \mu| \geq C\varepsilon^{2}$,  we finally get for $N$ large enough:
\begin{equation}
\mathcal{E}^{\text{HFB}}_{N}(\omega_{N}, \alpha_{N}) - \mathcal{E}^{\text{HFB}}_{N}(\omega_{N}^{\text{FFG}}, 0) \geq K \varepsilon^{2}\tr\, \Big( ( \omega_{N} - \omega_{N}^{\text{FFG}} )^{2} + \omega_{N} (1 - \omega_{N}) \Big) \geq 0\;.
\end{equation}
The right-hand side is zero if and only if $\omega_{N} = \omega_{N}^{\text{FFG}}$ (and hence $\alpha_{N} = 0$). This concludes the proof of Proposition \ref{prp:FFGgs}.
\end{proof}

\end{document}